\documentclass{siamltex}

\usepackage{amsfonts}
\usepackage{float}
\usepackage{graphicx}
\usepackage{amstext}
\usepackage{amsmath}
\usepackage{amssymb}

\title{Large deviations for a mean field model of systemic risk}

\author{
Josselin Garnier\thanks{Laboratoire de Probabilit\'{e}s et Mod\`{e}les 
Al\'{e}atoires \& Laboratoire Jacques-Louis Lions, Universit\'{e} Paris VII 
(\tt{garnier@math.jussieu.fr})} 
\and George Papanicolaou\thanks{Mathematics Department, Stanford University 
(\tt{papanicolaou@stanford.edu})}
\and Tzu-Wei Yang\thanks{Institute for Computational and Mathematical Engineering 
(ICME), Stanford University (\tt{twyang@stanford.edu})}}

\begin{document}

\maketitle

\begin{abstract}
We consider a system of diffusion processes that interact through their empirical mean and have a 
stabilizing force acting on each of them, corresponding to a bistable potential. There are three parameters 
that characterize the system: the strength of the intrinsic stabilization, the strength of the external 
random perturbations, and the degree of cooperation or interaction between them. The latter is the rate of 
mean reversion of each component to the empirical mean of the system. We interpret this model in the context 
of systemic risk and analyze in detail the effect of cooperation between the components, that is, the rate 
of mean reversion. We show that in a certain regime of parameters increasing cooperation tends to increase 
the stability of the individual agents but it also increases the overall or systemic risk. We use the theory 
of large deviations of diffusions interacting through their mean field.
\end{abstract}



\pagestyle{myheadings}

\thispagestyle{plain}

\section{Introduction}

Systemic risk is the risk that in an interconnected system of agents that can fail individually, a large 
number of them fails simultaneously or nearly so, leading to the overall failure of the system. It is a 
property of the interconnected system as a whole, and not only of the individual components, in the sense 
that assessment of the risk of individual failure alone cannot provide an assessment of the systemic risk. 
The interconnectivity of the agents, its form and evolution, play an essential role in systemic risk 
assessment \cite{Bisias2012}.

In this paper we consider a simple model of interacting agents for which systemic risk can be assessed 
analytically in some interesting cases. Each agent can be in one of two states, a normal and a failed one, 
and it can undergo transitions between them. We assume that the dynamic evolution of each agent has the 
following features. First, there is an intrinsic stabilization mechanism that tends to keep the agents near 
the normal state. Second, there are external destabilizing forces that tend to push away from the normal 
state and are modeled by stochastic processes. Third, there is cooperation among the agents that acts as 
individual stabilizer by diversification. This means that in such a system we expect that there is a 
decrease in the risk of destabilization or "failure" for each agent because of the cooperation or 
diversification. What is less obvious is the effect of cooperation on the overall or system's risk, which 
can be defined in a precise way for the model considered here. We show in this paper that for the models 
under consideration and in a certain regime of parameters, the systemic risk increases with increasing 
cooperation. The aim of this paper is to analyze this tradeoff between individual risk and systemic risk for 
a class of interacting systems subject to failure.

Perhaps a simple mathematical model of interacting agents having the features we want is a system of 
stochastic differential equations with mean-field interaction. Let $x_j(t)$ be the state of risk of agent or 
component $j$, taking real values. For $j=1,\ldots,N$, the $x_j(t)$'s are modeled as continuous-time 
stochastic processes satisfying the system of It\^{o} stochastic differential equations:
\begin{equation}
	\label{eq:SDE of single component}
	dx_j(t) = -hU(x_j(t))dt + \theta(\bar{x}(t)-x_j(t))dt + \sigma dw_j(t),
\end{equation}
with given initial conditions. Here $-h U(y)=-hV'(y)$ is the restoring force, $V$ is a potential  which we 
assume has two stable states, and $\{w_j(t)\}_{j=1}^N$ are independent, standard Brownian motions. The 
parameter $h$ controls the level of intrinsic stabilization and $\sigma$ is the strength of the destabilizing 
random forces. The interaction or cooperation is the mean reversion term with rate of mean reversion  
$\theta$ and with $\bar{x}(t):=\frac{1}{N}\sum_{i=1}^N x_i(t)$ denoting the empirical mean of the processes, 
that is, the empirical mean of the individual risks. For $\theta>0$ the individual risk processes tend to 
mean-revert to their empirical mean, which is a simple but non-trivial form of cooperation. We take the 
empirical mean $\bar{x}(t)$ to be a measure of the systemic risk. The bi-stable-state structure of $V(y)$ 
determines the normal and failed states of the agents. We will assume in this paper that  $U(y)=y^3-y$, so that 
$V(y)= \frac{1}{4} y^4 - \frac{1}{2}y^2 +c$ and we take $c=0$ since it is inessential. 
The two stable states are then $\pm 1$ and we let $-1$ be the normal state and 
$+1$ to be the failed state. The potential $V(y)$ ensures that each risk variable $x_j(t)$ stays around $-1$ 
(normal) or $+1$ (failed). The evolution of the system is characterized by the initial conditions, the three 
parameters ($h$, $\theta$, $\sigma$) and by the system size $N$.



We have chosen a mean-field interaction because it is a simple form of cooperative behavior. More elaborate 
models are considered in Section \ref{sec:diversity}, where some heterogeneity is introduced between the 
components of the system. For mean-field models a natural measure of systemic risk is the transition 
probability of the empirical mean $\bar{x}(t)$ from the normal state to the failed state. More precisely, 
the mathematical problem we address here is this: For $N$ large we calculate approximately such transition 
probabilities and analyze how they depend on $h,\sigma$ and $\theta$, the three parameters of the system. We 
are interested in a regime of these parameters for which there are two collective, that is, large $N$, 
equilibria centered around the normal and failed states. These two equilibria can be identified through the 
mean-field limit of the system, that is, the weak limit in probability of the empirical density of the 
agents risk $x_j$. Mean field models with multiple stable points, not only bistable ones, could be 
considered but their analysis is more involved while the main result about systemic risk, and dependence on 
the parameters  ($h$, $\theta$, $\sigma$) and by the system size $N$, is clearly seen in the bistable model
that we consider here.

The mathematical analysis of bistable mean field models like (\ref{eq:SDE of single component}) was 
initiated by Dawson \cite{Dawson1983, Gartner1988}, including the mean field limit, the existence of multiple equilibria, 
and a fluctuation theory. Non-equilibrium statistical mechanics and phase transitions have been studied 
extensively in the sciences \cite{Haken1983}. The large deviation theory that we use here was developed by 
Dawson and G\"artner \cite{Dawson1987,Dawson1989}. In particular, they introduced and analyzed the rate function for 
large deviations associated with (\ref{eq:SDE of single component}) when $N$ is large and with
more general potentials \cite{Dawson1989}. Their theory may be 
considered as an infinite dimensional extension of the Freidlin-Wentzell theory of large deviations for 
stochastic differential equations with small noise \cite{Freidlin1998, Dembo2010}. The main result in this 
paper is the analysis of this rate function for small $h$. That is, for a shallow two-well potential, where 
transitions from one well (quasi-equilibrium) to the other are exponentially small in $N$, the "constant" in 
the exponent is small when $h$ is small. Other mean field models have been studied in  \cite{Tanaka1984, 
Gartner1988, Meleard1996, BenArous1999, DelMoral2005, DelMoral2011, Fouque2012}, and large deviations 
results for various models can be found in \cite{Dawson1994, Arous1995, DelMoral1998, Dawson2005, 
Herrmann2008, Budhiraja2008, Budhiraja2012}. In \cite{Budhiraja2012} a general large deviations theory is 
developed for a model with both drift and volatility interactions, as well as with degenerate noise, using 
weak convergence and optimal control methods.

The main contribution of the paper as far as systemic risk theory is concerned is the demonstration that, 
within the range of the bistable mean field model (\ref{eq:SDE of single component}), while cooperation 
between agents decreases the individual risk of each agent, the systemic or overall risk is increased. This 
is discussed in detail in Section \ref{sec:comparison}, in terms of the three parameters 
$(h,\theta,\sigma)$, with $h$ small. The fact that reducing individual risk by cooperation or 
diversification can lead to increased systemic risk has been anticipated in macroeconomics and elsewhere and 
it has been extensively discussed, modeled, and analyzed in \cite{Nier2007, Battiston2009, Haldane2009, 
Gai2010, May2010, Stiglitz2010, Beale2011, Battiston2011, Haldane2011, Ibragimov2011}. However, the dynamic 
phase transitions formulation and the large deviations theory exploited in this paper have not been used in 
the economics literature, to our knowledge. The use of coupled stochastic equations for modeling evolution 
of individual risk and the effects of interactions among agents is also considered in \cite{Battiston2009, 
Hull2001} where there is some discussion regarding the economic interpretation of the variables 
$\{x_j(t)\}$. They could, for example, represent some form of equity ratio in a very simple model in 
insurance or banking. 

The paper is organized as follows. In Section \ref{sec:mean field}, we briefly review the classical 
mean-field limit in \cite{Dawson1983}, and we discuss the intrinsic stability of equilibria 
\cite{Dawson1983} when $h$ is small. Section \ref{sec:diversity} generalizes 
(\ref{eq:SDE of single component}) by replacing the rate of mean reversion $\theta$ by an agent-dependent 
$\theta_j$. The mean-field limit and the explicit conditions are also studied. In Section 
\ref{sec:simulation}, we carry out numerical simulations of both the homogeneous and the heterogeneous model 
in various parameter ranges. 
Section \ref{sec:large deviations} uses the large deviation principle in \cite{Dawson1987} to formulate the 
dynamic phase transition of interest here, that is, the system transition from the normal state to the 
failed state. In Section \ref{sec:small h}, we specialize the large deviations theory when $h$ is small so  
as to obtain a result from which the systemic risk as a function the basic parameters $(h,\theta,\sigma)$ 
can be assessed and interpreted. In Section \ref{sec:reduced LDP} we introduce a formal expansion of the 
rate function for small $h$ and obtain a reduced variational principle for the systemic risk that appears to 
come from a large deviations principle for an one-dimensional dynamical system. It gives, of course, the 
same results about systemic risk as described in Section \ref{sec:small h}. In Section 
\ref{sec:diversity2} we discuss the case where there is diversity in mean reversion and it 
is shown that under some natural conditions the heterogeneous model is systemically more unstable than the  
homogeneous one. The technical details of the proofs are in the appendices.
\section{The Mean-Field Limit}
\label{sec:mean field}

We briefly review the mean field limit in \cite{Dawson1983,Gartner1988} and carry out a small $h$ analysis 
of results since they will be used in calculating large deviation probabilities. We want to analyze the 
systemic behavior  of the interacting diffusion processes (\ref{eq:SDE of single component}), through their 
empirical mean $\bar{x}(t)$, but this is not possible in a direct way since 
(\ref{eq:SDE of single component}) is nonlinear. We consider instead the empirical density of $x_j(t)$, 
which is a measure valued process that has a limit as $N\rightarrow\infty$. Let $M_1(\mathbb{R})$ be the 
space of probability measures endowed with the weak (Prohorov) topology and let $C([0,T],M_1(\mathbb{R}))$ 
be the space of continuous $M_1(\mathbb{R})$-valued processes on $[0,T]$ endowed with the corresponding weak 
topology. Define the empirical probability measure process 
$X_N(t,dy):=\frac{1}{N}\sum_{j=1}^N\delta_{x_j(t)}(dy)$ and note that $X_N\in C([0,T],M_1(\mathbb{R}))$. The 
mean field limit theorem for $X_N$, proved in \cite{Dawson1983,Gartner1988}, is as follows:
\begin{theorem}
	\label{thm: Dawson, mean-field limit}
	(Dawson, 1983) Assume that the force is $U(y)=y^3-y$ and that $X_N(0)$ converges weakly to a probability
	measure $\nu_0$. Then the measure valued process $X_N$ converges weakly in law as $N\rightarrow\infty$ to a 
	deterministic process with density $u(t,y)dy \in C([0,T],M_1(\mathbb{R}))$, which is the unique weak solution of the Fokker-Planck 
	equation:
	\begin{equation}
		\label{eq: full Fokker-Planck eqn}
		\frac{\partial}{\partial t}u
		= h\frac{\partial}{\partial y}[U(y)u]
		- \theta\frac{\partial}{\partial y}\left\{\left[\int yu(t,y)dy-y\right]u\right\} 
		+ \frac{1}{2}\sigma^2\frac{\partial^2}{\partial y^2}u,
	\end{equation}
	with initial condition $\nu_0$.
\end{theorem}

By Theorem \ref{thm: Dawson, mean-field limit}, we can analyze $u$ and view $X_N$ as a perturbation of $u$ 
for $N$ large. We may consider $\bar{x}(t)$ in the same way because $\bar{x}(t)=\int y X_N(t,dy)$. However, 
the limit problem is infinitely dimensional, as is expected.

Explicit solutions of (\ref{eq: full Fokker-Planck eqn}) are not available in general, but we can find 
equilibrium solutions. Assuming that $\xi = \lim_{t\to\infty}\int yu(t,y)dy$, then an equilibrium solution 
$u_\xi^e$ satisfies
\[
	h\frac{d}{dy}[(y^3-y)u_\xi^e] - \theta\frac{d}{dy}[(\xi-y)u_\xi^e] 
	+ \frac{1}{2}\sigma^2\frac{d^2}{dy^2}u_\xi^e = 0,
\]
and has the form
\begin{equation}
	\label{eq:equilibrium of Fokker-Planck eq}
	u^e_\xi(y) = \frac{1}{Z_\xi \sqrt{2\pi \frac{\sigma^2}{2\theta}}}
	\exp \left\{ -\frac{(y-\xi)^2}{2\frac{\sigma^2}{2\theta}} 
	- h\frac{2}{\sigma^2}V(y) \right\},
\end{equation}
with $Z_\xi$ the normalization constant:
\[
	Z_\xi = \int \frac{1}{\sqrt{2\pi \frac{\sigma^2}{2\theta}}}
	\exp \left\{ -\frac{(y-\xi)^2}{2\frac{\sigma^2}{2\theta}} 
	- h\frac{2}{\sigma^2}V(y) \right\} dy.
\]
Now $\xi$ must satisfy the compatibility or consistency condition:
\begin{equation}
	\label{eq:xi equals m(xi)}
	\xi = m(\xi):= \int y u^e_\xi(y)dy.
\end{equation}
Finding equilibrium solutions has thus been reduced to finding solutions of this equation.

For $U(y)=y^3-y$, $\xi=0$ is a solution for (\ref{eq:xi equals m(xi)}). With the same 
$U(y)$, it can be shown (see also \cite[Theorem 3.3.1 and 3.3.2]{Dawson1983}) that there are two additional 
non-zero solutions $\pm\xi_b$ if and only if $\frac{d}{d\xi}m(0)>1$, and for given $h$ 
and $\theta$, there exists a critical $\sigma_c(h,\theta) >0$ such that $\frac{d}{d\xi}m(0)>1$ if 
and only if $\sigma<\sigma_c(h,\theta)$.

An explanation for this bifurcation at equilibrium is that when $\sigma\geq\sigma_c$, randomness dominates the 
interaction among the components, i.e., $\theta(\bar{x}(t)-x_j(t))dt$ is negligible. In 
this case, the system behaves like $N$ independent diffusions and hence, by the
symmetry of $V(y)$, at any given time roughly half of them stay around $-1$ and half around $+1$ so the 
average is $0$. When, however, $\sigma<\sigma_c$, then the interactive force is 
significantly larger (now $\sigma dw_j(t)$ is less important). Therefore all agents stay 
around the same place (either $-\xi_b$ or $+\xi_b$) and the zero average equilibrium is unstable. 
Since we want to model systemic risk phenomena, 
we assume that 
$\sigma<\sigma_c$ throughout this paper, and we regard $-\xi_{b}$ as the normal state 
of the system and $+\xi_b$ as the failed state. The calculation of
transitions probabilities  between these two states is our objective.

For small $h$ we can approximate the solution 
of (\ref{eq:xi equals m(xi)}) to order $O(h)$ as follows.
\begin{proposition}
	\label{prop:explicit condition and solution for xi equal to m(xi) for small h}
	For small $h$, the critical value $\sigma_c$ can be expanded as 
	\begin{equation}
		\label{eq:sigma_c of the homogeneous case}
		\sigma_c = \sqrt{\frac{2 \theta}{3}} + O(h).
	\end{equation}
	In addition, the non-zero solutions $\pm \xi_{b}$ are 
	\begin{equation}
	\label{eq:explicit solution for xi equal to m(xi)}
	\pm \xi_b = \pm \sqrt{1 - 3\frac{\sigma^2}{2\theta}} 
	\left( 1 + h\frac{6}{\sigma^2}\left(\frac{\sigma^2}{2\theta}\right)^2
	\frac{1 - 2(\sigma^2/2\theta)}{1 - 3(\sigma^2/2\theta)}\right) + O(h^2).
	\end{equation}
\end{proposition}
\begin{proof}
	See Appendix 
	\ref{pf:explicit condition and solution for xi equal to m(xi) for small h}.
\end{proof}

From Proposition 
\ref{prop:explicit condition and solution for xi equal to m(xi) for small h}, we see the 
relation between the existence of the bi-stable states and the ratio 
$\sigma^2/2\theta$: For a given $\theta$, and for small $h$, 
(\ref{eq:xi equals m(xi)}) has non-zero solutions if and only if 
$ 3\sigma^2/2\theta < 1$. Moreover, these non-zero solutions $\pm \xi_b$ are 
generally not $\pm 1$ since the magnitude $|\xi_b|$ is less than $1$. 
Note that the coefficient of order $h$ in the expansion 
(\ref{eq:explicit solution for xi equal to m(xi)}) depends significantly on $\theta$ and 
$\sigma$. Thus, when $3\sigma^2/2\theta$ tends to $1$, $\xi_b$ in 
(\ref{eq:explicit solution for xi equal to m(xi)}) will not go to $+\infty$ while, in fact,
$\xi_b$ goes to $0$. From the $O(1)$ term in 
(\ref{eq:explicit solution for xi equal to m(xi)}), we also see that $\xi_b$ is roughly 
decreasing as $\sigma^2/2\theta$ is increasing.
\section{Diversity of Sensitivities}
\label{sec:diversity}

We can generalize (\ref{eq:SDE of single component}) by allowing for agent
dependent coefficients. We consider a particular case in which 
each agent can have a different rate of mean reversion to the empirical mean, that is, for 
$j=1,\ldots,N$, 
\begin{equation}
	\label{eq:SDE of single component, heterogeneous sensitivity}
	dx_j = -h\frac{\partial}{\partial x_j}V(x_j)dt 
	+ \sigma dw_j + \theta_j(\bar{x}-x_j)dt,
\end{equation}
and as before $V(y)=\frac{1}{4}y^4 - \frac{1}{2} y^2$.
We consider the case where $\theta_1,\ldots,\theta_N$ take $K$ distinct positive numbers, 
$\Theta_1,\ldots,\Theta_K$. We define $\mathcal{I}_l=\{j:\theta_j=\Theta_l\}$, 
$\rho_l=|\mathcal{I}_l|/N$ and 
$X_N^l=\frac{1}{\rho_l N}\sum_{j\in\mathcal{I}_l}\delta_{x_j}$. Assuming that 
$\lim_{N\rightarrow\infty}\rho_l$ exists and is positive for all $l$, the limit of 
$(X_N^1,\ldots,X_N^K)$ as $N\rightarrow\infty$ are the weak solutions $(u_1,\ldots,u_K)$ of 
the set of $K$ coupled Fokker-Planck equations.
\begin{theorem}
	\label{thm:mean-field limit, diversity case}
	Assume that $U(y)=y^3-y$ and that $(X_N^1(0), \ldots, X_N^K(0))$ converge weakly in probability 
	to the probability measures $(\nu^1,\ldots, \nu^K)$. Then the measure valued vector process $(X_N^1, \ldots, X_N^K)$ converges weakly as 
	$N\rightarrow\infty$ to the weak solution $(u_1, \ldots, u_K)$ of the system of the Fokker-Planck equations:
	\begin{align}
		\label{eq:system of Fokker-Planck equations}
		\frac{\partial}{\partial t}u_1 
		&= \frac{1}{2}\sigma^2\frac{\partial^2}{\partial y^2}u_1 
		- \Theta_1\frac{\partial}{\partial y}\left\{\left[
		\int y\sum_{l=1}^K\rho_lu_l(t,y)dy-y\right]u_1\right\} 
		+ h\frac{\partial}{\partial y}[U(y)u_1]\\
		&\vdots \notag\\
		\frac{\partial}{\partial t}u_{K} 
		&= \frac{1}{2}\sigma^2\frac{\partial^2}{\partial y^2}u_K 
		- \Theta_K\frac{\partial}{\partial y}\left\{\left[
		\int y\sum_{l=1}^K\rho_lu_l(t,y)dy-y\right]u_K\right\} 
		+ h\frac{\partial}{\partial y}[U(y)u_K], \notag
	\end{align}
	with initial condition $(\nu^1,\ldots, \nu^K)$.
\end{theorem}
\begin{proof}
	See Appendix \ref{pf:mean-field limit, diversity case} for the outline of the proof
	following \cite{Gartner1988}.
\end{proof}

The equilibrium solutions $\{u_{l,\xi}^e\}_{l=1}^K$ have the form
\begin{align}
	\label{eq:equilibrium of Fokker-Planck eqn, heterogeneous case}
	u_{l,\xi}^e(y) &= \frac{1}{Z_{l,\xi}\sqrt{2\pi\frac{\sigma^2}{2\Theta_l}}}
	\exp\left\{ -\frac{(y-\xi)^2}{2\frac{\sigma^2}{2\Theta_l}} 
	- h\frac{2}{\sigma^2} V(y) \right\}\\
	Z_{l,\xi} &= \int \frac{1}{\sqrt{2\pi\frac{\sigma^2}{2\Theta_l}}}
	\exp\left\{ -\frac{(y-\xi)^2}{2\frac{\sigma^2}{2\Theta_l}} 
	- h\frac{2}{\sigma^2} V(y) \right\} dy, \notag
\end{align}
and $\xi$ must satisfy the compatibility condition
\begin{equation}
	\label{eq:xi equal to m(xi), diversity case}
	\xi = m(\xi) := \sum_{l=1}^K \rho_l \int yu_{l,\xi}^e (y)dy.
\end{equation}
For $U(y)=y^3-y$, $\xi=0$ is the trivial solution of 
(\ref{eq:xi equal to m(xi), diversity case}), and a simple extension of Theorem 
3.3.1 in \cite{Dawson1983}, shows that there are two sets of non-trivial solutions 
$\{u_{l,\xi_b}^e\}_{l=1}^K$ and $\{u_{l,-\xi_b}^e\}_{l=1}^K$ if and only if 
$\frac{d}{d\xi}m(0)>1$. The numerical simulations presented in the next section show 
that diversity in the rate of mean reversion can have significant impact on the stability of the 
mean-field model.

As in the homogeneous case, we can get an approximate condition for 
equilibrium bifurcation for small $h$.
\begin{proposition}
	\label{prop:explicit condition for xi equal to m(xi) for small h, diversity case}
	The compatibility condition (\ref{eq:xi equal to m(xi), diversity case}) has 
	non-zero solutions if and only if $\sigma < \sigma_c^{\text{div}}$.
	For small $h$,  $\sigma_c^{\text{div}}$ has the expansion 
	\[ 
		\sigma_c^{\text{div}} = \sqrt{ \sum_{l=1}^K \frac{\rho_l}{\Theta_l}
		/ \sum_{l=1}^K\frac{3\rho_l}{2\Theta_l^2} } + O(h).
	\]
\end{proposition}
\begin{proof}
	See Appendix 
	\ref{pf:explicit condition for xi equal to m(xi) for small h, diversity case}.
\end{proof}

We note that diversity does affect the threshold condition and makes the analysis 
more difficult. The non-zero solutions $\pm \xi_b$ can be computed approximately when $h$ is small:
\begin{equation} 
	\label{eq:explicit solution for xi equal to m(xi), diversity case}
	\pm \xi_b = \pm \sqrt{ \sum_{l=1}^K \frac{\rho_l}{\Theta_l} 
	\left( 1 - 3\frac{\sigma^2}{2\Theta_l} \right) 
	/ \sum_{l=1}^K \frac{\rho_l}{\Theta_l} } + O(h).
\end{equation}
Higher order terms in the expansion of 
(\ref{eq:explicit solution for xi equal to m(xi), diversity case}) can also be obtained 
but we will omit them in this paper.
In the following Proposition we show that $\sigma_c^{\text{div}} \leq \sigma_c^{\text{homo}}$, 
where $\sigma_c^{\text{homo}}=\sigma_c$, the critical value (\ref{eq:sigma_c of the homogeneous case}) of 
the homogeneous case.
\begin{proposition}
	\label{prop:sigma_c comparison}
	With $\theta = \sum_{l=1}^K \rho_l \Theta_l$, we have 
	$\sigma_c^{\text{homo}} \geq \sigma_c^{\text{div}}$ for small $h$.
\end{proposition}
\begin{proof}
	See Appendix 
	\ref{pf:sigma_c comparison}.
\end{proof}

This result shows that when there is diversity the parameter region of existence of equilibria $\pm \xi_b$ 
is smaller than in the homogeneous case . From this observation we can anticipate that these equilibria are 
less stable in the presence of diversity, and this is confirmed next by numerical simulations and 
analytically.

By noting that $\xi_b^\text{homo}=\sqrt{1-(\sigma^2/\sigma_c^\text{homo})^2}+O(h)$ and 
$\xi_b^\text{div}=\sqrt{1-(\sigma^2/\sigma_c^\text{div})^2}+O(h)$, we have the following corollary:
\begin{corollary}
	With $\theta = \sum_{l=1}^K \rho_l \Theta_l$, we have $1>\xi_b^\text{homo} \geq \xi_b^\text{div}$ for 
	small $h$.
\end{corollary}
\section{Numerical Simulations}
\label{sec:simulation}

Before going into a detailed analysis of the models, we carry out
numerical simulations of 
(\ref{eq:SDE of single component}) and 
(\ref{eq:SDE of single component, heterogeneous sensitivity}) so as to get 
a quick impression of their behavior. We discretize with 
a uniform time grid, and let $X_j^n$ denote the simulated $X_j$ at time $n \Delta t$.

\subsection{Homogeneous Model}
\label{subsec:numhomo}

We simulate (\ref{eq:SDE of single component}) using the Euler scheme
\begin{equation}
	\label{eq:Euler method, homogeneous case}
	X_j^{n+1} = X_j^n - h U(X_j^n)\Delta t + \sigma\Delta W_j^{n+1} 
	+ \theta(\frac{1}{N}\sum_{k=1}^N X_k^n -X_j^n)\Delta t.
\end{equation}
We take $U(y)=y^3-y$, $=1$, $X_j^0=-1$, $\Delta t=0.02$, and let
$\{\Delta W_j^n\}_{j,n}$ be independent Gaussian random variables with mean zero and 
variance $\Delta t$. In the figures presented, the dashed lines show the 
numerical solutions of the compatibility equation (\ref{eq:xi equals m(xi)}),
$\xi=m(\xi)$. 
As noted earlier, if $\frac{d}{d\xi}m(0)\leq 1$, then
$0=m(0)$ is the unique solution and $0$ is a stable state. Therefore we should observe 
that the systemic risk fluctuates around $0$. If $\frac{d}{d\xi}m(0)>1$, there are two 
additional non-zero solutions $\pm\xi_b=m(\pm\xi_b)$ and $\pm\xi_b$ are 
stable while $0$ is unstable. We also know that when $h$ is small, the condition 
$\frac{d}{d\xi}m(0)>1$ can be simplified to be $3\sigma^2/2\theta<1$.

\begin{figure}
	\includegraphics[width=0.32\textwidth]{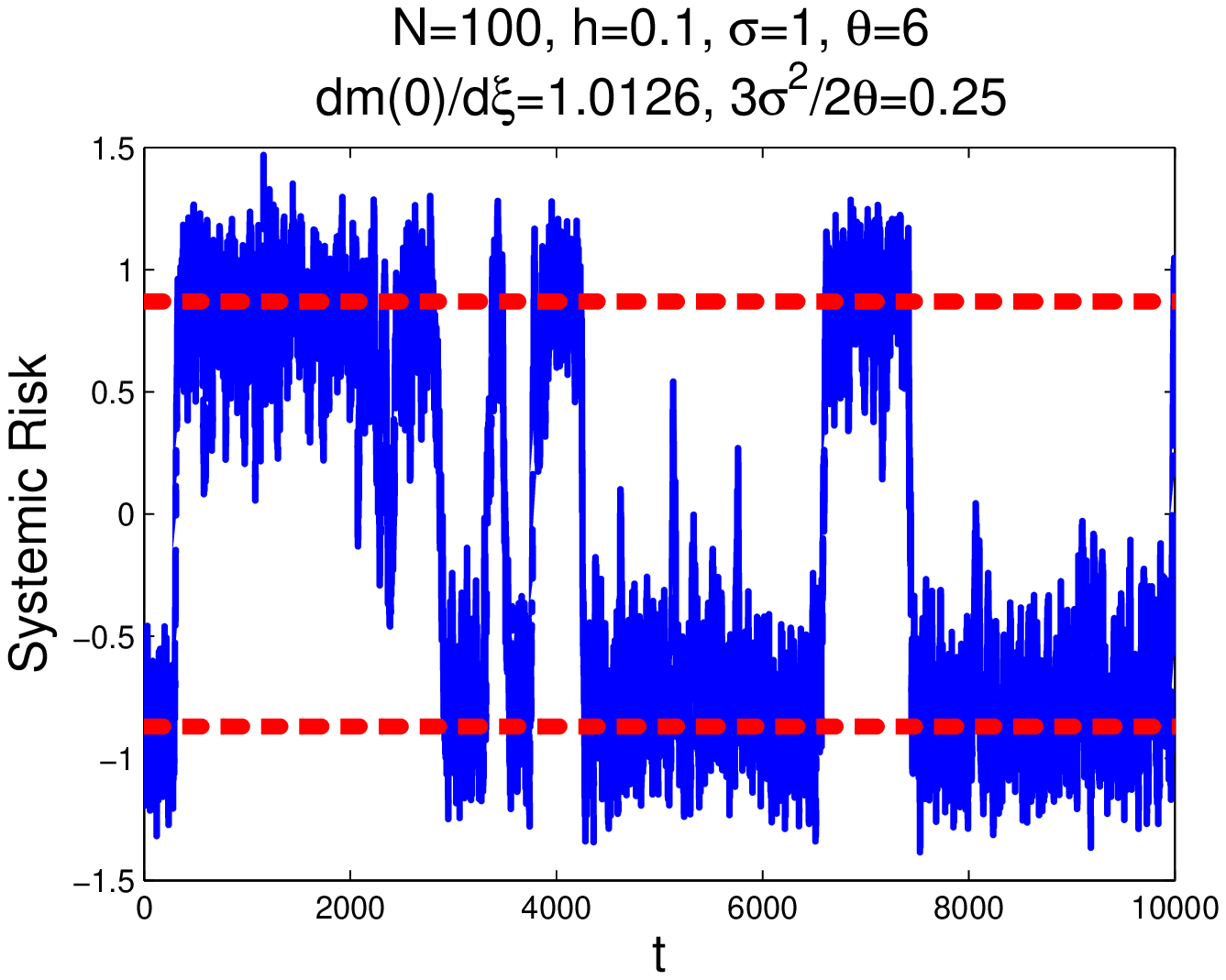}
	\includegraphics[width=0.32\textwidth]{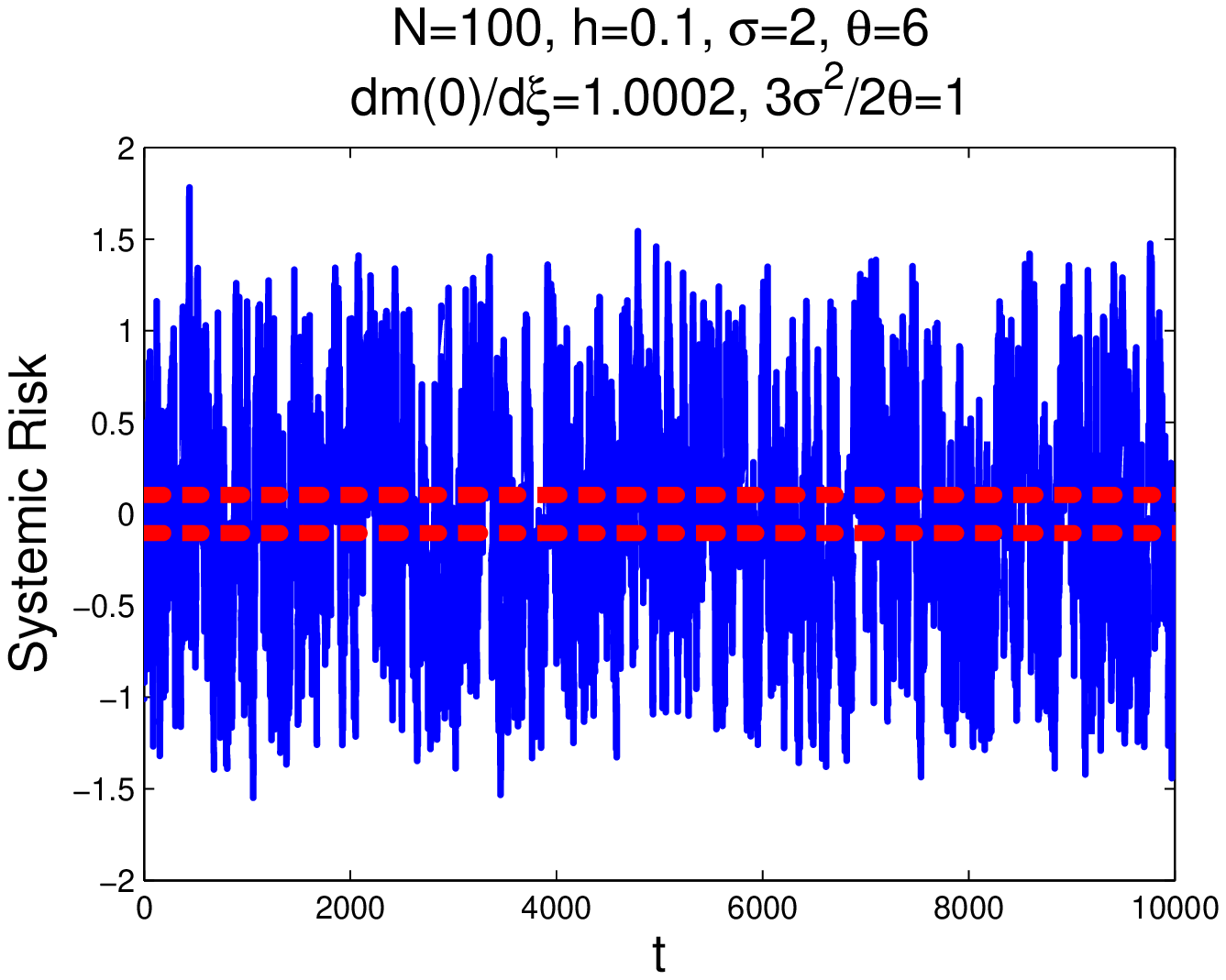}
	\includegraphics[width=0.32\textwidth]{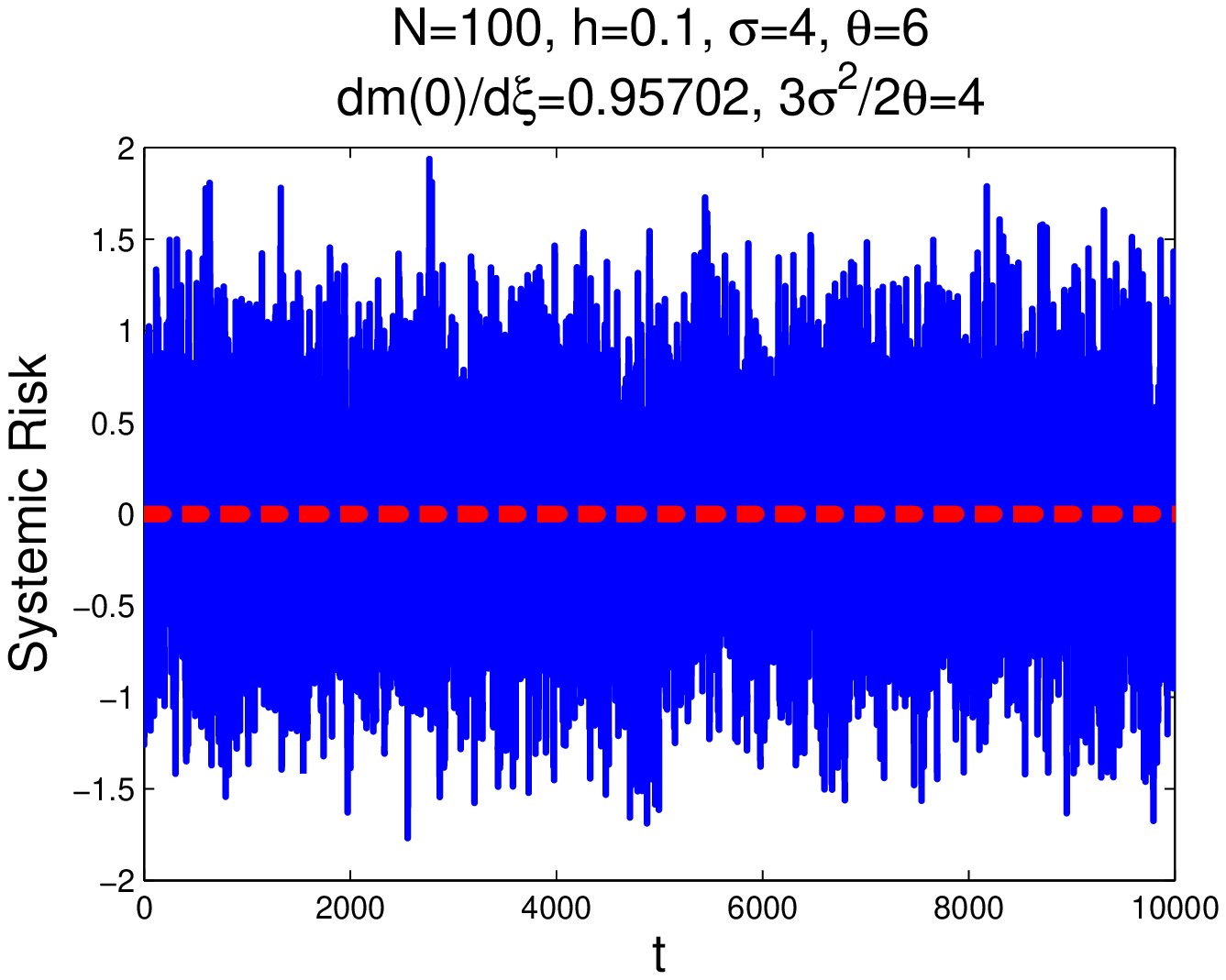}
	\caption{
		\label{fig:The change of sigma}
		Simulations for different $\sigma$. The system has two (statistically) stable equilibria when $\sigma$ is below the 
		critical value or otherwise has single stable state $0$. For small $h$, 
		$3\sigma^2/2\theta<1$ is the approximate criterion.}
\end{figure}

\begin{figure}
	\includegraphics[width=0.32\textwidth]{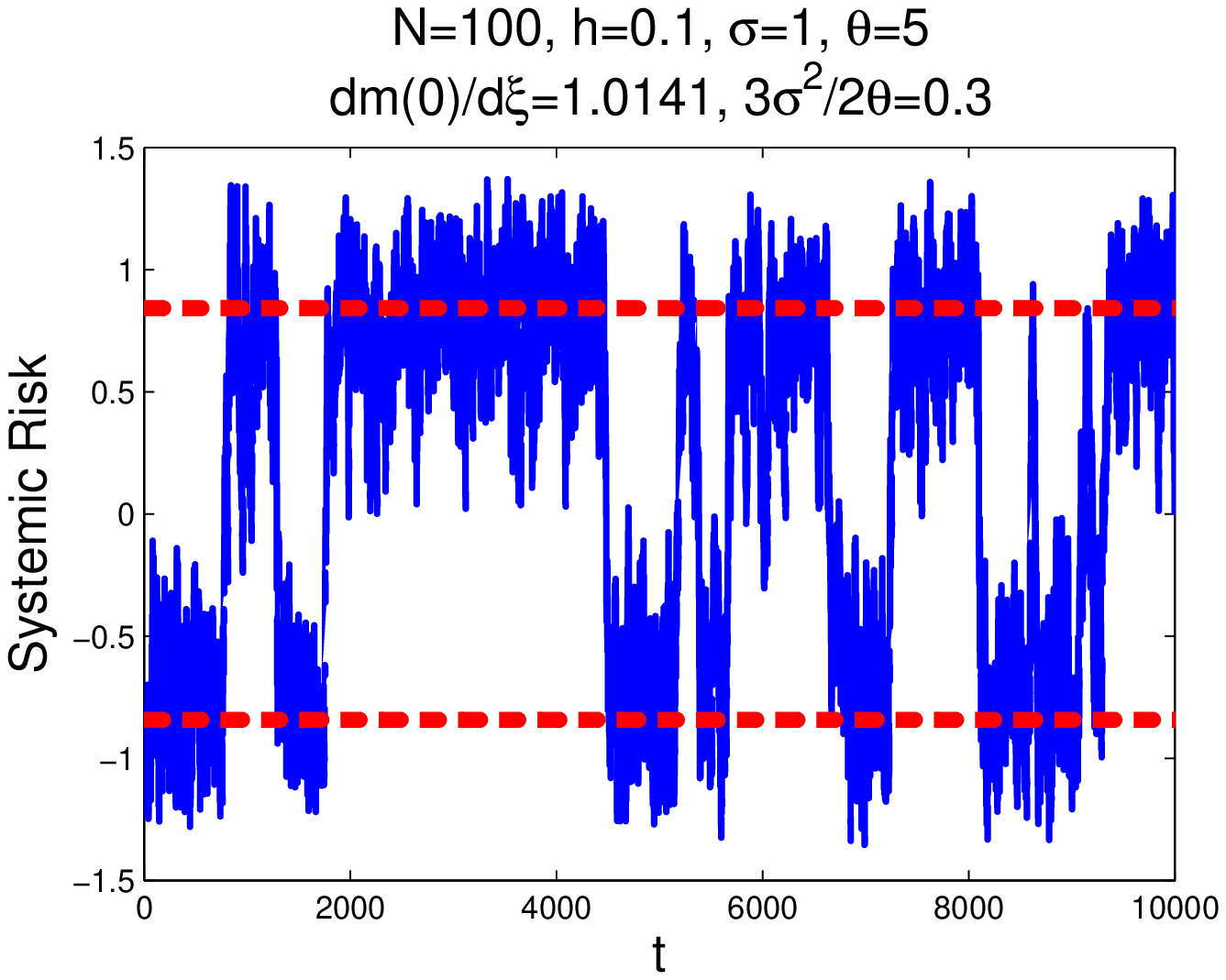}
	\includegraphics[width=0.32\textwidth]{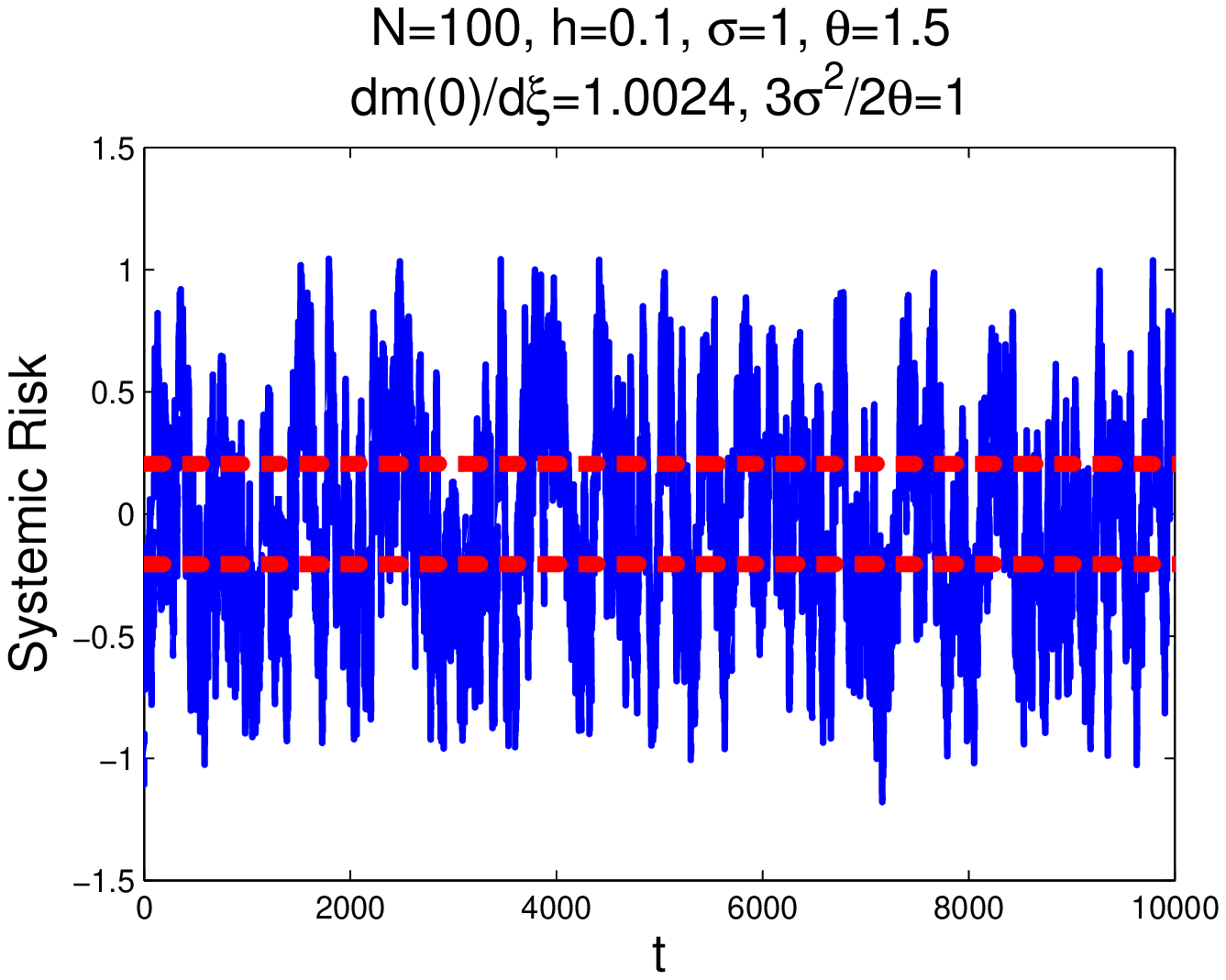}
	\includegraphics[width=0.32\textwidth]{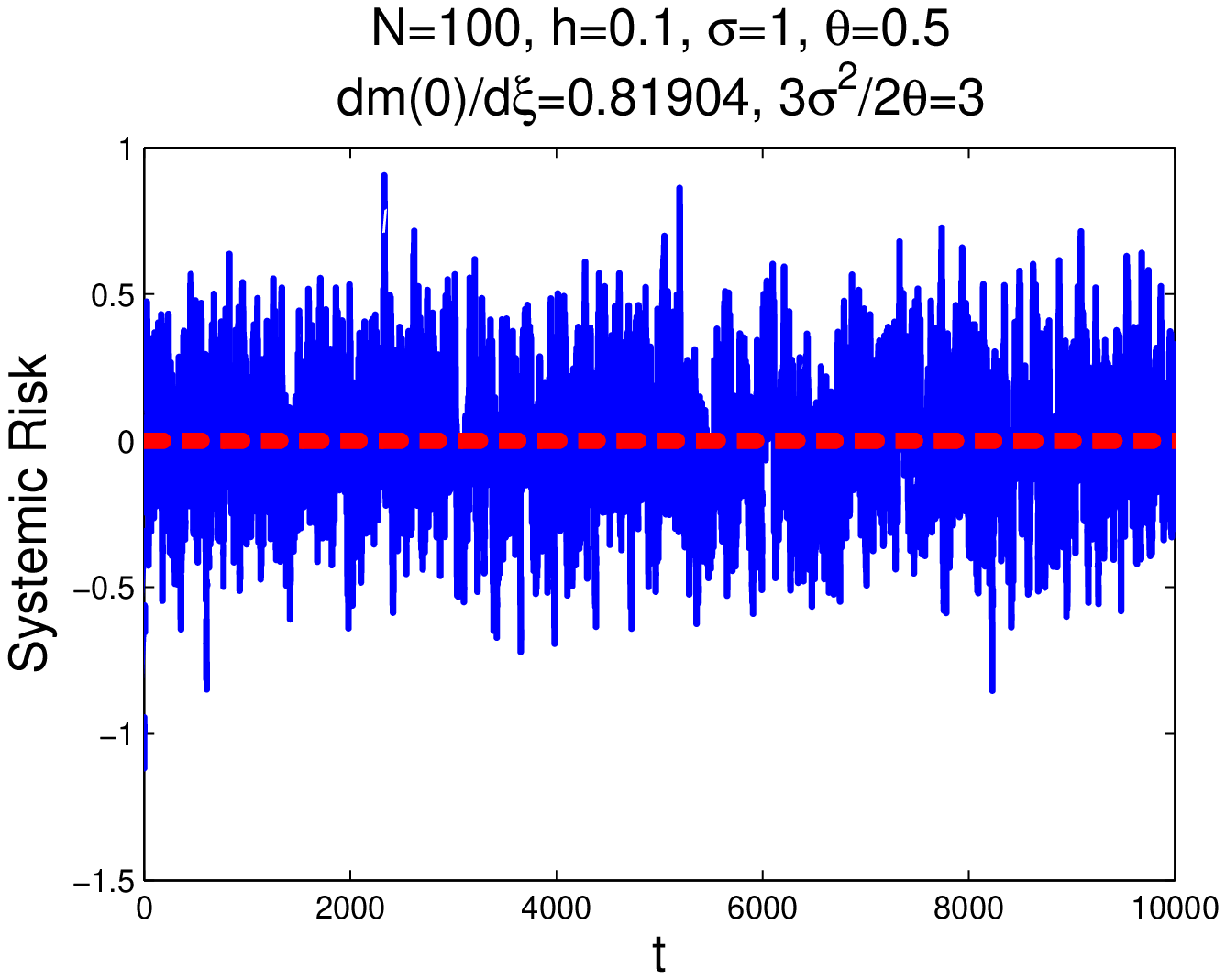}
	\caption{
		\label{fig:The change of theta}
		Simulations for different  $\theta$. The system has two stable equilibria if $\theta$ is above the 
		critical value or otherwise has single stable state $0$. For small $h$, 
		$3\sigma^2/2\theta<1$ is the approximate criterion.}
\end{figure}


Figure \ref{fig:The change of sigma} and Figure \ref{fig:The change of theta} illustrate the behavior of the 
empirical mean as the system transitions from having two equilibria to having a single one, which is controlled 
by the value of $\frac{d}{d\xi}m(0)$. This is an instance of a bifurcation of equilibria. From Proposition 
\ref{prop:explicit condition and solution for xi equal to m(xi) for small h}, we know that when $h$ is 
small, the existence condition of two equilibria, $\frac{d}{d\xi}m(0)>1$, can be approximated by the 
condition $3\sigma^2/2\theta<1$. In the simulations, we let $h=0.1$ so the approximate condition 
$3\sigma^2/2\theta<1$ can be applied. In Figure \ref{fig:The change of sigma} we change $\sigma$ but fix the 
other parameters, and consider the three cases $\frac{d}{d\xi}m(0)<1$ ($3\sigma^2/2\theta>1$), 
$\frac{d}{d\xi}m(0)\approx 1$ ($3\sigma^2/2\theta=1$) and $\frac{d}{d\xi}m(0)>1$ ($3\sigma^2/2\theta<1$). In 
Figure \ref{fig:The change of theta} we change $\theta$. We can see that even though the  
parameters varied in the numerical simulations are not the same, the bifurcation behavior is similar. 

\begin{figure}
	\includegraphics[width=0.32\textwidth]{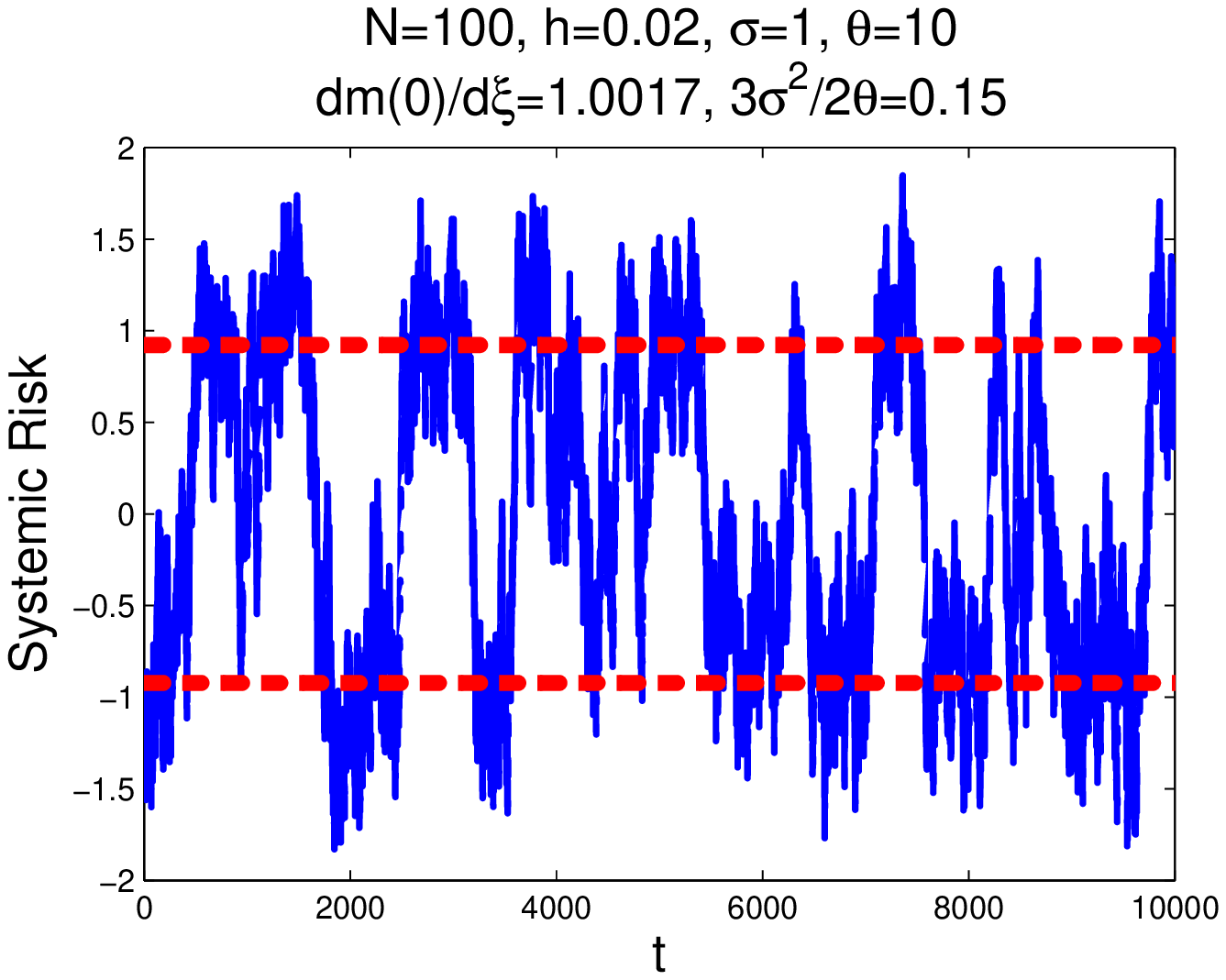}
	\includegraphics[width=0.32\textwidth]{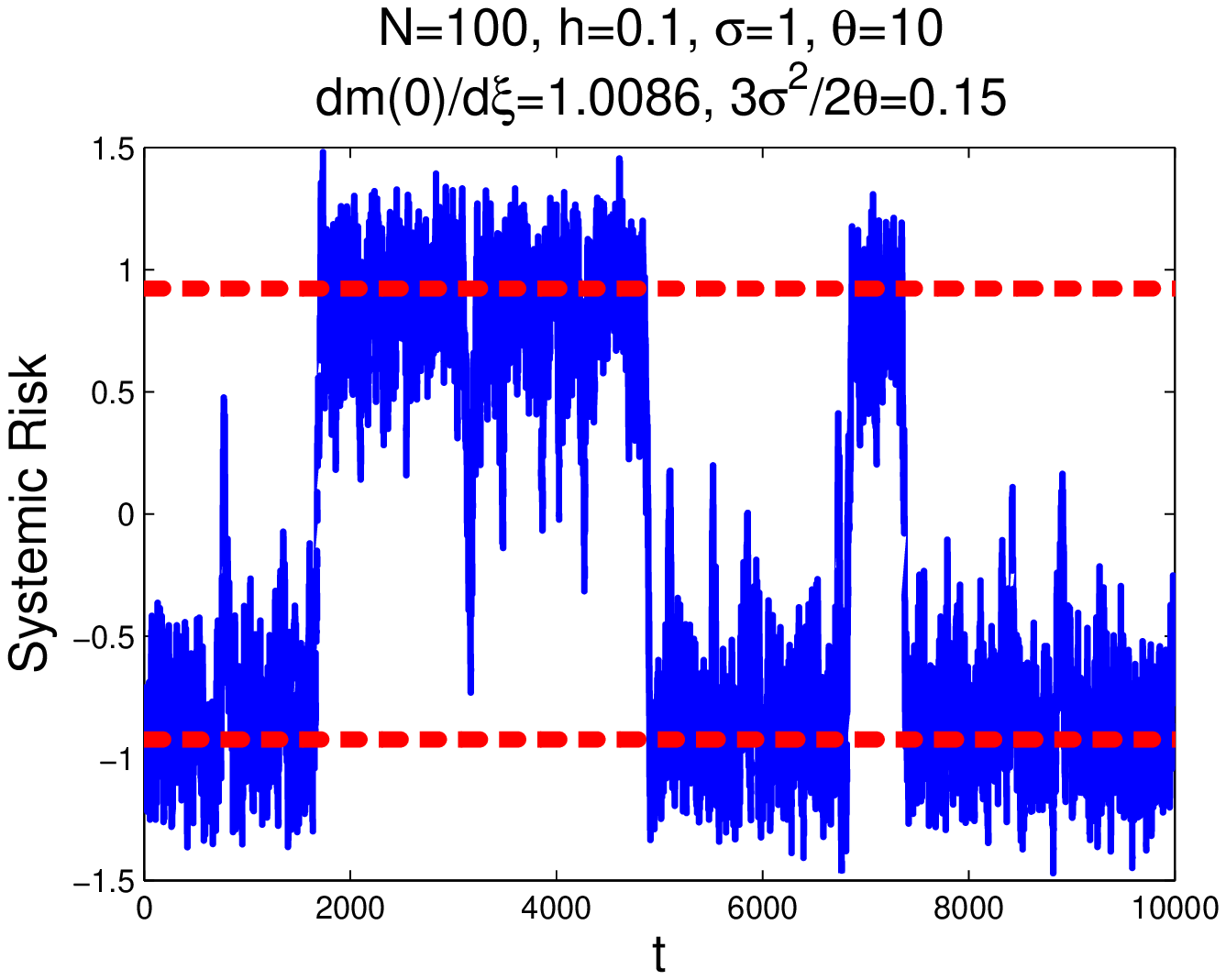}
	\includegraphics[width=0.32\textwidth]{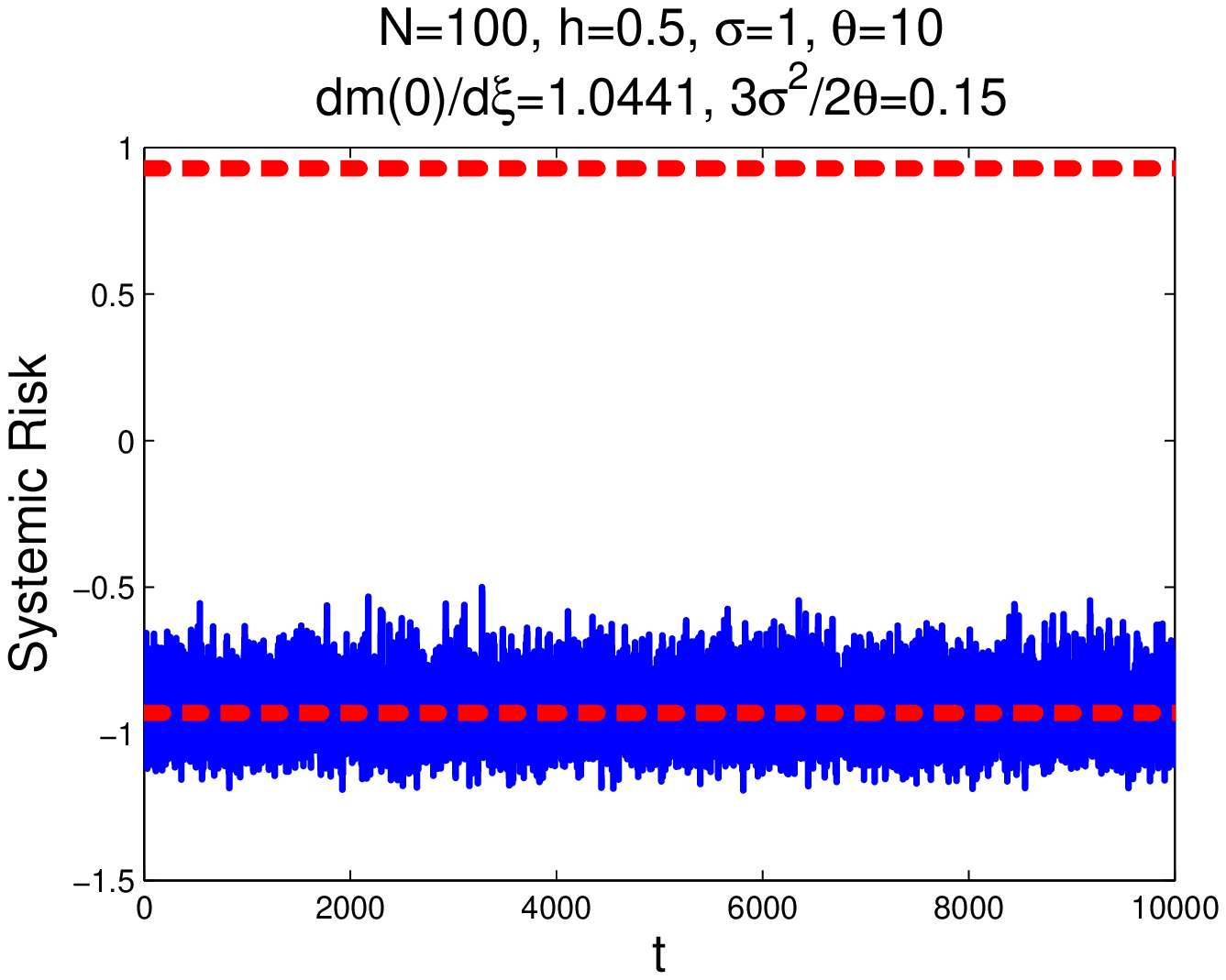}
	\caption{
		\label{fig:The change of h}
		The effect of changing $h$. Increasing it stabilizes the system.}
\end{figure}

Figure \ref{fig:The change of h} shows the effect of increasing $h$ on the system stability.  
By stability we mean resistance to the transition of the empirical mean of the system 
from one state to the other (because the model is symmetric). The parameter $h$  
is proportional to the height of the potential barrier of each agent. Thus we increase
the overall system stability if we increase the component's stability. This observation 
is analogous to comments in \cite{Nier2007, Lorenz2009, May2010}. It is clear that 
$h$ influences system stability substantially. 

\begin{figure}
	\includegraphics[width=0.32\textwidth]{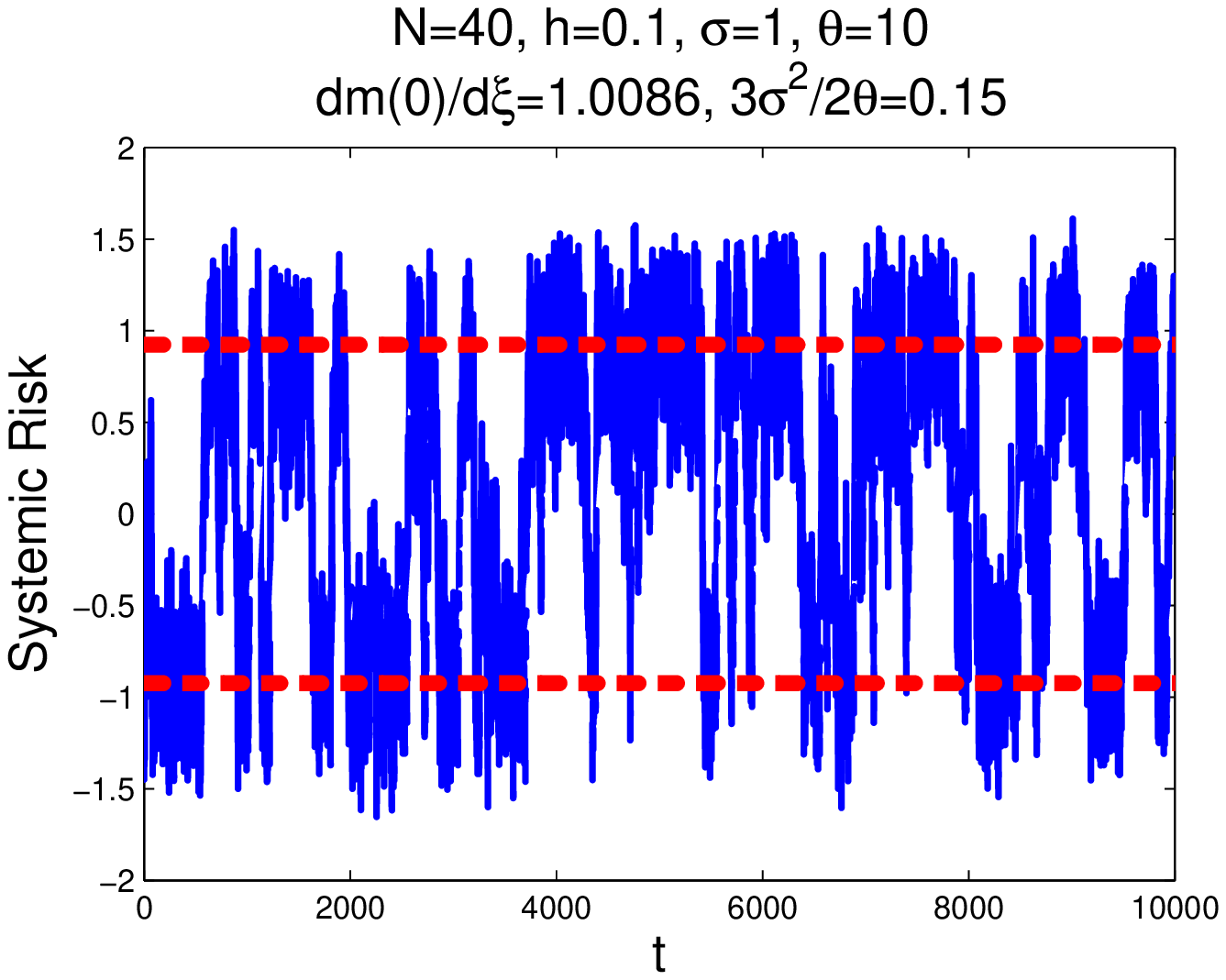}
	\includegraphics[width=0.32\textwidth]{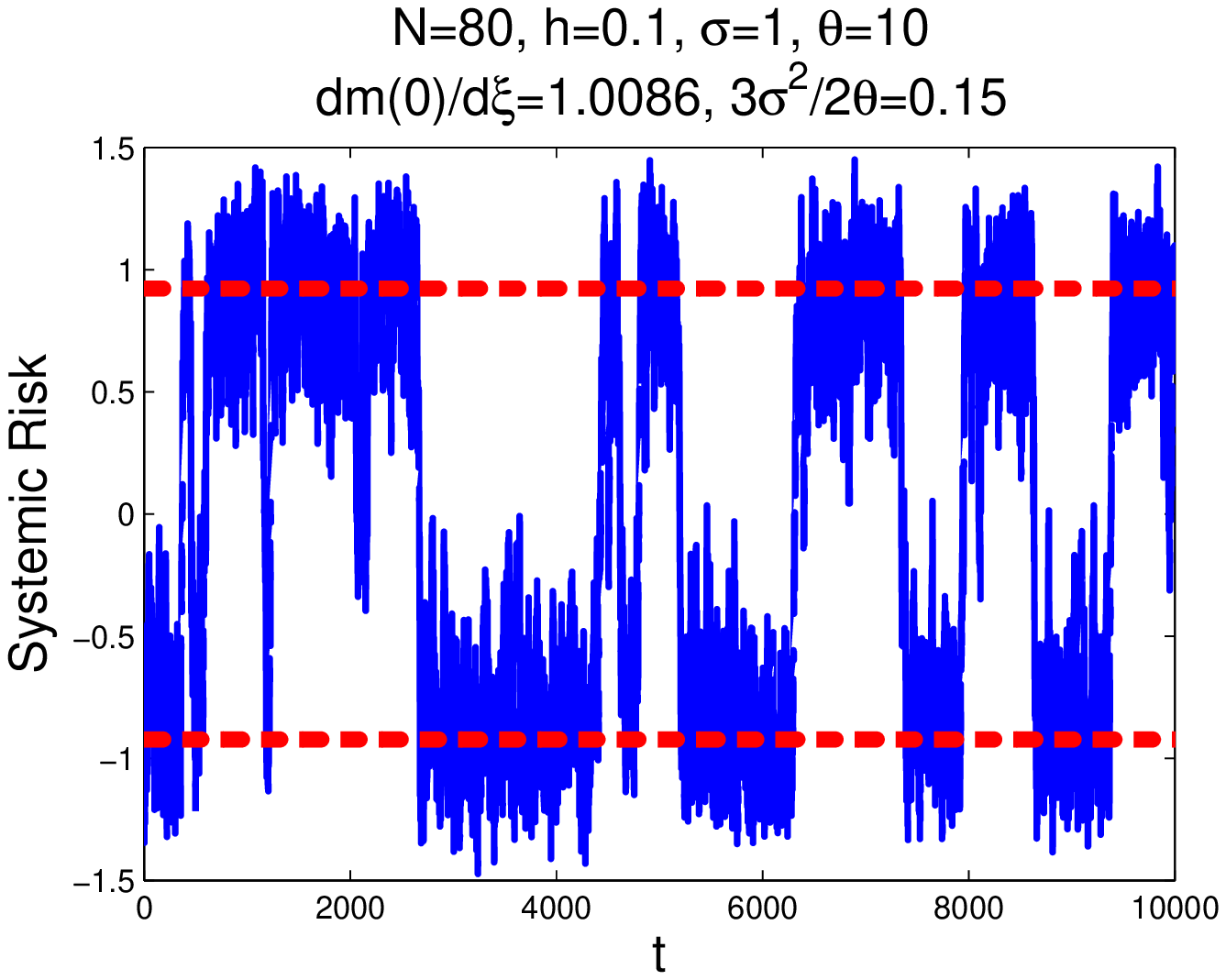}
	\includegraphics[width=0.32\textwidth]{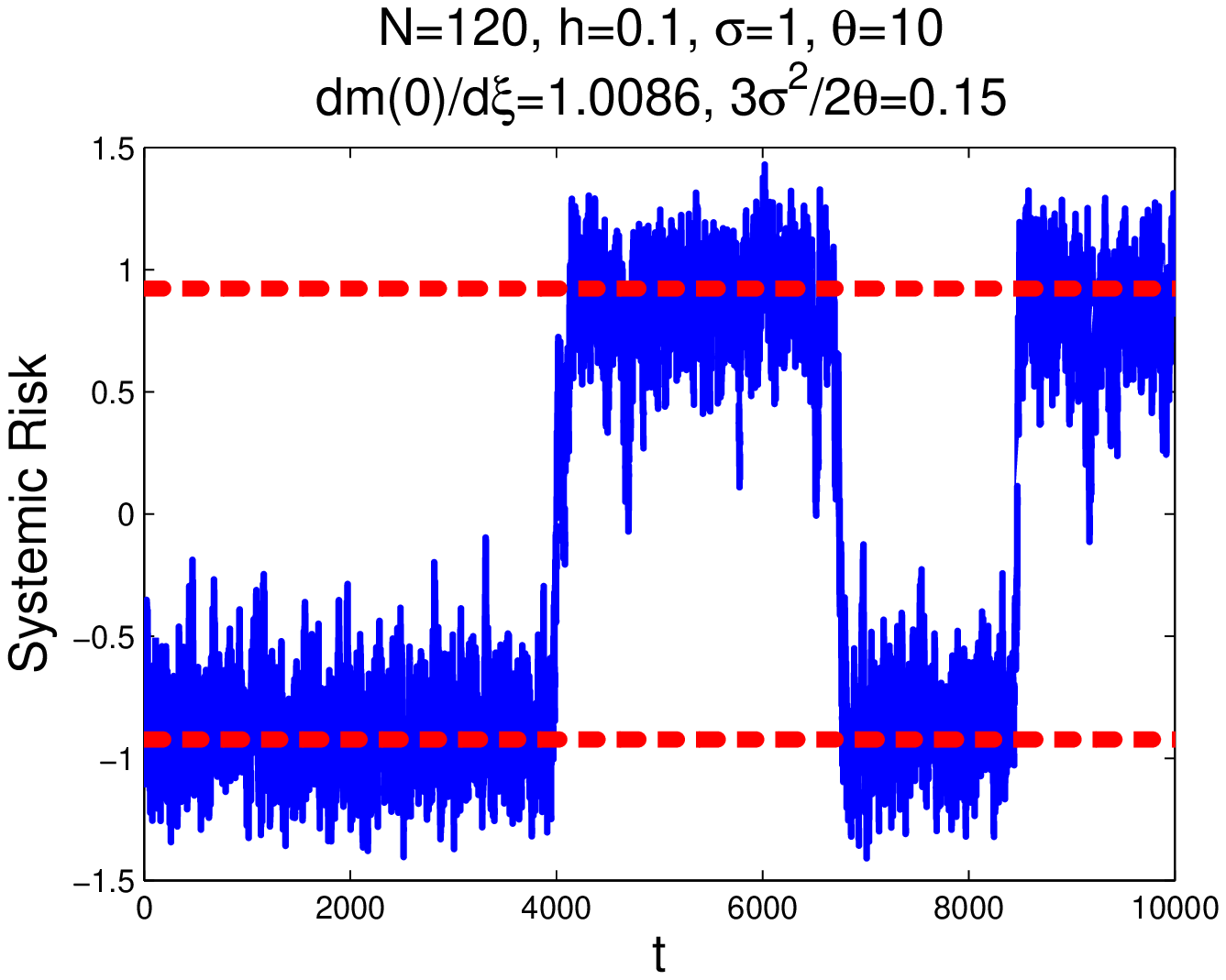}
	\caption{
		\label{fig:The change of N}
		Influence of the system size $N$. A larger system tends to have a more stable behavior.}
\end{figure}

Figure \ref{fig:The change of N} illustrates the effect of system size on its  
stability. Clearly a larger system is more stable.
These stability phenomena will be quantified with the large deviations analysis  of Section 
\ref{sec:large deviations}.

\subsection{Heterogeneous Model}

For the heterogeneous model, $\theta$ is replaced by $\theta_j$, 
and the discretization is 
\begin{equation}
	\label{eq:Euler method, heterogeneous case}
	X_j^{n+1} = X_j^n - h U(X_j^n)\Delta t + \sigma\Delta W_j^{n+1} 
	+ \theta_j(\frac{1}{N}\sum_{k=1}^N X_k^n - X_j^n)\Delta t,
\end{equation}
with the same parameter settings. The different values of $\theta_j$ are controlled
by the parameters  $\Theta_l$ and 
$\rho_l$. In the simulation, we take $K=3$ and
$\{\Theta_l\}_{l=1}^K=\{ \Theta_L,\Theta_M,\Theta_H\}$  for a system
a low, medium and high rates of mean reversion to the empirical mean, that is, the systemic risk. 
We also take $\{\rho_l\}_{l=1}^K=\{\rho_L,\rho_M,\rho_H\}$ for the corresponding fractions. We 
use the normalized standard deviation of the distribution of $\theta_j$ values in order to quantify  
diversity. We find that the heterogeneous model behaves like the homogeneous one when
$h$, $\sigma$ and $N$ change.  But, diversity on the rates of mean reversion 
has significant impact on system stability.

\begin{figure}
	\includegraphics[width=0.32\textwidth]{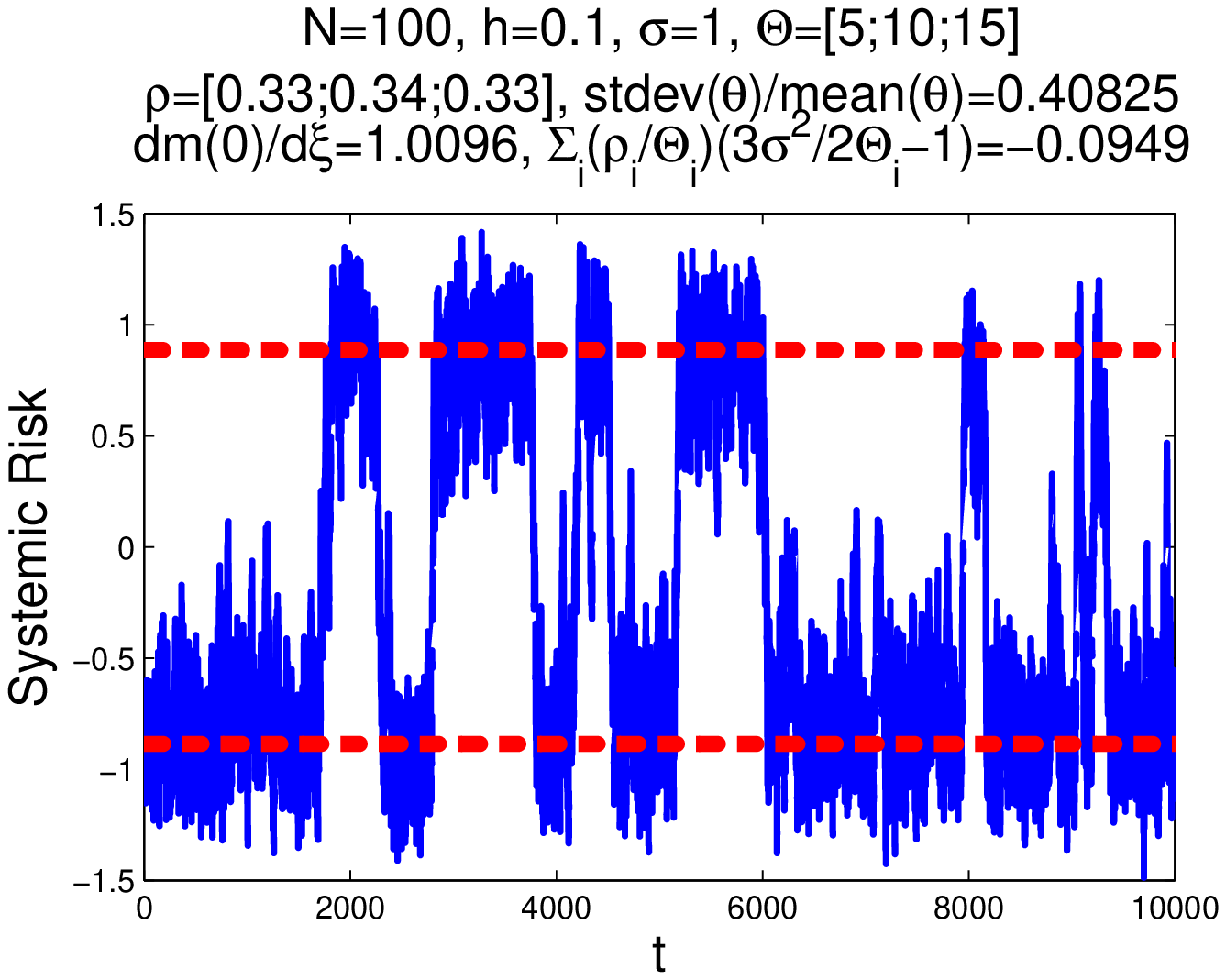}
	\includegraphics[width=0.32\textwidth]{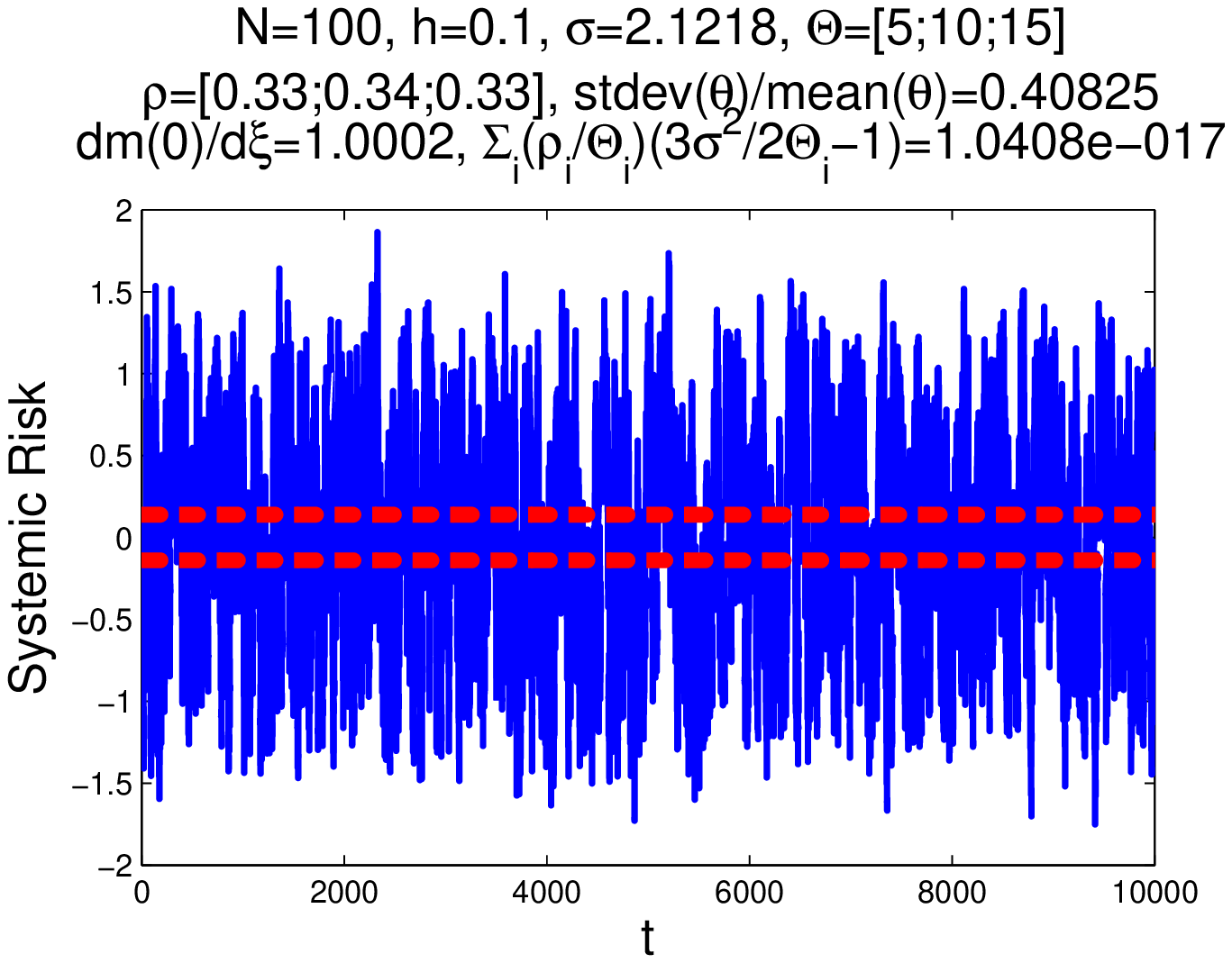}
	\includegraphics[width=0.32\textwidth]{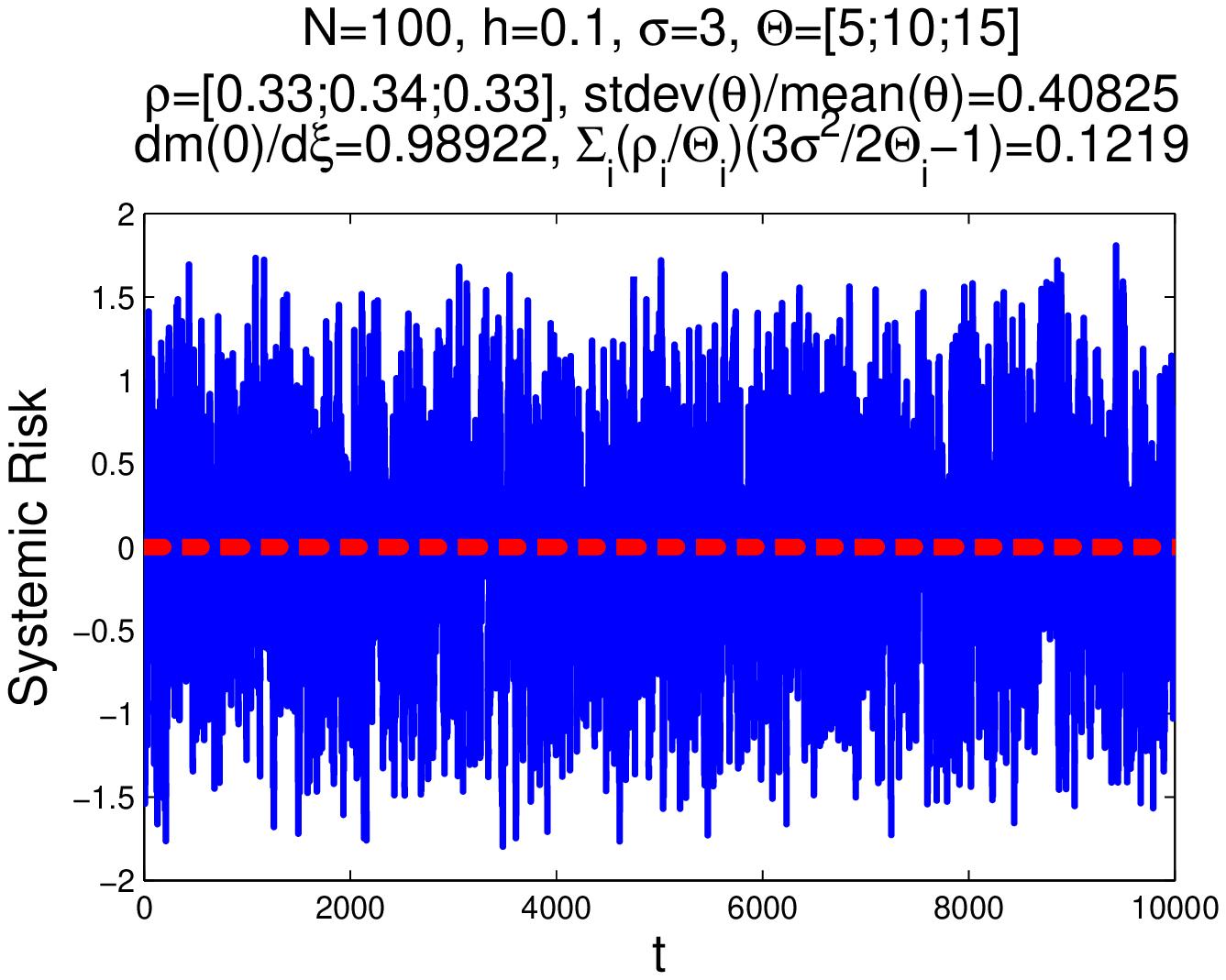}
	\caption{
		\label{fig:The change of sigma, diversity case}
		Effect of changes in $\sigma$. The system has two stable equllibria when $\sigma$ is below the 
		critical value and has single one otherwise. For small $h$, 
		$\sum_{l=1}^K(\rho_l/\Theta_l)(3\sigma^2/2\Theta_l-1) < 1$ is the approximate 
		criterion.}
\end{figure}

As in the homogeneous case, in Figure \ref{fig:The change of sigma, diversity case} we consider cases with  
$\sigma$ below, close to and above the critical value. The results are similar to the homogeneous case as 
expected. For $\sigma$ below the critical value we have two equlibria and for $\sigma$ above the critical 
value one equilibrium. The condition $\frac{d}{d\xi}m(0)>1$ is still necessary and sufficient for the
existence two equlibria. The condition $\sum_{l=1}^K(\rho_l/\Theta_l)(3\sigma^2/2\Theta_l-1) < 1$ is also a 
good approximation to the exact one when $h$ is small.

\begin{figure}
	\includegraphics[width=0.32\textwidth]{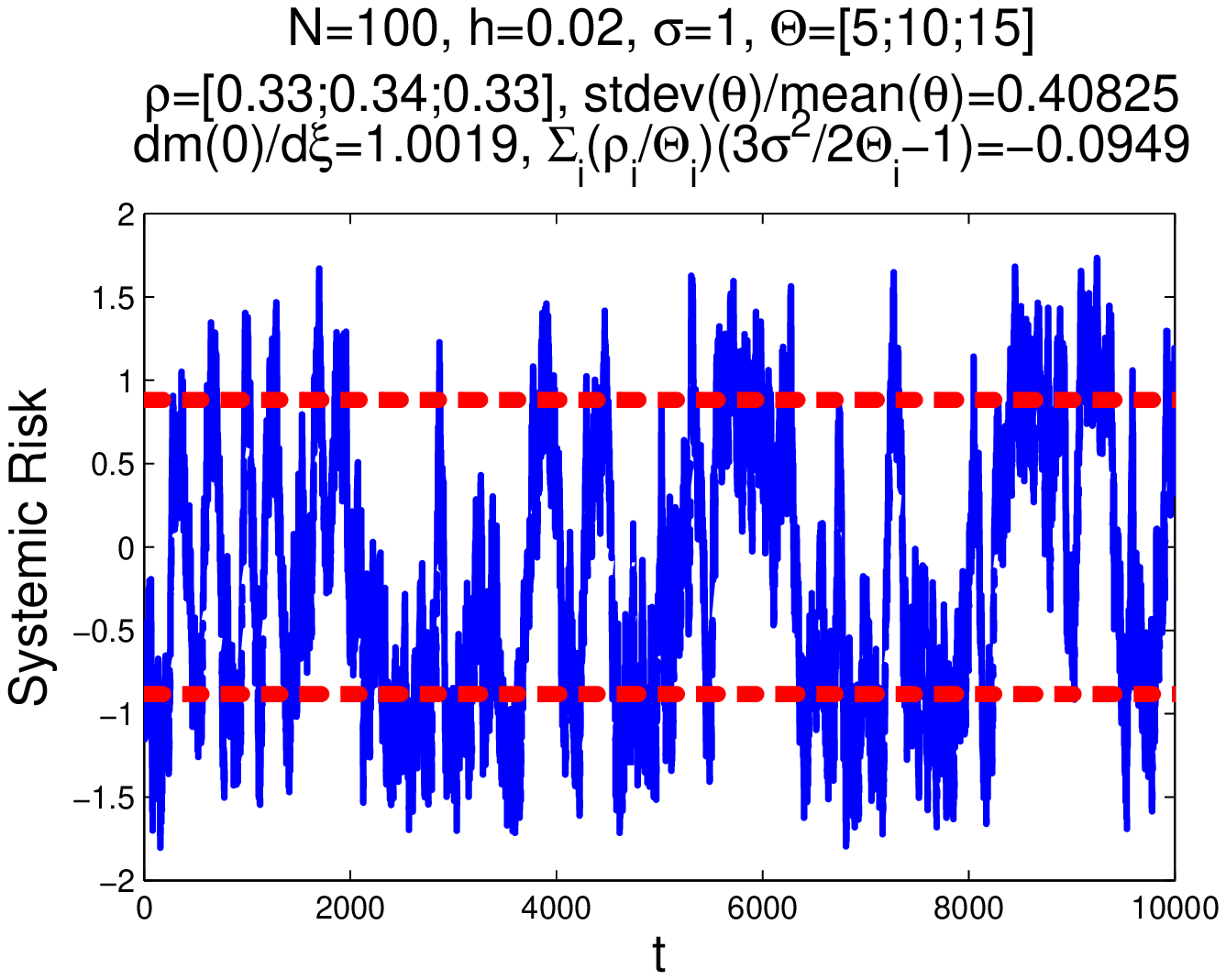}
	\includegraphics[width=0.32\textwidth]{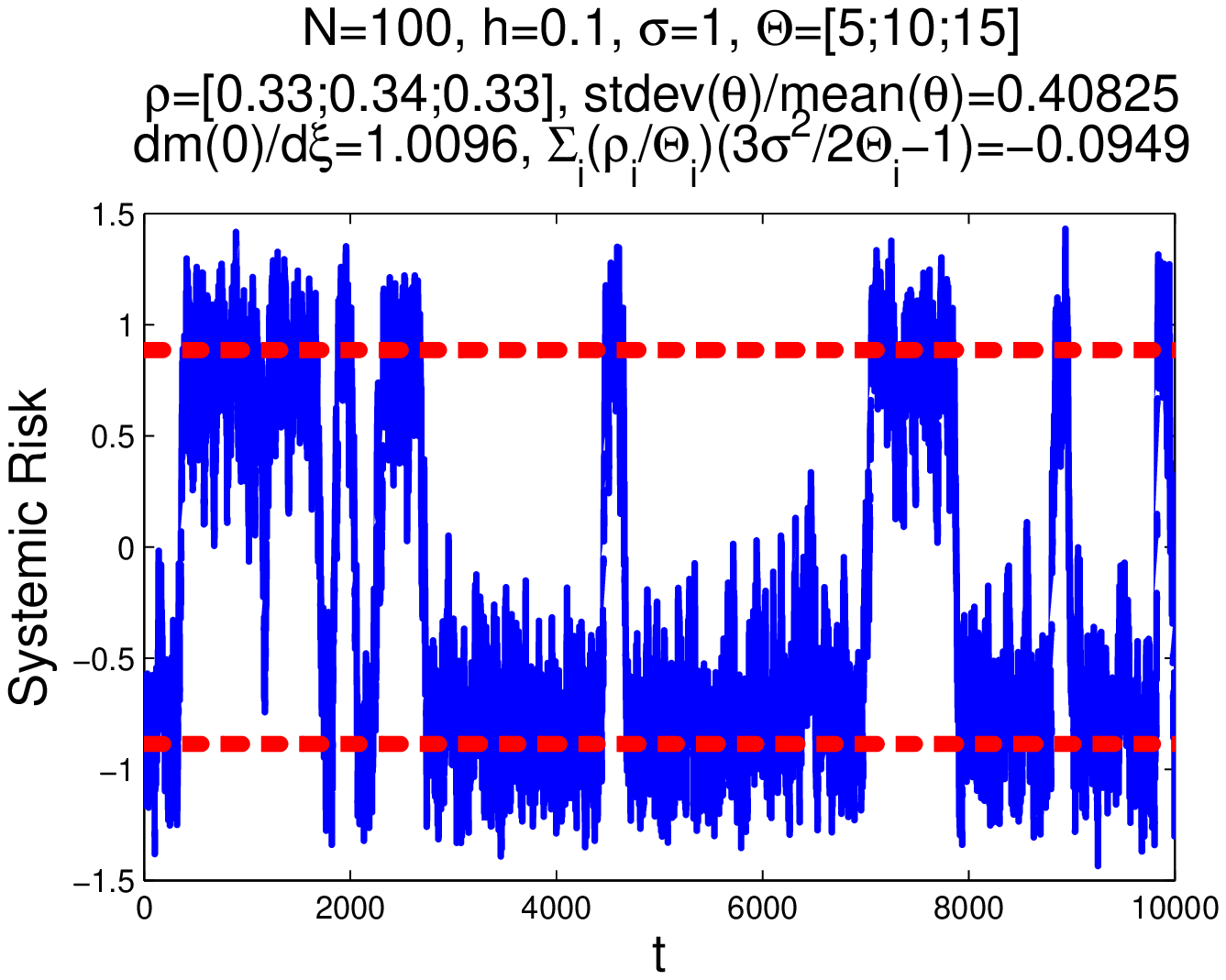}
	\includegraphics[width=0.32\textwidth]{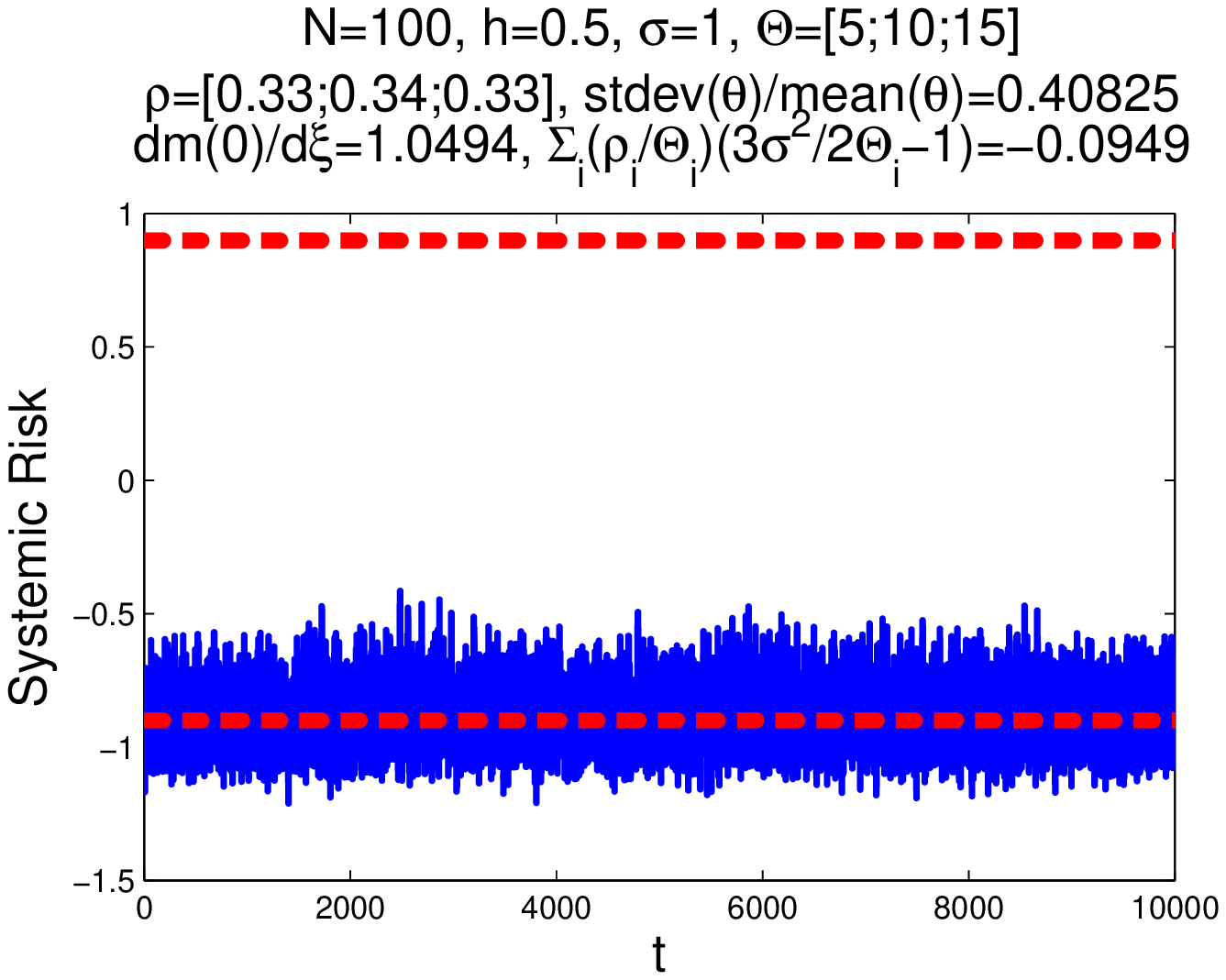}
	\caption{
		\label{fig:The change of h, diversity case}
		Effect of changing $h$. Increasing it stabilizes the system.}
\end{figure}

\begin{figure}
	\includegraphics[width=0.32\textwidth]{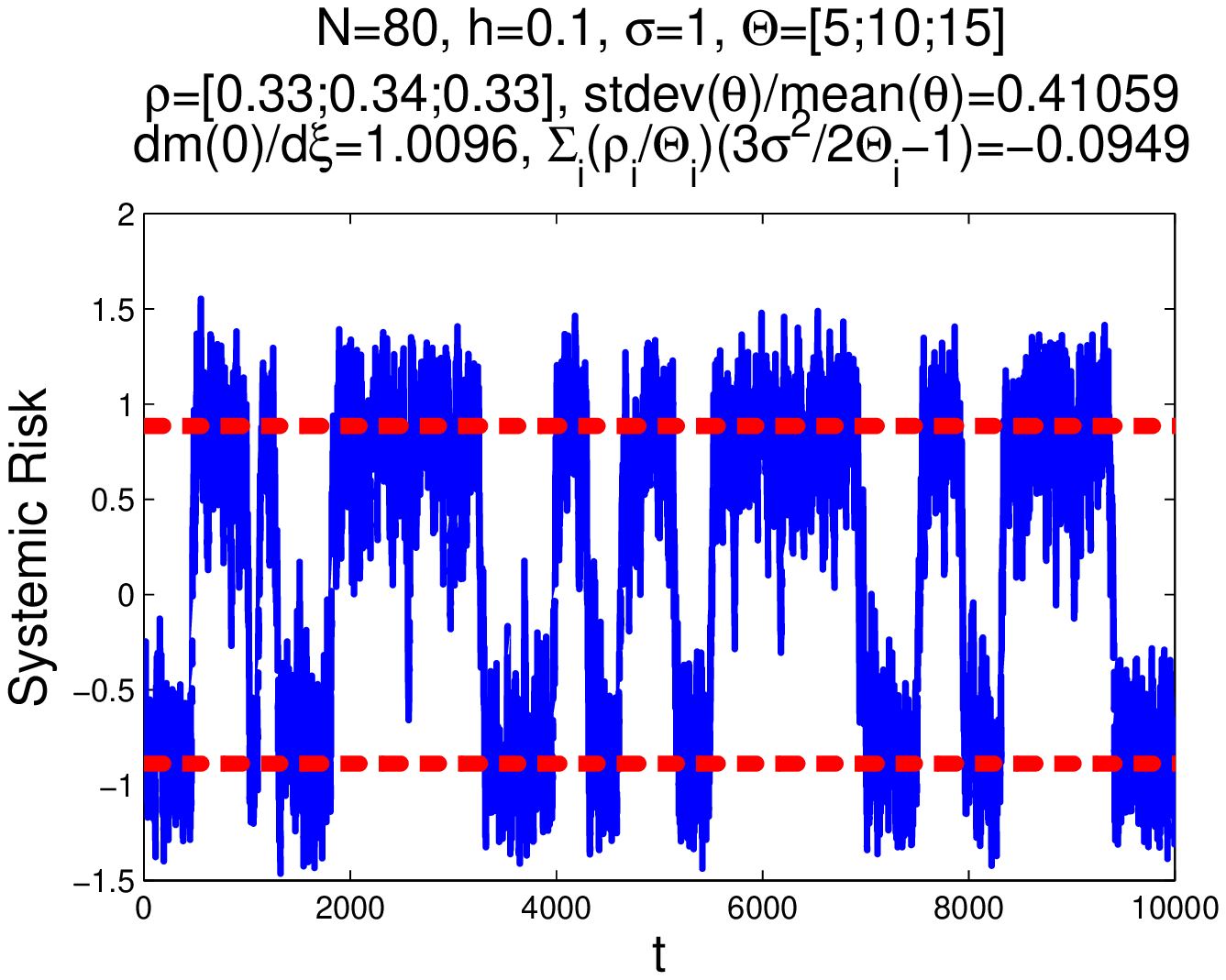}
	\includegraphics[width=0.32\textwidth]{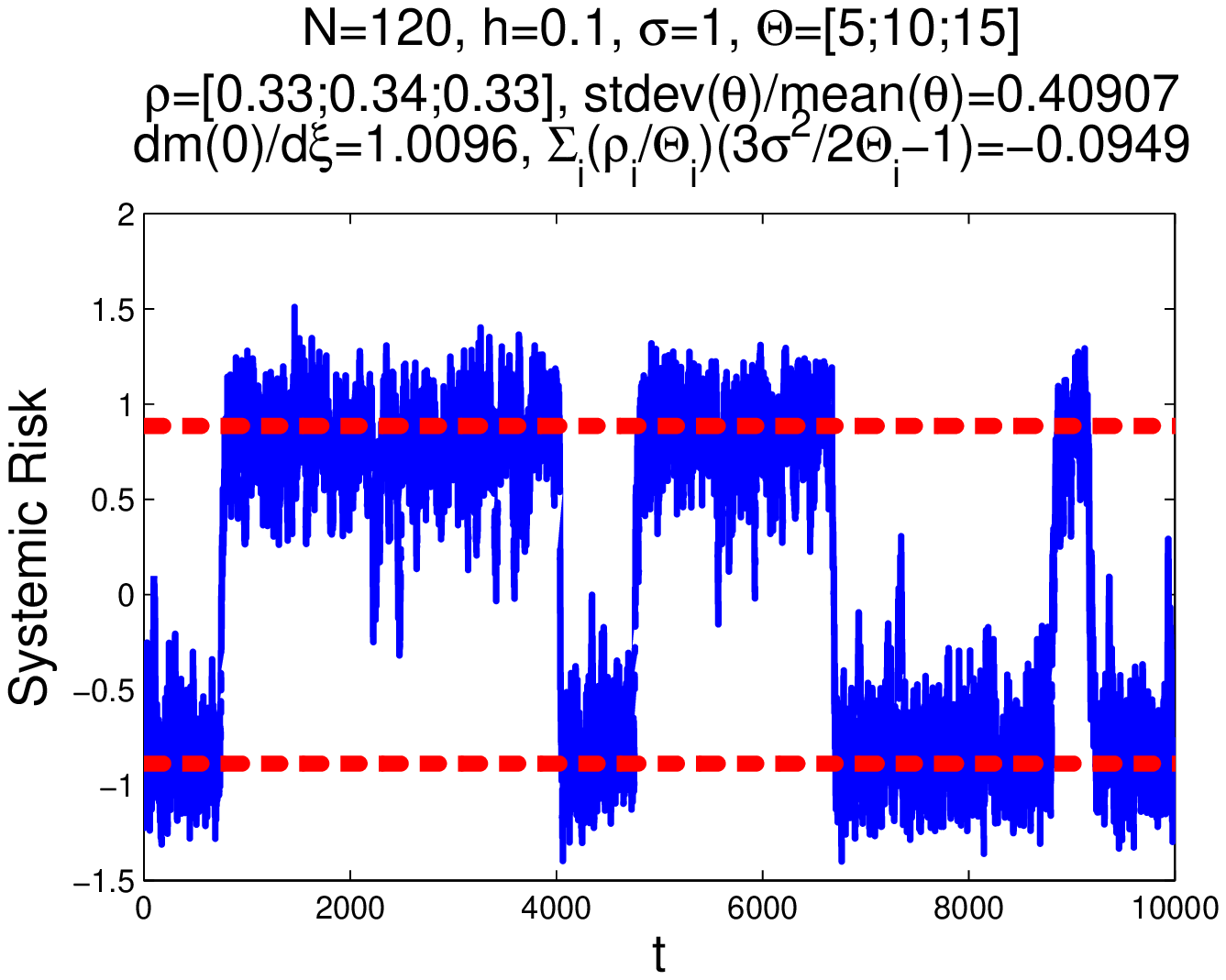}
	\includegraphics[width=0.32\textwidth]{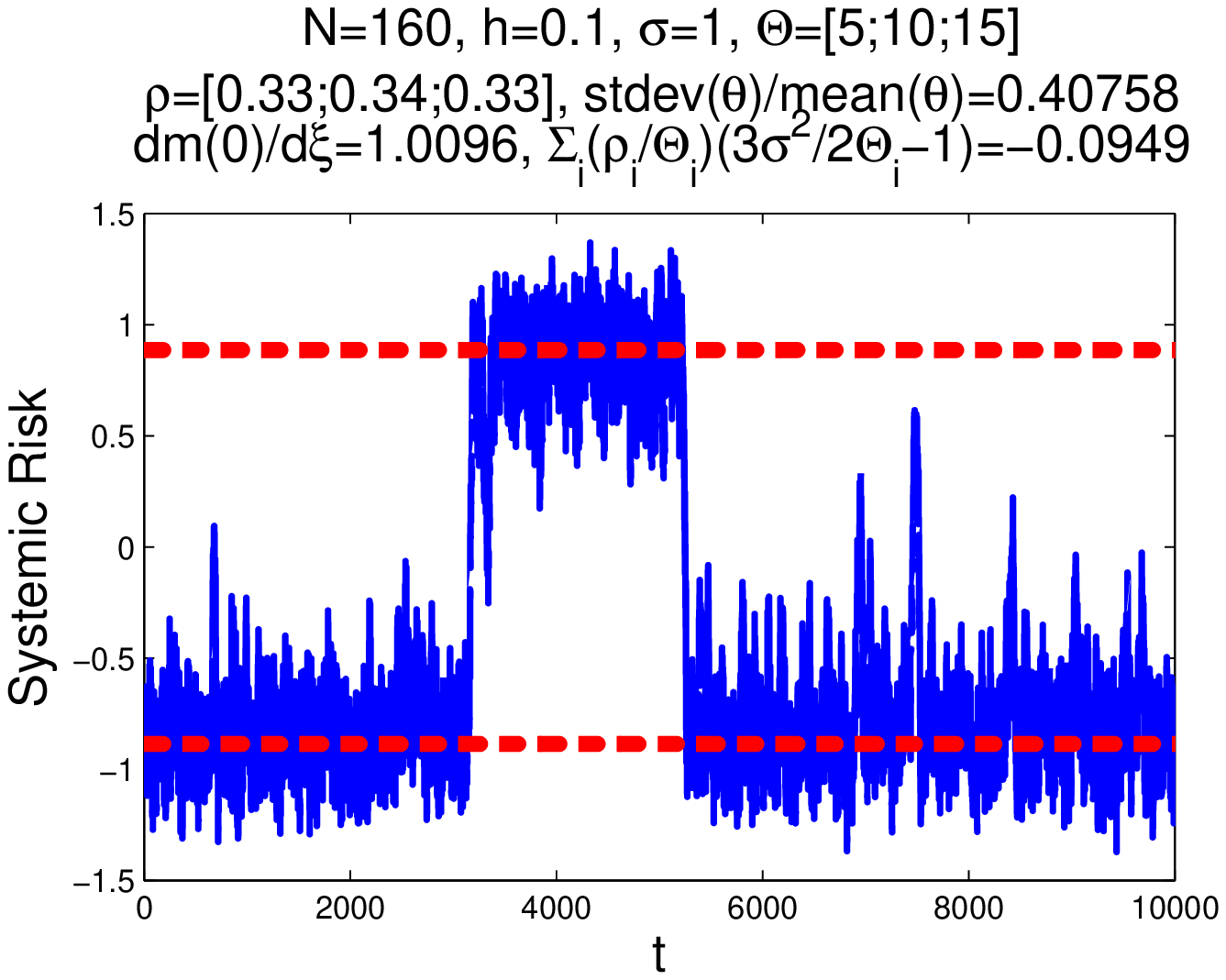}
	\caption{
		\label{fig:The change of N, diversity case}
		Effect of changing the system size $N$. Larger system have a more stable behavior.}
\end{figure}

The parameter $h$ and the system size $N$ are closely associated with system stability. We note that in 
Figure \ref{fig:The change of h, diversity case} and Figure \ref{fig:The change of N, diversity case} when 
$h$ or $N$ are increased, the system becomes visibly more stable. Another observation is that with $h$, 
$\sigma$ and $N$ fixed, and with the mean of $\theta_j$ of (\ref{eq:Euler method, heterogeneous case}) equal 
to $\theta$ of (\ref{eq:Euler method, homogeneous case}), the heterogeneous system is consistently more 
unstable than the corresponding homogeneous model (see Figure \ref{fig:The change of h} and Figure 
\ref{fig:The change of N}). Clearly diversity tends to destabilize the system.

\begin{figure}
	\includegraphics[width=0.32\textwidth]{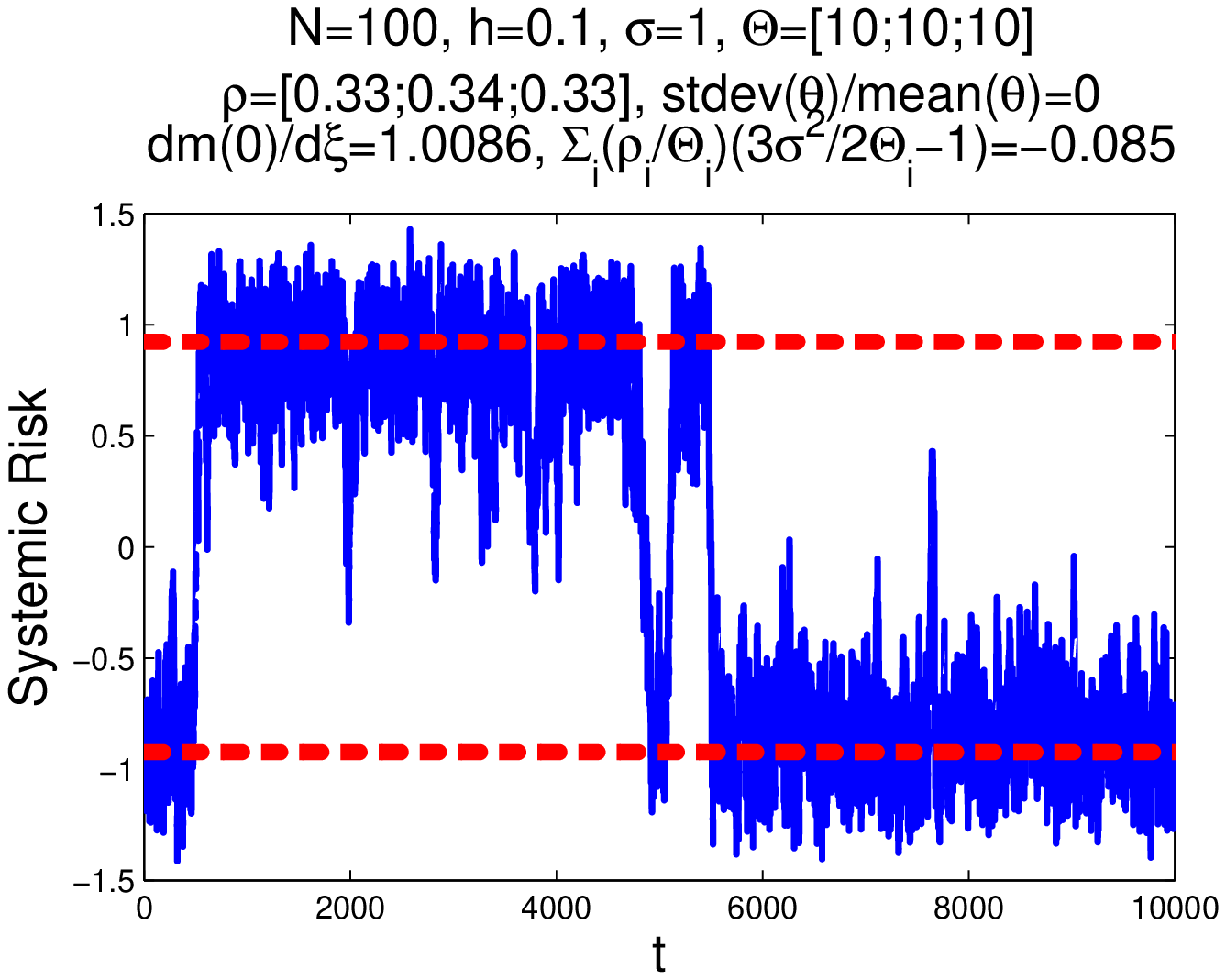}
	\includegraphics[width=0.32\textwidth]{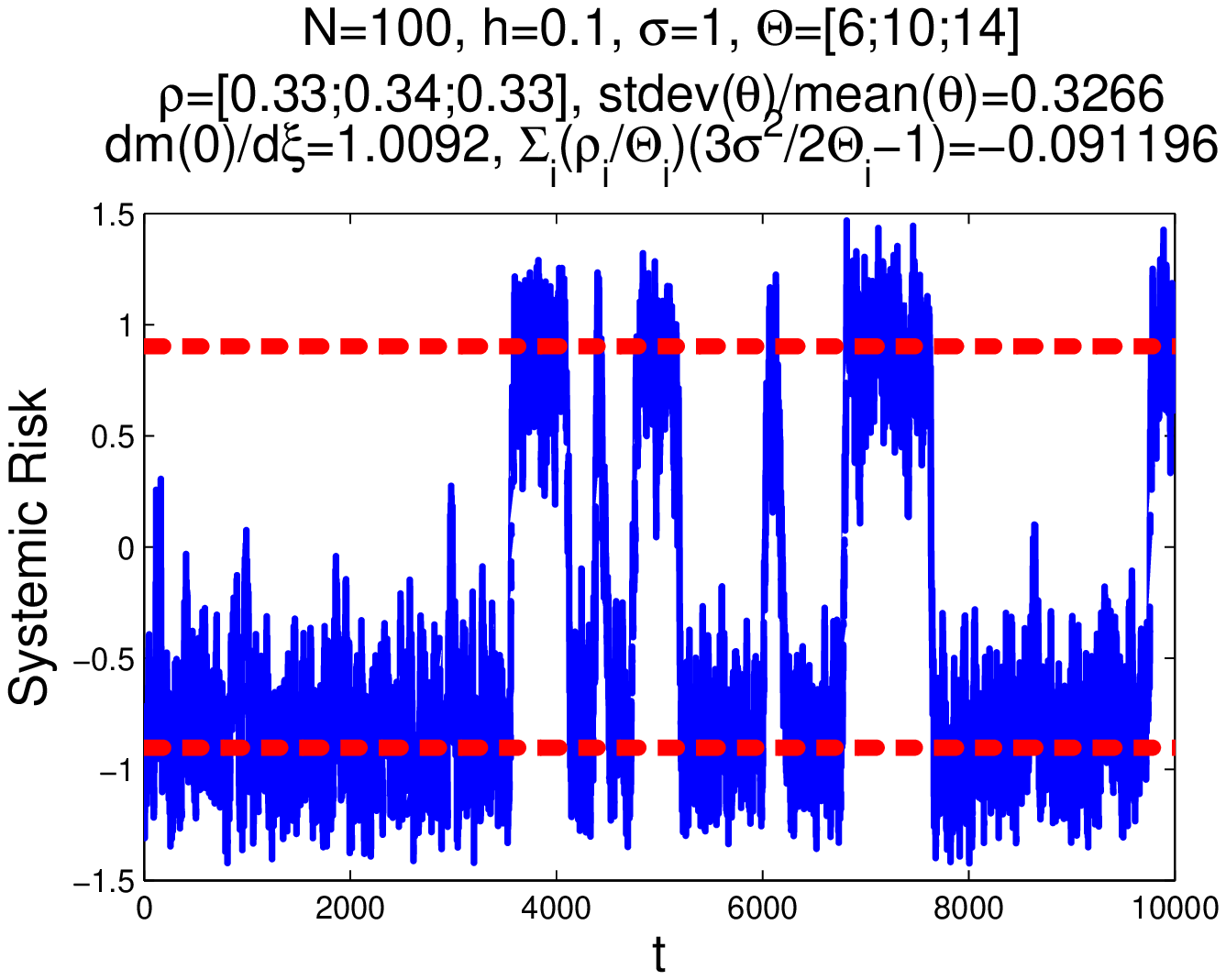}
	\includegraphics[width=0.32\textwidth]{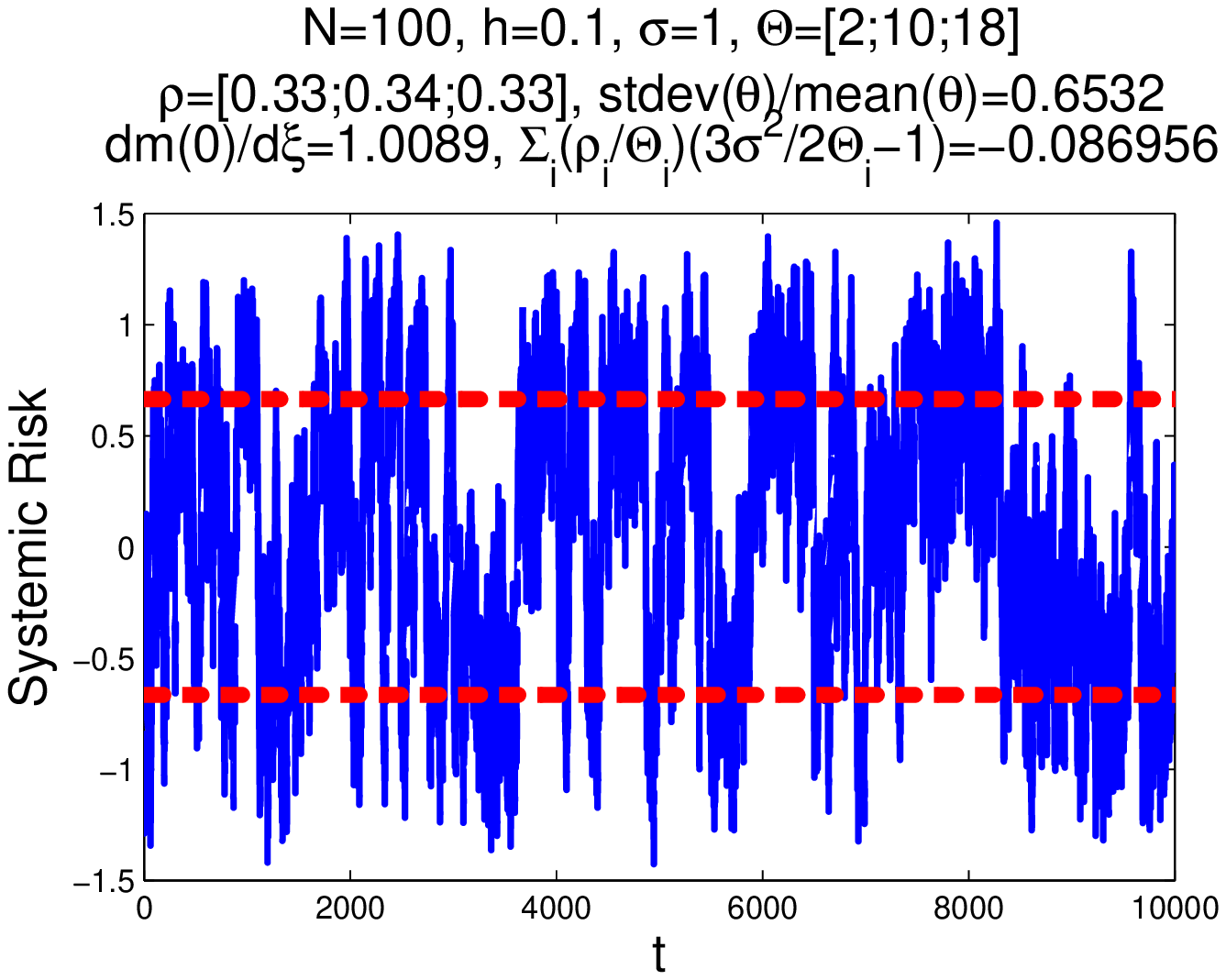}
	\caption{
		\label{fig:The change of Theta, diversity case}
		The effect of changes in $\Theta_l$. The median of the diversity values is fixed but the low and high 
		sensitivities are changed to adjust the level of diversity of $\theta_j$ while $\rho_l$ 
		and the mean of $\theta_j$ are the same. Increasing diversity tends to destabilize 
		the system.}
\end{figure}

\begin{figure}
	\includegraphics[width=0.32\textwidth]{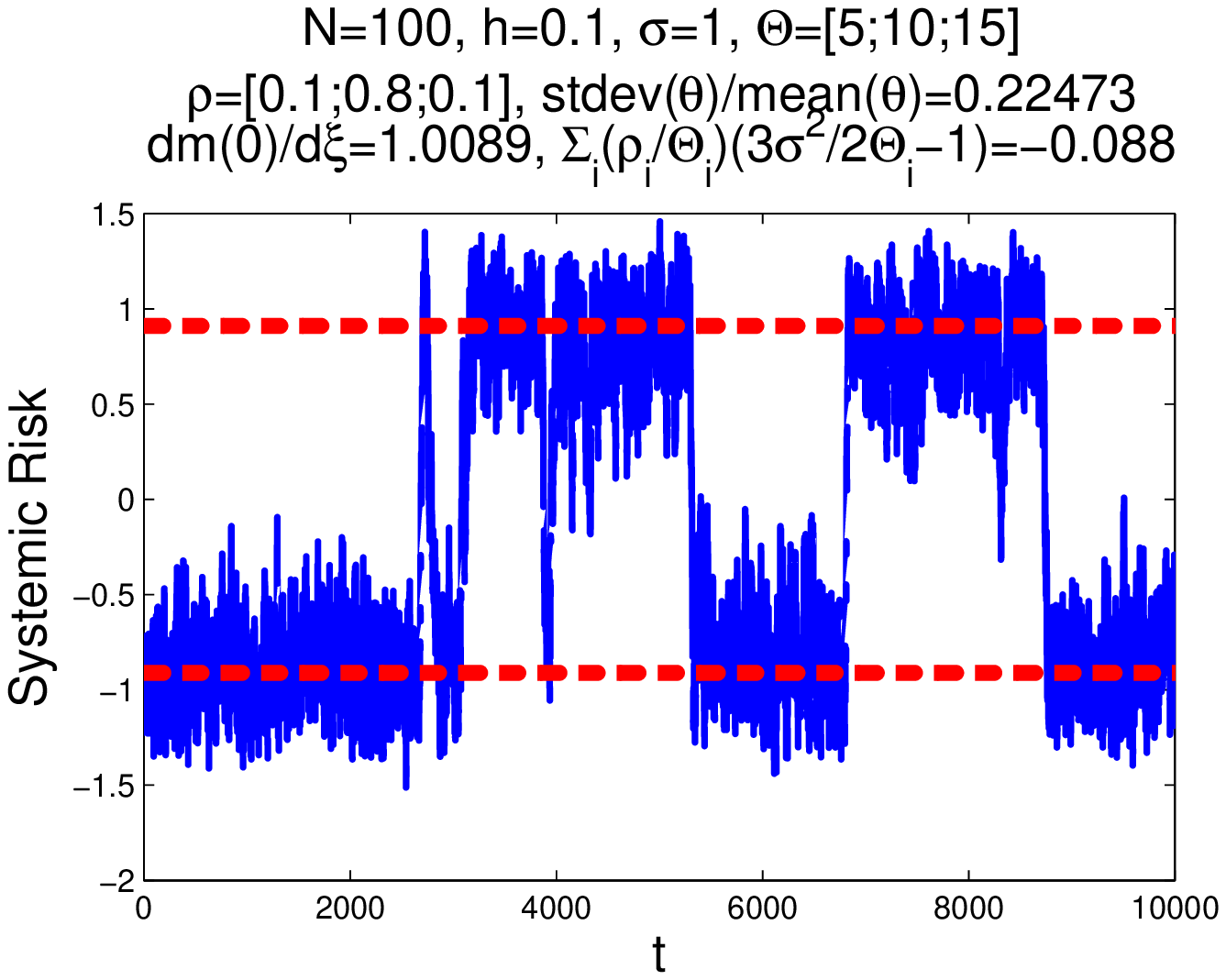}
	\includegraphics[width=0.32\textwidth]{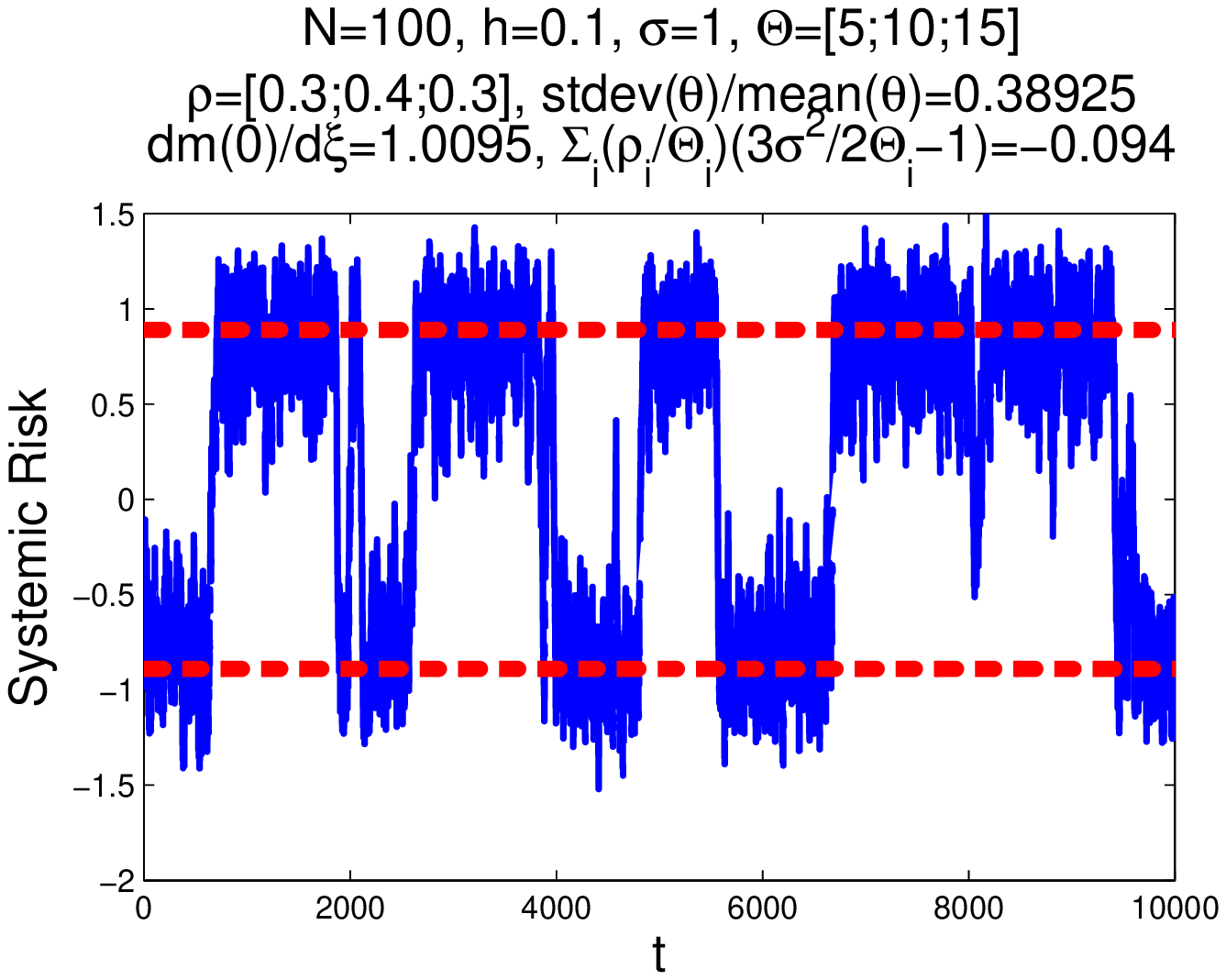}
	\includegraphics[width=0.32\textwidth]{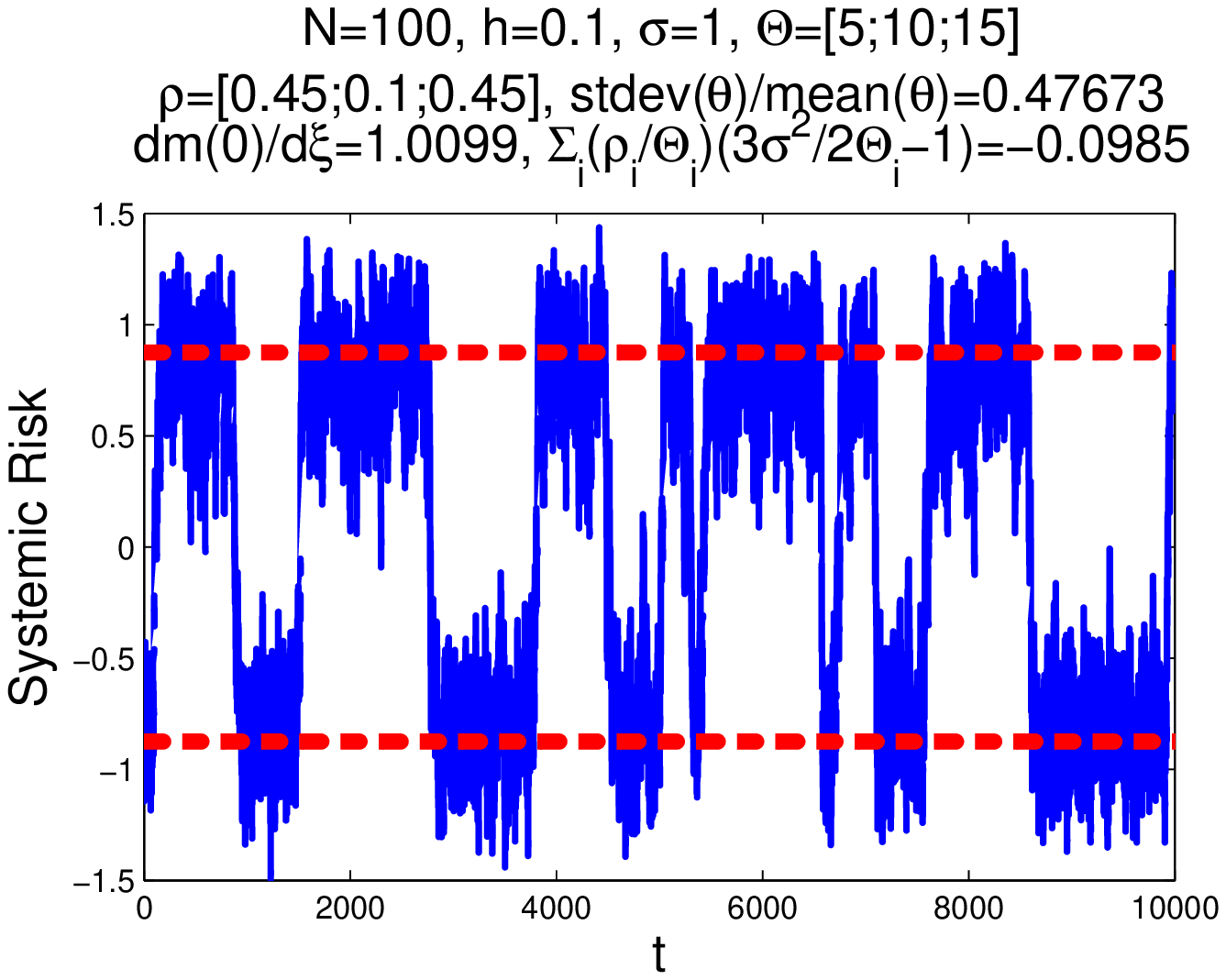}
	\caption{
		\label{fig:The change of rho, diversity case}
		The effect of changes in $\rho_l$, with $\Theta_l$ and the mean of $\theta_j$ fixed. 
		Increasing diversity tends to destabilize the system.}
\end{figure}

We also change the diversity of $\theta_j$ by changing $\Theta_l$ and $\rho_l$. To compare with the 
homogeneous case, in Figure \ref{fig:The change of Theta, diversity case} and Figure 
\ref{fig:The change of rho, diversity case} we change the standard deviation of $\theta_j$ while the mean of 
$\theta_j$ is fixed. In this most interesting part of the simulations we see that when we increase the 
standard deviation of diversity values, the number of transitions is notably larger than that in the 
homogeneous case.
\section{Large Deviations}
\label{sec:large deviations}

In the previous two sections we saw both analytically and numerically
that for large $N$, the empirical density $X_N(t,dy)$ is close (weakly, in probability) to
the solution of the Fokker-Planck equation (\ref{eq: full Fokker-Planck eqn}),
and so the mean $\bar{x}(t)$ in (\ref{eq:SDE of single component})
stays around the first order moment of the deterministic limit,
$\int_{-\infty}^{\infty} y u(t,y)dy$. 
If the condition of existence of two equilibria is satisfied, 
then $\bar{x}(t)$ will remain close to either $-\xi_b$ or $+\xi_b$
for relatively long time intervals, depending in particular on the parameter $h$.
However, as long as $N<\infty$, as we have seen in the simulations 
the random forcing by the Brownian motions $\{w_j(t)\}_{j=1}^N$ will cause transitions 
with non-zero probability.
A systemic transition is the event that $\bar{x}(t)$
is displaced from $\pm \xi_b$ to $\mp \xi_b$ within a finite time horizon.
Thus, systemic transition means that a large number of agents transition in a finite time. 
In this paper, we are interested in computing the probability of such a systemic transition. 
Mathematically, given a finite time horizon $[0,T]$ 
and the conditions for existence of two equilibria, we want to compute the probability
\begin{equation}
	\label{eq:probability of phase transition}
	\mathbf{P}(\bar{x}(0)=-\xi_b, \bar{x}(T)=\xi_b)
\end{equation}
when $N$ is large and as a function of the parameters $(h,\theta,\sigma)$ in (\ref{eq:SDE of single component}).

\subsection{Large Deviations of Mean-fields}

According to \cite{Dawson1987}, we can calculate this probability asymptotically for large $N$ 
using large deviations. To state the large deviations theory that we will use, 
we will review briefly some notation and terminology  from \cite{Dawson1987}.
\begin{itemize}
	\item $M_1(\mathbb{R})$ is the space of probability measures on 
	$\mathbb{R}$ with the Prohorov metric $\rho$, associated with weak convergence.
	
	\item $C([0,T],M_1(\mathbb{R}))$ is the space of continuous functions from $[0,T]$ 
	to $M_1(\mathbb{R})$ with the metric $\sup_{0\leq t\leq T}\rho(\phi_1(t),\phi_2(t))$.

	\item $M_R(\mathbb{R})=\{\mu\in M_1(\mathbb{R}),\int\varphi(y)\mu(dy)\leq R\}$, 
	where $\varphi \in C^2(\mathbb{R})$ is a nonnegative function with 
	$\lim_{|x|\to\infty}\varphi(x)=\infty$. From \cite{Dawson1987}, if $U(y)=y^3-y$, we 
	can choose $\varphi(y)=1+y^2+\gamma y^4$, $0\leq\gamma\leq h/2$.
		
	\item $M_\infty(\mathbb{R})=\cup_{R>0}M_R(\mathbb{R})
	=\{\mu\in M_1(\mathbb{R}), \int\varphi(y)\mu(dy)<\infty\}$ endowed with the 
	inductive topology: $\mu_n\rightarrow\mu$ in $M_\infty(\mathbb{R})$ if and only if 
	$\mu_n\rightarrow\mu$ in $M_1(\mathbb{R})$ and 	
	$\sup_n\int\varphi(y)\mu_n(dy)<\infty$. 
	
	\item $C([0,T],M_\infty(\mathbb{R}))$ is the space of continuous functions from 
	$[0,T]$ to $M_\infty(\mathbb{R})$ endowed with the topology: 
	$\phi_n(\cdot)\rightarrow\phi(\cdot)$ in $C([0,T],M_\infty(\mathbb{R}))$ 
	if and only if $\phi_n(\cdot)\rightarrow\phi(\cdot)$ in $C([0,T],M_1(\mathbb{R}))$ 
	and $\sup_{0\leq t\leq T}\sup_n\int\varphi(y)\phi_n(t,dy)<\infty$. 

	\item Given $\nu\in M_\infty(\mathbb{R})$, we let
	$\mathcal{E}^\nu=\{\phi\in C([0,T],M_\infty(\mathbb{R})): \phi(0)=\nu\}$,
	endowed with the relative topology.
\end{itemize}

To simplify the notation, we rewrite (\ref{eq: full Fokker-Planck eqn}) as 
$u_t = \mathcal{L}_u^* u + h\mathcal{M}^* u$, where 
\[
	\mathcal{L}_\psi^* \phi = \frac{1}{2}\sigma^2 \phi_{yy} 
	+\theta\frac{\partial}{\partial y}\left\{\left[y-\int y\psi(t,y)dy\right]\phi\right\},
	\quad
	\mathcal{M}^*\phi =	 \frac{\partial}{\partial y}\left[U(y)\phi\right].
\]

\begin{theorem}
	\label{thm: Dawson and Gartner, general LD}
	(Dawson and G\"{a}rtner, 1987) Given a finite horizon $\left[0,T\right]$, 
	$\nu\in M_\infty(\mathbb{R})$ and $A \subseteq \mathcal{E}^\nu$, 
	if $X_N(0)=\frac{1}{N}\sum_{j=1}^N\delta_{x_j(0)}\rightarrow \nu$ in  
	$M_\infty(\mathbb{R})$ as $N \rightarrow \infty$, then the law of 
	$X_N(t)=\frac{1}{N}\sum_{j=1}^N\delta_{x_j(t)}$ satisfies the large deviation 
	principle with the good rate function $I_h$:
	\begin{align*}
		-\inf_{\phi \in \mathring{A}}I_h(\phi) 
		& \leq \liminf_{N\rightarrow\infty}\frac{1}{N}\log\mathbf{P}(X_N\in A)\\
		& \leq \limsup_{N\rightarrow\infty}\frac{1}{N}\log\mathbf{P}(X_N\in A)
		\leq - \inf_{\phi \in \bar{A}}I_h(\phi),
	\end{align*}
	where $\mathring{A}$ and $\bar{A}$ are the interior and closure of $A$ in 
	$\mathcal{E}^\nu$, respectively, and 
	\begin{equation}
		\label{def:Ih}
		I_h(\phi) = \frac{1}{2\sigma^2} \int_0^T 
		\sup_{f:\langle \phi, f_y^2 \rangle \neq 0} J_h(\phi,f)dt,
	\end{equation}
	\[
		J_h(\phi,f) = \langle \phi_t - \mathcal{L}_\phi^*\phi 
		- h\mathcal{M}^*\phi, f\rangle^2 / \langle \phi, f_y^2 \rangle, \quad
		\langle \phi,f\rangle = \int_{-\infty}^\infty f(y)\phi(dy),
	\]
	if $\phi(t)$ is absolutely continuous in $t\in[0,T]$ and $I_h(\phi)=\infty$ otherwise.
\end{theorem}

\textbf{Remark.} Here for $\phi \in {\cal E}^\nu$ and $t\in[0,T]$, $\phi(t)$ is viewed as a real Schwartz distribution 
on $\mathbb{R}$, $\mathcal{L}_\psi^*$ and $\mathcal{M}^*\phi$ are differential operators in the distribution 
sense, and $f$ in (\ref{def:Ih}) is a real Schwartz test function. The definition of absolute continuity for 
the path of measures $(\phi(t))_{t \in [0,T]}$ is in the sense of Definition 4.1 in \cite{Dawson1987}, that 
is to say: for each compact set $K \subset \mathbb{R}$ there exists a neighborhood $U_K$ of the null 
function in the set of test functions with compact support in $K$ and an absolutely continuous function 
$H_K$ from $[0,T]$ to $\mathbb{R}$ such that $| \langle \phi(t),f\rangle -	\langle \phi(s),f\rangle | \leq 
|H_K(t)-H_K(s)|$ for all $s,t \in [0,T]$ and $f\in U_K$. Note that by Lemma 4.2 in \cite{Dawson1987}, if 
$\phi(t)$ is absolutely continuous in $t\in[0,T]$, $\phi_t(t)$ exists in the distribution sense almost 
everywhere on $t\in[0,T]$.

In order to use Theorem \ref{thm: Dawson and Gartner, general LD}, we let 
$\nu = u^e_{-\xi_b}$ in (\ref{eq:equilibrium of Fokker-Planck eq}) 
and define the rare event $A$ of systemic transition by
\begin{equation}
	\label{def: the set of phase transition}
	A = \left\{ \phi\in\mathcal{E}^\nu: \phi(T)=u^e_{\xi_b} \right\}.
\end{equation}
However, since $\mathring{A}$ is an empty set, 
Theorem \ref{thm: Dawson and Gartner, general LD} give a  trivial lower bound 
for the probability in question. Therefore we consider instead the closed rare event $A_\delta$:
\[
	A_\delta=\left\{\phi\in\mathcal{E}^\nu: \rho(\phi(T),u^e_{\xi_b})\leq\delta\right\}.
\]
Then Theorem \ref{thm: Dawson and Gartner, general LD} implies that 
\begin{align*}
	-\inf_{\phi \in \mathring{A}_\delta}I_h(\phi) 
	& \leq \liminf_{N\rightarrow\infty}\frac{1}{N}\log\mathbf{P}(X_N\in A_\delta)\\
	& \leq \limsup_{N\rightarrow\infty}\frac{1}{N}\log\mathbf{P}(X_N\in A_\delta)
	\leq - \inf_{\phi \in A_\delta}I_h(\phi).
\end{align*}
In addition, we show that $\inf_{\phi \in A_\delta}I_h(\phi)$ can be bounded from below 
by $\inf_{\phi \in A}I_h(\phi)$ as $\delta \rightarrow 0$.
\begin{lemma}
	\label{lma:lower bound of I_h over A_delta}
	By definition $\inf_{\phi\in A_\delta}I_h(\phi)$ is decreasing with $\delta>0$ and 
	bounded from above by $\inf_{\phi\in A}I_h(\phi)$. In addition,
	\[
		\lim_{\delta\to 0}\inf_{\phi\in A_\delta}I_h(\phi)\geq\inf_{\phi\in A}I_h(\phi).
	\]
\end{lemma}
\begin{proof}
	See Appendix \ref{pf:lower bound of I_h over A_delta}.
\end{proof}

Combining Lemma \ref{lma:lower bound of I_h over A_delta} and the fact that 
$\inf_{\phi \in \mathring{A}_\delta}I_h(\phi) \leq \inf_{\phi \in A}I_h(\phi)$, for any 
$\epsilon>0$, we have for sufficiently small $\delta>0$ 
\begin{align*}
	-\inf_{\phi \in A}I_h(\phi)
	& \leq \liminf_{N\rightarrow\infty}\frac{1}{N}\log\mathbf{P}(X_N\in A_\delta)\\
	& \leq \limsup_{N\rightarrow\infty}\frac{1}{N}\log\mathbf{P}(X_N\in A_\delta)
	\leq - \inf_{\phi \in A}I_h(\phi) + \epsilon.
\end{align*}
Therefore for large $N$ and sufficiently small $\delta$,
\begin{equation}
	\label{eq:probability of phase transition, LD form}
	\mathbf{P}(X_N\in A_\delta) \approx \exp\left(-N\inf_{\phi\in A}I_h(\phi)\right).
\end{equation}
This tells us that a larger system has a more stable empirical mean
trajectory, which is consistent with what we have seen in the numerical
simulation. Now the main step is finding $\inf_{\phi\in A}I_h(\phi)$,
which is a min-max variational problem
\begin{equation}
	\label{eq:general LD}
	\inf_{\phi \in A} I_h(\phi)
	= \inf_{\phi \in A} \frac{1}{2\sigma^2} \int_0^T 
	\sup_{f:\langle \phi, f_y^2 \rangle \neq 0}
	\langle \phi_t - \mathcal{L}_\phi^*\phi - h\mathcal{M}^*\phi, f\rangle^2
	/ \langle \phi, f_y^2 \rangle dt,
\end{equation}
where the $f$ in the $\sup$ is a  real Schwartz test function.

\subsection{An Alternative Expression for the Rate Function}

The representation of the rate function (\ref{def:Ih}) is somewhat complicated, but we 
can simplify it if $\phi$ has the density with some additional properties.
If $\phi$ is a density function such that $\phi(t,y)$ is smooth, rapidly decreasing in 
$y\in\mathbb{R}$ for each $t\in[0,T]$ and is absolutely continuous in $t\in[0,T]$ for 
each $y\in\mathbb{R}$, then let $g(t,y)$ satisfy
\begin{equation}
	\label{eq:Fokker-Planck equation with the driving force}
	\phi_t - \mathcal{L}_{\phi}^{*}\phi - h\mathcal{M}^{*}\phi = (\phi g)_y.
\end{equation}
Note that because of the properties of $\phi$, the left hand side of 
(\ref{eq:Fokker-Planck equation with the driving force}) is well-defined in 
$y\in\mathbb{R}$ and almost everywhere in $t\in[0,T]$. In addition, because $\phi$ is 
positive valued, $g$ exists and is unique except on a measure zero set in $[0,T]$.

%

Note that for the pair $(\phi,g)$ satisfying 
(\ref{eq:Fokker-Planck equation with the driving force})
\[
	\sup_{f:\langle \phi(t), f_y^2 \rangle \neq 0} J_h(\phi(t),f)
	= \sup_{f:\langle \phi(t), f_y^2 \rangle \neq 0} 
	\langle\phi(t),f_y g\rangle^2/\langle\phi(t),f_y^2\rangle
	= \langle \phi(t), g^2\rangle,
\]
and therefore we have the following proposition.

\begin{proposition}
	\label{prop:I_h in terms of the driving force}
	If $\phi$ is a density function such that $\phi(t)$ is a Schwartz function for each 
	$t\in[0,T]$ and is absolutely continuous in $t\in[0,T]$ for each $y\in\mathbb{R}$, 
	and $g(t,y)$ satisfies (\ref{eq:Fokker-Planck equation with the driving force}), the 
	rate function $I_h(\phi)$ in (\ref{def:Ih}) can be written in the form
	\begin{equation}
		\label{eq:I_h in terms of the driving force}
		I_h(\phi) = \frac{1}{2\sigma^2} \int_0^T \langle \phi, g^2\rangle dt.
	\end{equation}
\end{proposition}

We interpret (\ref{eq:Fokker-Planck equation with the driving force}) and 
(\ref{eq:I_h in terms of the driving force}) as follows. The function $g$ is regarded as 
the driving force making $\phi$ deviate from the solution of the Fokker-Planck equation 
(\ref{eq: full Fokker-Planck eqn}), and $I_h(\phi)$ is the $L^2(\phi)$ norm of $g$, 
which measures how difficult it is to have this deviation $\phi$. 

\section{Small $h$ Analysis}
\label{sec:small h}

The goal of this section is to analyze the min-max problem (\ref{eq:general LD}) which
controls the asymptotic systemic transition probability. 
This problem is nonlinear and infinitely dimensional and is 
difficult to analyze. To get some useful information about it
we will assume that $h$ is small and analyze it in this regime. 
We will first solve (\ref{eq:general LD}) when $h$ is exactly $0$, and 
then we will get rigorous upper and lower bounds for (\ref{eq:general LD}) when $h$ is nonzero 
but small. We will then compare the large deviations result with the local
fluctuation theory of a single agent so as to explain why interconnectedness destabilizes 
the system. 


\subsection{The $h=0$ and the Small $h$ Analysis}
\label{sec:small h analysis}

We note that when $h=0$, $u_{\pm\xi_b}^e = u_{\pm\xi_0}^e$, where
\begin{equation}
	u_{\pm\xi_0}^e(y) = \frac{1}{\sqrt{2\pi \frac{\sigma^2}{2\theta}}}
	\exp\left\{ -\frac{(y - (\pm\xi_0))^2}{2\frac{\sigma^2}{2\theta}} \right\},
	\quad\quad \xi_0 = \sqrt{1-3\frac{\sigma^2}{2\theta}}.
\end{equation}
In this case, (\ref{eq:general LD}) is solvable and the optimal path is a Gaussian, 
starting from $u_{-\xi_0}^e$ and ending in $u_{+\xi_0}^e$.
\begin{theorem}
	\label{thm:LD for zero h and Gaussian path}
	Let $h=0$ and define 
	\begin{equation}
		\label{def:pe}
		p^e(t,y) = \frac{1}{\sqrt{2\pi \frac{\sigma^2}{2\theta}}}
		\exp\left\{ -\frac{(y - a^e(t))^2}{2\frac{\sigma^2}{2\theta}} \right\}, \quad
		a^e(t)= \frac{2\xi_0}{T} t - \xi_0.
	\end{equation}
	Then $p^e \in A$ is the unique minimizer for (\ref{eq:general LD}) and
	\[
		\inf_{\phi \in A} I_0(\phi) = I_0(p^e) = \frac{2\xi_0^2}{\sigma^2 T}.
	\]
\end{theorem}
\begin{proof}
	See Appendix \ref{pf:LD for zero h and Gaussian path}.
\end{proof}

We show next that (\ref{eq:general LD}) is continuous at $h=0$. 
\begin{theorem}
	\label{thm:cotinuity for small h}
	There exists $\gamma(h)$ such that $\gamma(h)\to 0$ as $h \to 0$ and
	\begin{equation}
		\label{eq:cotinuity for small h}
		\Big| \inf_{\phi \in A} I_h(\phi) - \frac{2\xi_b^2}{\sigma^2 T} \Big| \leq \gamma(h).
	\end{equation}
	We recall here that 
	\begin{equation}
		\xi_b = \xi_0+h \xi_1 +O(h^2), \quad \quad \xi_1 =  \sqrt{1 - 3\frac{\sigma^2}{2\theta}} 
		 \frac{6}{\sigma^2}\left(\frac{\sigma^2}{2\theta}\right)^2
		\frac{1 - 2(\sigma^2/2\theta)}{1 - 3(\sigma^2/2\theta)} .
	\end{equation}
\end{theorem}

\begin{proof}
	See Appendix \ref{pf:cotinuity for small h, upper bounds} and 
	\ref{pf:cotinuity for small h, lower bounds}.
\end{proof}

As it is stated we could replace $\xi_b$ by $\xi_0$ in Theorem \ref{thm:cotinuity for small h}, 
since $\xi_b = \xi_0+o(1)$ as $h \to 0$. We will see in the next section 
(in Proposition \ref{prop:inf_A I_h up to O(h)})  that $\gamma(h)=O(h^2)$. In fact we show this rigorously 
for the upper bound but only formally for the lower bound. Since $\xi_b=\xi_0+h\xi_1+O(h^2)$ we see that the 
term $2\xi_b^2/(\sigma^2 T)$ contains the leading-order term and the first-order correction in the 
$h$-expansion of $\inf_{\phi \in A} I_h(\phi)$.

\subsection{Large Deviations for the First Exit Time}

In this subsection, we consider the rare event $B$ of systemic transition at some time before $T$:
\[
	B_\delta = \{\phi\in\mathcal{E}^\nu: \exists t\in(0,T], \rho(\phi(t),u^e_{\xi_b}) \leq \delta \}.
\]
In other words, $B_\delta=\cup_{t\in(0,T]}A_\delta(t)$, where 
\[
	A_\delta(t) = \{\phi\in\mathcal{E}^\nu: \rho(\phi(t),u^e_{\xi_b}) \leq \delta \}.
\]
We let $B:=B_0$. We then have that 
\begin{lemma}
	\label{lma:lower bound of I_h over B_delta}
	By definition $\inf_{\phi\in B_\delta}I_h(\phi)$ is decreasing with $\delta>0$ and bounded from above by 
	$\inf_{\phi\in B}I_h(\phi)$. In addition,
	\[
		\lim_{\delta\to 0}\inf_{\phi\in B_\delta}I_h(\phi) 
		= \inf_{\phi\in \cup_{t\in(0,T]}A(t)}I_h(\phi)
		= \inf_{\phi\in B}I_h(\phi),
	\]
	where $A(t):=A_0(t)$.
\end{lemma}
\begin{proof}
	See Appendix \ref{pf:lower bound of I_h over B_delta}.
\end{proof}

From Theorem \ref{thm:cotinuity for small h}, we see that in the sense of large deviations the 
probability of system failure at some time before time $T$ is essentially the same as the probability of 
system failure at time $T$.
\begin{corollary}
	For any $t_1<t_2$, there exists a sufficiently small $h$ such that 
	$\inf_{\phi\in A(t_1)} I_h(\phi)> \inf_{\phi\in A(t_2)}I_h(\phi)$. Consequently, 
	$\inf_{\phi\in B}I_h(\phi)\approx \inf_{\phi\in A(T)}I_h(\phi)$ for small $h$.
\end{corollary}

\subsection{Comparison with the Fluctuation Theory of a Single Agent}

To get a better understanding of the large deviations results we need to
carry out a standard fluctuation theory for a single agent. 
We assume that $x_j(0)=-1$ for all $j$ and that the
$x_j(t)$'s are in the vicinity of $-1$ so that we can linearize 
(\ref{eq:SDE of single component}):
\[
	x_j(t)=-1+z_j(t),\quad \bar{x}(t)=-1+\bar{z}(t),\quad 
	\bar{z}(t) = \frac{1}{N}\sum_{j=1}^N z_j(t).
\]
For $V(y)=\frac{1}{4}y^4-\frac{1}{2}y^2$, $z_j(t)$ and $\bar{z}(t)$ satisfy 
the linear stochastic differential equations
\[
	dz_j = -(\theta+2h)z_j dt + \theta\bar{z}dt + \sigma dw_j, \quad
	d\bar{z} = -2h\bar{z}dt + \frac{\sigma}{N}\sum_{j=1}^N dw_j,
\]
with $z_j(0)=\bar{z}(0)=0$. The processes $z_j(t)$ and $\bar{z}(t)$ are 
Gaussian and the mean and variance functions are easily calculated. We are 
especially interested in their behavior for large $N$.
\begin{lemma}
	\label{lma:fluctuation of a single agent}
	For all $t\geq 0$, $\mathbf{E}z_j(t)=\mathbf{E}\bar{z}(t)=0$ and 
	$\mathbf{Var}\bar{z}(t)=\frac{\sigma^2}{N}(1-e^{-4ht})$. In addition, 
	$\mathbf{Var}z_j(t) \rightarrow 
	\frac{\sigma^2}{2(\theta+2h)}(1-e^{-2(\theta+2h)t})$ as $N\rightarrow\infty$,
	uniformly in $t\geq 0$.
\end{lemma}

From Lemma \ref{lma:fluctuation of a single agent}, we see that $\sigma^2/N$ and 
$\sigma^2/2(\theta+2h)$ should be sufficiently small so that linearization is 
consistent with the results it produces.

\subsection{Increased Probability of Large Deviations for Increased $\theta$ and Its
Systemic Risk Interpretation}
\label{sec:comparison}

We have now the analytical results with which we
may conclude that individual risk diversification may increase the systemic risk. Assume that 
$\sigma^2/N$ and $\sigma^2/2(\theta+2h)$ are sufficiently small and $N$ is large. From 
Lemma \ref{lma:fluctuation of a single agent}, the risk $x_j(t)$ of the agent $j$ 
is approximately a Gaussian process with the stationary distribution 
$\mathcal{N}(-1,\sigma^2/2(\theta+2h))$. If the external risk, $\sigma$ is high, then in 
order to keep the risk $x_j(t)$ at an acceptable level, the agent may increase the 
intrinsic stability, $h$, or share the risk with other agents, that is, increase
$\theta$. Increasing $h$ is in general more costly (cuts into profits) than increasing 
$\theta$, and at the individual agent level there is no difference in risk assessment 
between increasing 
$h$ and increasing $\theta$. Therefore the agents are likely to increase $\theta$ and
reduce individual risk by diversifying it. Note that 
$\sigma^2/2(\theta+2h) \lesssim \sigma^2/2\theta$ when $\sigma^2$ and $\theta$ are 
significantly larger than $h$. Thus, individual agents can maintain 
low locally assessed risk by diversification, even in a very uncertain environment.

What is not perceived by the individual agents, however,
is that risk diversification may increases the systemic risk while it reduces
their individual risk. Because $\sigma^2$ and $\theta$ are significantly larger than $h$, the 
small $h$ analysis can be applied and from 
(\ref{eq:probability of phase transition, LD form}) and Theorem 
\ref{thm:cotinuity for small h}, the systemic risk (the probability of the system 
failure) is 
\begin{align*}
	&\mathbf{P}(X_N\in B_\delta) \approx \exp\left(-N\frac{2\xi_b^2}{\sigma^2 T}\right), \quad \text{for small $\delta$ and $h$,}\\
	&\xi_b = \sqrt{1 - 3\frac{\sigma^2}{2\theta}} 
			\left( 1 + h\frac{6}{\sigma^2}\left(\frac{\sigma^2}{2\theta}\right)^2
			\frac{1 - 2(\sigma^2/2\theta)}{1 - 3(\sigma^2/2\theta)}\right) + O(h^2).
\end{align*}
We see that there are additional systemic-level $\sigma^2$ terms in the exponent and 
$\xi_b$, which can not be observed by the agents, increasing the systemic risk, even if 
the individual risk $\sigma^2/2\theta$ is fixed. In other words, the individual agents may 
believe that they are able to withstand larger external fluctuations as long as their risk can be 
diversified, but a higher $\sigma$ tends to destabilize the system.

\section{A Reduced Large Deviations Principle for Small $h$}
\label{sec:reduced LDP}

In Section \ref{sec:small h analysis}, we show that the large deviation problem $\inf_{\phi\in A}I_h(\phi)$ 
is continuous in $h$ so that we have the upper and lower bounds for $\inf_{\phi\in A}I_h(\phi)$ when $h$ is 
small. In this section, we analyze with a formal expansion the optimal path for $\inf_{\phi\in A}I_h(\phi)$ 
by assuming that it is of the form $p^e + O(h)$, motivated by the fact
that the optimal path is $p^e$ for $h=0$. In this way, 
we can obtain a  reduced large deviations principle (a reduced Freidlin-Wentzell theory) for the systemic 
risk. That is, we obtain a reduced rate function corresponding to a finite dimensional system after ignoring 
higher order terms. The reduced rate function has all relevant information up to $O(h^2)$ 
terms, and we also need to expand $\phi$ to $O(h^2)$.


We assume that the optimal $\phi = p + hq^{(1)} + h^2 q^{(2)} + \ldots$, where 
\[
	p(t,y) = \frac{1}{\sqrt{2\pi \frac{\sigma^2}{2\theta}}}
	\exp\left\{ -\frac{(y - a(t))^2}{2\frac{\sigma^2}{2\theta}} \right\},
	\quad a(t) = \langle \phi, y \rangle.
\]
In other words, we let the first moment of $\phi$ be determined by $a(t)$, and from the zero $h$ case we 
know that $a(t)=a^e(t)+O(h)$. From the form of $p$ and 
(\ref{eq:Fokker-Planck equation with the driving force}), a natural parameterization for $q^{(1)}$ and $q^{(2)}$ is 
the Hermite expansion
\[
	q^{(1)}(t,y) = \sum_{n=2}^{\infty}b_n(t)\frac{\partial^n}{\partial y^n}p(t,y), \quad
	q^{(2)}(t,y) = \sum_{n=2}^{\infty}c_n(t)\frac{\partial^n}{\partial y^n}p(t,y).
\]
Note that by the properties of $p$ and $a(t)$, $\langle q^{(1)},y^n\rangle = \langle q^{(2)},y^n\rangle = 0$ 
for $n=0,1$ so we can start the Hermite expansion from $n=2$.

The formal expansion result of this section is that if the optimal $\phi = p + hq^{(1)} + h^2 q^{(2)}$, then 
\begin{equation}
	\label{eq:reduced Freidlin-Wentzell}
	\inf_{\phi\in A}I_h(\phi) \approx 
	\inf_{\substack{a(t): 0\leq t\leq T\\a(0)=-\xi_b\\a(T)=\xi_b}} 
	\frac{1}{2\sigma^2} \int_0^T \left(\frac{d}{dt}a + h (a^3 + 3\frac{\sigma^2}{2\theta}a - a)\right)^2 dt,
\end{equation}
for small $h$. Note that $a(t)=\langle \phi, y \rangle=\bar{x}(t)$. The right hand side of 
(\ref{eq:reduced Freidlin-Wentzell}) is an one-dimensional variational problem that has the form of a rate 
function of the Freidlin-Wentzell theory. 
In fact, the right side of (\ref{eq:reduced Freidlin-Wentzell}) is the large deviations variational problem for
the rate function of
the small-noise stochastic differential equation
\begin{equation}
	\label{eq:approximate sde for x_bar}
	d\bar{x}(t) = -h\left[\bar{x}^3(t) - \left(1-\frac{3\sigma^2}{2\theta}\right)\bar{x}(t)\right] dt 
	+ \epsilon \sigma dw(t)
\end{equation}
where here $\epsilon =1/\sqrt{N}$ is small. Note that ${3\sigma^2}/{2\theta}<1$, as assumed above, and 
therefore (\ref{eq:approximate sde for x_bar}) also represents a bi-stable structure. In the remainder of 
this section we describe how this result is obtained by formal expansions and then in Section 
\ref{sec:ldprob_reduced} we show how we recover from (\ref{eq:reduced Freidlin-Wentzell}) the main result of 
the paper stated in the previous section.

An important remark about the expansion is that the Hermite functions are a basis of the $L^2$ space and 
thus $p + hq^{(1)} + h^2 q^{(2)}$ is generally a signed measure. However, if $q^{(1)}$ and $q^{(2)}$ can be 
expressed as the linear combinations of finite Hermite functions, then we can see that for any $\epsilon>0$, 
there exists a sufficiently small $h$ such that the negative part of $p + hq^{(1)} + h^2 q^{(2)}$ is less 
than $\epsilon$.

\subsection{Optimization over $g$}

The first step in finding the optimal $\phi = p + h q^{(1)} + h^2 q^{(2)}$ is determining 
the optimal $g$ by using (\ref{eq:Fokker-Planck equation with the driving force}) 
for $\phi$. Once we obtain $g$, we can compute $I_h(\phi)$ by using 
(\ref{eq:I_h in terms of the driving force}). It is also natural to assume that 
$g = g^{(0)} + h g^{(1)} + h^2 g^{(2)}$ along with the Hermite expansion:
\[
	g^{(0)} = p^{-1}\sum_{n=0}^{\infty}\alpha_n(t)\frac{\partial^n}{\partial y^n}p, \quad
	g^{(1)} = p^{-1}\sum_{n=0}^{\infty}\beta_n(t) \frac{\partial^n}{\partial y^n}p, \quad
	g^{(2)} = p^{-1}\sum_{n=0}^{\infty}\gamma_n(t)\frac{\partial^n}{\partial y^n}p.
\]
In addition, since $\langle q^{(1)},y\rangle=\langle q^{(2)},y\rangle=0$, we can see that 
$\phi = p + h q^{(1)} + h^2 q^{(2)}$ satisfies
\[
	\mathcal{L}_{\phi}^{*}\phi = \mathcal{L}_{p}^{*}p +h\mathcal{L}_{p}^{*}q^{(1)} 
	+ h^2\mathcal{L}_{p}^{*}q^{(2)}, \quad
	\mathcal{M}^{*}\phi = \mathcal{M}^{*}p + h\mathcal{M}^{*}q^{(1)} + h^2\mathcal{M}^{*}q^{(2)}.
\]
The force $U(y)=y^3-y$ can also be expanded in Hermite polynomials:
\[
	U(y) = p^{-1}\sum_{n=0}^3 \delta_n(t)\frac{\partial^n}{\partial y^n}p.
\]
Now everything is expanded in the orthogonal basis and we can find the optimal $g^{(0)}$ and 
$g^{(1)}$ by putting everything into 
(\ref{eq:Fokker-Planck equation with the driving force}) and comparing coefficients.
\begin{lemma}
	\label{lma:the optimal g^0 and g^1}
	With the expansions mentioned above, the optimal $g^{(0)}$ is $-\frac{d}{dt}a$, 
	and the optimal $\beta_n$ for $g^{(1)}$ are 
	\begin{equation}
		\label{eq:beta_n}
		\beta_n =
		\begin{cases}
			-\delta_0 = -\langle p, U(y) \rangle,		& n=0,\\
			\frac{d}{dt}b_{n+1} + \theta(n+1)b_{n+1} - \delta_n,	& 1\leq n\leq 3,\\
			\frac{d}{dt}b_{n+1} + \theta(n+1)b_{n+1},				& n\geq 4.
		\end{cases}
	\end{equation}
\end{lemma}
\begin{proof}
	See Appendix \ref{pf:the optimal g^0 and g^1}.
\end{proof}

It remains to determine $g^{(2)}$. 
From (\ref{eq:I_h in terms of the driving force}) we see that the only contribution of 
$g^{(2)}$ to $I_h$ up to $O(h^2)$ is 
$\langle p, 2g^{(0)} g^{(2)} \rangle = -2\gamma_0\frac{d}{dt}a$. 
Thus it suffices to determine $\gamma_0$, which can also be obtained from 
(\ref{eq:Fokker-Planck equation with the driving force}).
\begin{lemma}
	\label{lma:the optimal g^2}
	With the expansions mentioned above, the optimal $\gamma_0$ is 
	\[ 
		\gamma_0 = -\langle q^{(1)}, U(y)+g^{(1)} \rangle. 
	\]
\end{lemma}
\begin{proof}
	See Appendix \ref{pf:the optimal g^2}.
\end{proof}

\subsection{Optimization over $\phi$}

We are now ready to find the optimal $\phi$. For given $\phi = p + hq^{(1)} + h^2q^{(2)}$ and 
the corresponding optimal $g = g^{(0)} + hg^{(1)} + h^2g^{(2)}$, 
(\ref{eq:I_h in terms of the driving force}) gives
\begin{align*}
	I_h(\phi) 
	&= \frac{1}{2\sigma^2} \int_0^T \langle p+hq^{(1)}+h^2q^{(2)}, (g^{(0)} + hg^{(1)} + h^2g^{(2)})^2 \rangle dt\\
	&= \frac{1}{2\sigma^2} \int_0^T \langle p, (g^{(0)})^2 \rangle dt
	+ \frac{h}{2\sigma^2} \int_0^T \langle p, 2 g^{(0)} g^{(1)} \rangle dt\\
	&\quad + \frac{h^2}{2\sigma^2} \int_0^T \left(\langle p, (g^{(1)})^2 + 2 g^{(0)} g^{(2)} \rangle
	+ \langle q^{(1)}, 2 g^{(0)} g^{(1)} \rangle \right)dt + O(h^3).
\end{align*}
From Lemma \ref{lma:the optimal g^2}, $\langle p, 2 g^{(0)} g^{(2)} \rangle 
= -2g^{(0)}\langle q^{(1)}, U(y)+g^{(1)} \rangle$, and therefore
\[
	\langle p, 2 g^{(0)} g^{(2)} \rangle
	+ \langle q^{(1)}, 2 g^{(0)} g^{(1)} \rangle
	= -2g^{(0)} \langle q^{(1)}, U(y) \rangle
	= -2g^{(0)} \sum_{n=2}^3 H_n \delta_n b_n, 
\]
where $H_n(t) := \langle p^{-1}, (\partial^n p /\partial y^n)^2 \rangle$.
We note that
\[
	\langle p, 2 g^{(0)} g^{(1)} \rangle = -2g^{(0)}\delta_0, \quad 
	\langle p,(g^{(1)})^2 \rangle 
	= \delta_0^2+\sum_{n=1}^\infty H_n \beta_n^2, \quad
	\langle p, (g^{(0)})^2 \rangle = (g^{(0)})^2.
\]
Then $I_h(\phi)$ can be written as
\begin{align}
	\label{eq:expansion of I_h}
	I_h(\phi) &= \frac{1}{2\sigma^2} \int_0^T (g^{(0)} - h\delta_0)^2 dt
	+ \frac{h^2}{2\sigma^2} \int_0^T (H_1\beta_1^2 - 2H_2 g^{(0)}\delta_2 b_2) dt\\
	&\quad + \frac{h^2}{2\sigma^2} \int_0^T (H_2\beta_2^2 - 2H_3 g^{(0)}\delta_3 b_3) dt \notag
	+ \frac{h^2}{2\sigma^2} \sum_{n=3}^\infty \int_0^T H_n\beta_n^2 dt + O(h^3).
\end{align}
We see that $a$ and $b_n$ are coupled at the $O(h^2)$ level of 
(\ref{eq:expansion of I_h}). However, from the results of the zero $h$ case, 
$a = a^e + O(h)$ and $p = p^e + O(h)$ so we can decouple $a$ and $b_n$ and express 
the expanded $I_h(\phi)$ up to $O(h^2)$ as the sum of independent terms.
\begin{proposition}
	To order $O(h^2)$, the rate function $I_h(\phi)$ can be written as the sum of 
	independent terms:
	\begin{align}
		\label{eq:I_h as the sum of independent terms}
		I_h(\phi) 
		&= \frac{1}{2\sigma^2} \int_0^T (g^{(0)}-h\delta_0)^2 dt
		+ \frac{h^2}{2\sigma^2} \int_0^T (\tilde{H}_1\tilde{\beta}_1^2 
		+ 2\frac{d}{dt}a^e \tilde{H}_2 \tilde{\delta}_2 b_2) dt\\
		&\quad + \frac{h^2}{2\sigma^2} \int_0^T (\tilde{H}_2\tilde{\beta}_2^2 
		+ 2\frac{d}{dt}a^e \tilde{H}_3 \tilde{\delta}_3 b_3) dt \notag
		+ \frac{h^2}{2\sigma^2} \sum_{n=3}^\infty \int_0^T \tilde{H}_n\tilde{\beta}_n^2 dt + O(h^3),
	\end{align}
	where $\tilde{H}_n(t)=\langle (p^e)^{-1}, 
	(\partial^n p^e /\partial y^n)^2 \rangle$, 
	$U(y)=(p^e)^{-1}\sum_{n=0}^3\tilde{\delta}_n(t)\frac{\partial^n}{\partial y^n}p^e$, 
	and 
	\begin{equation}
		\label{eq: tilde_beta_n}
		\tilde{\beta}_n =
		\begin{cases}
			-\tilde{\delta}_0 = -\langle p^e, U(y) \rangle, & n=0,\\
			\frac{d}{dt}b_{n+1}+\theta(n+1)b_{n+1}-\tilde{\delta}_n,   & 1\leq n\leq 3,\\
			\frac{d}{dt}b_{n+1} + \theta(n+1)b_{n+1},				   & n\geq 4.
		\end{cases}
	\end{equation}
\end{proposition}

We can see from (\ref{eq:I_h as the sum of independent terms}) that $q^{(2)}$ does not 
appear in terms up to $O(h^2)$.
From the $h$ expansion of $u^e_{\pm\xi_b}$ in (\ref{eq:equilibrium of Fokker-Planck eq}), and the fact that 
$V(y)$ is a polynomial of degree four, we have $b_{n+1}(0)=b_{n+1}(T)=0$ for $n\geq 4$. The variational 
problem for $b_{n+1}$ is to minimize $\int_0^T\tilde{H}_n \tilde{\beta}_n^2 dt$ where $\tilde{\beta}_n$ is 
given in terms of $b_{n+1}$ by (\ref{eq: tilde_beta_n}). The obvious solution of this problem is $b_{n+1}=0$ 
and $\tilde{\beta}_n=0$ for $n \geq 4$. Consequently, in order to find the optimal $\phi$ for $I_h(\phi)$ in 
(\ref{eq:I_h as the sum of independent terms}), we may solve separately the variational problems for $a$, 
$b_1$, $b_2$ and $b_3$.

\subsection{Probability of Systemic Transitions for Small $h$}
\label{sec:ldprob_reduced}

We consider the small probability of systemic  transitions for large $N$ and small $h$ 
through the large deviation $\inf_{\phi\in A}I_h(\phi)$. Here we consider the solution up to $O(h)$ terms. 
That is, using (\ref{eq:I_h as the sum of independent terms}), 
we solve the variational problem for $a(t)$:
\begin{equation}
	\label{eq:variational problem of a}
	\inf_{\substack{a(t): 0\leq t\leq T\\a(0)=-\xi_b\\a(T)=\xi_b}} 
	\int_0^T (g^{(0)}-h\delta_0)^2 dt
	= \inf_{\substack{a(t): 0\leq t\leq T\\a(0)=-\xi_b\\a(T)=\xi_b}} 
	\int_0^T (\frac{d}{dt}a + h (a^3 + 3\frac{\sigma^2}{2\theta}a - a))^2 dt.
\end{equation}
By simple calculus of variations methods we find the optimal $a$.
\begin{lemma}
	\label{lma:ODE for the optimal a}
	The optimal $a(t)$ for (\ref{eq:variational problem of a}) satisfies 
	the second order ordinary differential equation
	\[
		\frac{d^2}{dt^2}a = h^2 (a^3 + (3\frac{\sigma^2}{2\theta}-1)a) 
		(3a^2+ (3\frac{\sigma^2}{2\theta}-1))
	\]
	with $a(0)=-\xi_b$ and $a(T)=\xi_b$. Consequently, the optimal path is 
	\begin{equation}
		\label{eq:optima a(t)}
		a(t) = \frac{2\xi_b}{T}t - \xi_b + O(h^2).
	\end{equation}
\end{lemma}

By inserting (\ref{eq:optima a(t)}) into (\ref{eq:variational problem of a}) we
obtain $\inf_{\phi\in A}I_h(\phi)$ up to $O(h)$.
\begin{proposition}
	\label{prop:inf_A I_h up to O(h)}
	For small $h$, the large deviations problem, $\inf_{\phi\in A}I_h(\phi)$, 
	up to $O(h)$, is 
	\begin{equation}
		\label{eq:inf_A I_h up to O(h)}
		\inf_{\phi\in A}I_h(\phi) = \frac{2\xi_0}{\sigma^2 T}(\xi_0 + 2h\xi_1) + O(h^2),
	\end{equation}
	where $\xi_b = \xi_0 + h \xi_1 + O(h^2)$ from 
	(\ref{eq:explicit solution for xi equal to m(xi)}). 
	Note that $\xi_1$ is positive because $2\theta > 3\sigma^2$.
\end{proposition}
\begin{proof}
	See Appendix \ref{pf:inf_A I_h up to O(h)}.
\end{proof}

The asymptotic probability of systemic transition for large $N$ and sufficiently small 
$\delta$ and $h$ has the form 
\[
	\mathbf{P}(X_N\in A_\delta) 
	\approx \exp\left(-N\inf_{\phi\in A}I_h(\phi)\right)
	= \exp\left(-N
	 \left\{ \frac{2\xi_0}{\sigma^2 T}(\xi_0 + 2h\xi_1) + O(h^2) \right\} \right).
\]
\section{Effect of Diversity of Sensitivities on the Transition Probability}
\label{sec:diversity2}

We consider the situation introduced in Section \ref{sec:diversity} and analyze it when $h=0$. 
We aim at computing the transition probability in this situation.
The $K$ partial empirical averages
\begin{equation}
\label{eq:systembarxk}
\bar{x}_k (t) := \frac{1}{|\mathcal{I}_{k}|} \sum_{j \in \mathcal{I}_{k}} x_j(t), \quad k=1,\ldots,K 
\end{equation}
then satisfy a closed system of stochastic differential equations
\begin{equation}
d\bar{x}_k = \frac{\sigma}{\sqrt{\rho_k N}} d\bar{w}_k(t) - \theta_k ( \bar{x}_k -\bar{x} ) dt
\end{equation}
where $\bar{w}_k$ are independent  Brownian  motions
and  the empirical mean $\bar{x}(t)$ can be expressed in terms of the partial averages as
$$
\bar{x}(t) = \sum_{k=1}^K \rho_k \bar{x}_k (t)
$$
\begin{proposition}
\label{prop:xbar_diversity}
If $\bar{x}_k(0)= -\xi_b$ for all $k=1,\ldots,K$, then $\bar{x}(T)$ is a Gaussian random variable with mean $-\xi_b$ and variance
$\sigma_T^2:={\rm Var}(\bar{x}(T)) $ given by 
\begin{equation}
\label{eq:sigmaT2}
\sigma_T^2 = \frac{\sigma^2}{N} \int_0^T \varrho^\mathbf{T} e^{Ms} R^{-1} (e^{Ms})^\mathbf{T} \varrho  ds
\end{equation}
where $\varrho$ is the $K$-dimensional column vector $(\rho_k)_{k=1,\ldots,K}$, $M$ and $R$ are the 
$K \times K$ matrices  defined by
$$
M_{ij} = -\theta_i (\delta_{ij}  -\rho_j) ,\quad \quad R_{ij} =\rho_i \delta_{ij},  \quad \quad i,j=1,\ldots,K,
$$
and ${}^\mathbf{T}$ stands for the transpose.
\end{proposition}
\begin{proof}
	See Appendix \ref{pf:xbar_diversity}.
\end{proof}
 
We can then deduce that the transition probability is
\begin{equation}
\label{eq:pTdiverse}
p_T \approx \exp \Big( - \frac{2 \xi_b^2}{\sigma_T^2}\Big)
\end{equation}

Our next goal is to study the impact of the diversity on the transition probability.

\begin{proposition}
\label{prop:the diversity on the transition probability}
Let us assume that the diversity is small:
$$
\theta_k =\bar{\theta} (1 +\delta \alpha_k),\quad \quad \delta \ll 1
$$
where $\sum_k \rho_k \alpha_k=0$ so that $\bar{\theta}$ is the mean value of the $\theta_k$'s.
The equilibrium position $\xi_b$, the variance $\sigma_T^2$ and the transition probability $p_T$ can be expanded as powers of $\delta$ as
\begin{align*}
	\xi_b^2 &= \Big(1- \frac{3 \sigma^2}{2 \bar{\theta}}\Big) 
	- \delta^2\Big(\sum_k \rho_k \alpha_k^2 \Big)   \frac{3 \sigma^2}{2 \bar{\theta}} 
	+O(\delta^3),\\
	\sigma_T^2 &= \frac{\sigma^2 T}{N} 
	\Big[ 1 + \delta^2  \Big( \sum_k \rho_k \alpha_k^2 \Big) 
	\Big( \frac{1}{T} \int_0^T (1-e^{-\bar{\theta} s})^2 ds \Big) +O(\delta^3) \Big], \\
	p_T &\approx \exp \Big\{ - \frac{2 N}{\sigma^2 T} 
	\Big[ \Big(1- \frac{3 \sigma^2}{2 \bar{\theta}}\Big)
	-\delta^2   \Big(\sum_k \rho_k \alpha_k^2 \Big)  
	\Big(\frac{3 \sigma^2}{2 \bar{\theta}}   
	+\frac{1}{T} \int_0^T (1-e^{-\bar{\theta} s})^2 ds \Big)  \Big]\Big\}.
\end{align*}
\end{proposition}
\begin{proof}
	See Appendix \ref{pf:the diversity on the transition probability}.
\end{proof}

This proposition shows that the diversity reduces the gap between the two equilibrium states and enhances the fluctuations of the empirical mean.
Both effects contribute to the increase of the systemic transition probability. 
\section{Summary and Conclusions}
\label{sec:summary}

The aim of this paper is to introduce and analyze a mathematical model
for the evolution of risk in a system of interacting agents 
where cooperation between them can reduce
their individual risk of failure but increase the systemic or overall risk.
The model we use is a system of bistable diffusion processes that interact through
their empirical mean, a mean field model. 
We take the rate of mean reversion to the empirical mean
$\theta$ as a measure of cooperation, the depth of the bistable potential $h$
as a measure of intrinsic stability of each agent, and the strength 
of the external random perturbations $\sigma$ as the level of uncertainty in
which the agents function. Using the theory of large
deviations we calculate the probability that the system will transition
from one of the two bistable states to the other during a time interval of length
$T$, when the number of agents $N$ is large and when $h$ is small. In this
regime of parameters we find that systemic risk increases with cooperation.
The formula from which we draw this conclusion is given is Section \ref{sec:comparison}.
We also show that when the rate of mean reversion to the empirical mean
varies among the different agents,
that is, when there is diversity in the cooperative behavior then the probability
of transitions increases, which means that the systemic risk increases.

\section*{Acknowledgement}

This work is partly supported by the Department of Energy 
[National Nuclear Security Administration] under Award Number NA28614, and
partly by AFOSR grant FA9550-11-1-0266.

\appendix
\section{Proof of Proposition 
\ref{prop:explicit condition and solution for xi equal to m(xi) for small h}}
\label{pf:explicit condition and solution for xi equal to m(xi) for small h}

For small $h$, we view $u^e_{\xi}$ as a perturbed Gaussian density function. 
Let $p_\xi(y)$ be the Gaussian density function with mean $\xi$ and variance 
$\sigma^2/2\theta$, $Y$ be the Gaussian random variable with the density $p_\xi$,
and $\eta = 2/\sigma^2$. By using the expansion 
$\exp( -h \eta V) = 1 - h\eta V + h^2 \eta^2 V^2/2 + O(h^3)$, we have
\begin{align*}
	Z_\xi &= 1 - h\eta\mathbf{E}V(Y) + \frac{1}{2} h^2\eta^2\mathbf{E}V^2(Y) + O(h^3)\\
	Z_\xi^{-1} &= 1 + h\eta \mathbf{E}V(Y) - \frac{1}{2} h^2\eta^2 \mathbf{E}V^2(Y) 
	+ h^2\eta^2(\mathbf{E}V(Y))^2 + O(h^3).
\end{align*}
Then we calculate $m(\xi)$ as follows:
\begin{align*}
	m(\xi) 
	&= Z_\xi^{-1} \int y\left(1-h\eta V+\frac{1}{2}h^2\eta^2 V^2
	+ O(h^3)\right)p_\xi(y)dy\\
	&= Z_\xi^{-1} \left( \xi - h\eta \mathbf{E}[YV(Y)] 
	+ \frac{1}{2}h^2\eta^2 \mathbf{E}[YV^2(Y)] + O(h^3) \right)\\
	&= \xi + h\eta\{\xi\mathbf{E}V(Y) - \mathbf{E}[YV(Y)]\}
	+ h^2\eta^2 \{ -\frac{1}{2}\xi\mathbf{E}V^2(Y) + \xi(\mathbf{E}V(Y))^2\\
	&\quad - \mathbf{E}V(Y)\mathbf{E}[YV(Y)] + \frac{1}{2}\mathbf{E}[YV^2(Y)] \}+O(h^3)\\
	&= \xi - h\eta\frac{\sigma^2}{2\theta} \mathbf{E}V_y(Y)	
	+ h^2\eta^2 \frac{\sigma^2}{2\theta} 
	\{ \mathbf{E}[V(Y)V_y(Y)] - \mathbf{E}V(Y)\mathbf{E}V_y(Y)\} + O(h^3)\\
	&= \xi - h\eta\frac{\sigma^2}{2\theta} \mathbf{E}V_y(Y) 
	+ h^2\eta^2 \frac{\sigma^2}{2\theta} \mathbf{Cov}(V_y(Y), V(Y)) + O(h^3).
\end{align*}
The compatibility condition $\xi_b=m(\xi_b)$ gives
\begin{equation}
	\label{eq:compatibility condition for small h}
	\mathbf{E}V_y(Y) - h\eta\mathbf{Cov}(V_y(Y), V(Y)) + O(h^2) = 0.
\end{equation}
Assuming that $\xi_b=\xi_0 + h \xi_1 + O(h^2)$, the $O(1)$ terms in 
(\ref{eq:compatibility condition for small h}) give
\[
	\xi_0^3 + 3\frac{\sigma^2}{2\theta}\xi_0 -\xi_0 
	= \xi_0(\xi_0^2 + 3\frac{\sigma^2}{2\theta}-1) = 0.
\]
Then $\xi_0 = 0, \pm\sqrt{1-3\sigma^2/2\theta}$ if $3\sigma^2 < 2\theta$, or otherwise 
$\xi_0 = 0$. In order to obtain the nontrivial result, we suppose that 
$3\sigma^2 < 2\theta$ and $\xi_0$ takes $\pm\sqrt{1-3\sigma^2/2\theta}$ in the later 
calculations. Note that 
$\mathbf{E}V_y(Y) = \xi^3 + (3\sigma^2/2\theta -1)\xi = 2h \xi_0^2\xi_1 + O(h^2)$, and
\begin{align*}
	\mathbf{Cov}(V_y(Y), V(Y)) &= \mathbf{E}[V(Y)V_y(Y)] + O(h)
	= \mathbf{E}[(\frac{1}{4}Y^4 - \frac{1}{2}Y^2)(Y^3-Y)] + O(h)\\
	&= \mathbf{E}[\frac{1}{4}Y^7 - \frac{3}{4}Y^5 + \frac{1}{2}Y^3] + O(h).
\end{align*}
Along with the identity $\xi_0^2+3\sigma^2/2\theta=1$, we have
\begin{align*}
	\mathbf{E}Y^3 &= \xi_0 + O(h), \quad
	\mathbf{E}Y^5 = \left(1 + 4\frac{\sigma^2}{2\theta} 
	- 6\left(\frac{\sigma^2}{2\theta}\right)^2 \right) \xi_0 + O(h),\\
	\mathbf{E}Y^7 &= \left(1 + 12\frac{\sigma^2}{2\theta} 
	+ 6\left(\frac{\sigma^2}{2\theta}\right)^2
	- 48\left(\frac{\sigma^2}{2\theta}\right)^3\right) \xi_0 + O(h).
\end{align*}
Then $\mathbf{Cov}(V_y(Y), V(Y))=6(\sigma^2/2\theta)^2(1-2\sigma^2/2\theta)\xi_0+O(h)$. 
The $O(h)$ terms in (\ref{eq:compatibility condition for small h}) imply 
$\xi_1 = 3\eta(\sigma^2/2\theta)^2(1-2\sigma^2/2\theta)/\xi_0$.

\section{Proofs in Section \ref{sec:diversity}}

\subsection{Proof of Theorem \ref{thm:mean-field limit, diversity case}}
\label{pf:mean-field limit, diversity case}

The proof contains three steps.

\subsubsection{Existence and Uniqueness of the Weak Solution of the McKean-Vlasov Equation}

The existence and uniqueness of a probability measure valued process $(u_1(t),\ldots,u_K(t))$ that is a weak 
solution of the McKean-Vlasov equation (\ref{eq:system of Fokker-Planck equations}) is guaranteed by 
\cite[Theorem 2.11]{Gartner1988}.

\subsubsection{Weak Compactness of the Empirical Process}

By Prohorov's theorem, it suffices to prove that the sequence $\{(X_N^1, \ldots, X_N^K)\}_{N=1}^\infty$ is 
weakly compact by showing that
\begin{equation*}
	\label{eq:tightness}
	\sup_N\sup_{1\leq k\leq K}\sup_{0\leq t\leq T} \mathbf{E}[\langle X_N^k(t,dy),|y|\rangle] <\infty,
\end{equation*}
which can be done by using the calculations similar to (B1) and (B2) in \cite{Dawson1983}.

\subsubsection{Identification of the Limit}

For a test function $f\in\mathcal{S}(\mathbb{R})$, we define $X_N^{f,l}(t)
=\langle f(y),X_N^l(t,y)\rangle=\sum_{j\in\mathcal{I}_l}f(x_j(t))/|\mathcal{I}_l|$. 
By It\^{o}'s formula,
\begin{align*}
	dX_N^{f,l}
	&= \frac{1}{|\mathcal{I}_l|} \sum_{j\in\mathcal{I}_l} 
	[ -h U(x_j)dt + \sigma dw_j + \Theta_l(\bar{x}-x_j)dt ] f_y(x_j) 
	+ \frac{1}{2}\sigma^2f_{yy}(x_j)dt\\
	&= \langle -h U f_y 
	+ \Theta_l(\langle y,\sum_{l=1}^K\rho_l X_N^l\rangle - y)f_y 
	+ \frac{\sigma^2}{2} f_{yy}, X_N^l \rangle dt
	+ \langle  f_y, 
	\frac{\sigma}{|\mathcal{I}_l|}\sum_{j\in\mathcal{I}_l}\delta_{x_j}dw_j \rangle.
\end{align*}
Then by the integration by parts, we write 
\[
	dX_N^l
	= \{(h U X_N^l)_y 
	-[\Theta_l(\langle y,\sum_{l=1}^K\rho_l X_N^l\rangle-y)X_N^l]_y 
	+ \frac{\sigma^2}{2}(X_N^l)_{yy} \}dt
	- \frac{\sigma}{|\mathcal{I}_l|}\sum_{j\in\mathcal{I}_l}(\delta_{x_j})_y dw_j.
\]
For simplicity, we prove the case that $K=2$ and the general case is similar. 
We let $X_N^{1,\times n} \times X_N^{2,\times n}$ denote the product measure on 
$\mathbb{R}^{2n}$:
\[
	X_N^{1,\times n} \times X_N^{2,\times n}(y_1,\ldots,y_{2n})
	= X_N^1(t,y_1)\cdots X_N^1(t,y_n)X_N^2(t,y_{n+1})\cdots X_N^2(t,y_{2n}).
\]
For a test function $f\in\mathcal{S}(\mathbb{R}^{2n})$, we have 
\[
	d\langle f, X_N^{1,\times n} \times X_N^{2,\times n}\rangle 
	= d\langle f, X_N^{1,\times n} \times X_N^{2,\times n}\rangle^{(1)} 
	+ d\langle f, X_N^{1,\times n} \times X_N^{2,\times n}\rangle^{(2)},
\]
where $(1)$ and $(2)$ denote the first and the second order terms of 
$d\langle f, X_N^{1,\times n} \times X_N^{2,\times n}\rangle$, respectively:
\begin{align*}
	d\langle f, X_N^{1,\times n} \times X_N^{2,\times n}\rangle^{(1)}
	&= \sum_{j=1}^n \langle f,dX_N^1(t,y_j)\times X_N^{1,\times(n-1),j}
	\times X_N^{2,\times n}\rangle\\
	&\quad + \sum_{j=n+1}^{2n} \langle f,dX_N^2(t,y_j)\times X_N^{1,\times n,j}
	\times X_N^{2,\times(n-1),j}\rangle
\end{align*}
\begin{multline*}
	d\langle f,X_N^{1,\times n}\times X_N^{2,\times n}\rangle^{(2)}
	= \frac{1}{2}\sum_{\substack{j,k=1\\j\neq k}}^n 
	\langle f, dX_N^1(t,y_j) \times dX_N^1(t,y_k)\times X_N^{1,\times(n-2),j,k} 
	\times X_N^{2,\times n}\rangle \\
	+ \frac{1}{2}\sum_{\substack{j,k=n+1\\j\neq k}}^{2n}
	\langle f,dX_N^2(t,y_j) \times dX_N^2(t,y_k) \times X_N^{1,\times n}
	\times X_N^{2,\times(n-2),j,k}\rangle \\
	+ \frac{1}{2}\sum_{j=1}^n\sum_{k=n+1}^{2n}
	\langle f, dX_N^1(t,y_j) \times dX_N^2(t,y_k) \times X_N^{1,\times(n-1),j}
	\times X_N^{2,\times(n-1),k}\rangle.
\end{multline*}
Note that for $j\neq k$, $dX_N^l(t,y_j) \times dX_N^l(t,y_k) 
= \frac{\sigma^2}{|\mathcal{I}_l|^2} \sum_{i\in\mathcal{I}_l} 
(\delta_{x_i}(y_j))_j (\delta_{x_i}(y_k))_k dt
= \frac{\sigma^2}{\rho_l^2 N} (\delta(y_k-y_j) X_N^l(t,y_j))_{jk} dt$, and 
$dX_N^1(t,y_j) \times dX_N^2(t,y_k) = 0 $.
If we analogously represent the generator $G_{(X_N^{1,\times n},X_N^{2,\times n})} f$ of 
$\langle f, X_N^{1,\times n} \times X_N^{2,\times n}\rangle$ as 
\[
	G_{(X_N^{1,\times n},X_N^{2,\times n})} f
	= G_{(X_N^{1,\times n},X_N^{2,\times n})}^{(1)} f 
	+ G_{(X_N^{1,\times n},X_N^{2,\times n})}^{(2)} f,
\]
then $G_{(X_N^{1,\times n},X_N^{2,\times n})}^{(2)} f \rightarrow 0 $ as 
$N\rightarrow \infty$ and $G_{(X_N^{1,\times n},X_N^{2,\times n})}^{(1)} f 
= G_{(u_1^{\times n},u_2^{\times n})}f$, the generator of 
$\langle f, u_1^{\times n} \times u_2^{\times n}\rangle$, where $(u_1,u_2)$ satisfying 
(\ref{eq:system of Fokker-Planck equations}). Then the limit of $(X_N^1,X_N^2)$ is a solution of the 
martingale problem associated to (\ref{eq:system of Fokker-Planck equations}). In addition, by 
\cite[Corollary 2.10]{Gartner1988}, the solution is unique and therefore $(X_N^1,X_N^2)\rightarrow 
(u_1,u_2)$ weakly as $N\rightarrow\infty$.


\subsection{Proof of Proposition 
\ref{prop:explicit condition for xi equal to m(xi) for small h, diversity case}}
\label{pf:explicit condition for xi equal to m(xi) for small h, diversity case}

All we need to show is that for small $h$, $\frac{d}{d\xi}m(0)>1$ if and only if 
$\sigma < \sigma_c^{\text{div}}$, where $m(\xi)$ is defined by 
(\ref{eq:xi equal to m(xi), diversity case}). We obtain $\frac{d}{d\xi}m$ by 
calculate $\frac{d}{d\xi}\int yu_{l,\xi}^e (y)dy$. Note that 
$\frac{d}{d\xi}Z_{l,\xi}=(2\Theta_l/\sigma^2)(\int yu^e_{l,\xi}dy-\xi) Z_{l,\xi}$ and 
\begin{align}
	\label{eq:d2Z, first kind}
	\frac{d^2}{d\xi^2}Z_{l,\xi} 
	&= \frac{2\Theta_l}{\sigma^2}Z_{l,\xi}
	\left(\frac{d}{d\xi}\int yu^e_{l,\xi}dy - 1\right)
	+ \frac{2\Theta_l}{\sigma^2} \left(\int yu^e_{l,\xi}dy-\xi\right)
	\frac{d}{d\xi} Z_{l,\xi}\\
	&= \frac{2\Theta_l}{\sigma^2}Z_{l,\xi} 
	\left(\frac{d}{d\xi}\int yu^e_{l,\xi}dy - 1\right)
	+ \left(\frac{2\Theta_l}{\sigma^2}\right)^2
	Z_{l,\xi} \left(\int yu^e_{l,\xi}dy - \xi\right)^2.\notag
\end{align}
On the other hand, we can also compute $\frac{d^2}{d\xi^2}Z_{l,\xi}$ by directly taking 
the twice derivatives of $Z_{l,\xi}$:
\begin{equation}
	\label{eq:d2Z, second kind}
	\frac{d^2}{d\xi^2}Z_{l,\xi} = -\frac{2\Theta_l}{\sigma^2} Z_{l,\xi} 
	+ \left(\frac{2\Theta_l}{\sigma^2}\right)^2 Z_{l,\xi} \int (y-\xi)^2 u^e_{l,\xi}dy.
\end{equation}
By comparing (\ref{eq:d2Z, first kind}) and (\ref{eq:d2Z, second kind}), 
\[
	\frac{d}{d\xi}\int yu^e_{l,\xi}dy
	= \frac{2\Theta_l}{\sigma^2} \left[ \int y^2u^e_{l,\xi}dy
	- (\int yu^e_{l,\xi}dy)^2\right].
\]
Note that $\int yu^e_{l,0}dy=0$, so 
$\frac{d}{d\xi}m(0)=\sum_{l=1}^K\rho_l(2\Theta_l/\sigma^2)\int y^2u^e_{l,0}dy$. 
By using the same trick in the proof of Proposition 
\ref{prop:explicit condition and solution for xi equal to m(xi) for small h}, 
let $p_l(y)$ be the Gaussian density function with mean $0$ and variance 
$\sigma^2/2\Theta_l$, $Y_l$ be the Gaussian random variable with the density $p_l$,
and $\eta = 2/\sigma^2$. Then for small $h$, 
$Z_{l,0}^{-1}=1+h\eta\mathbf{E}V(Y_l)+O(h^2)$, and 

\begin{align*}
	\int y^2u^e_{l,0}dy
	&= Z_{l,0}^{-1} \int y^2 (1-h\eta V+O(h^2)) p_l(y) dy\\
	&= Z_{l,0}^{-1} ( \mathbf{E}Y_l^2 - h\eta \mathbf{E}[Y_l^2V(Y_l)] + O(h^2) )\\
	&= \mathbf{E}Y_l^2 + h\eta (\mathbf{E}Y_l^2\mathbf{E}V(Y_l)
	-\mathbf{E}[Y_l^2V(Y_l)]) + O(h^2).
\end{align*}
Therefore $\frac{d}{d\xi}m(0)>1$ if and only if 
$\sum_{l=1}^K \rho_l (2\Theta_l/\sigma^2)
(\mathbf{E}Y_l^2\mathbf{E}V(Y_l)-\mathbf{E}[Y_l^2V(Y_l)]) > 0$. Note that 
$\mathbf{E}Y_l^2 = \sigma^2/2\Theta_l$, 
$\mathbf{E}V(Y_l) = (3/4)(\mathbf{E}Y_l^2)^2 - (1/2)\mathbf{E}Y_l^2$, and
$\mathbf{E}[Y_l^2V(Y_l)] = (15/4)(\mathbf{E}Y_l^2)^3 - (3/2)(\mathbf{E}Y_l^2)^2$.
Then the sufficient and necessary condition becomes 
\[
	\sum_{l=1}^K \frac{\rho_l}{\Theta_l}\left(1 - 3\frac{\sigma^2}{2\Theta_l}\right) > 0.
\]

\subsection{Proof of Proposition \ref{prop:sigma_c comparison}}
\label{pf:sigma_c comparison}

It is equivalent to show that $\sum_{l=1}^K \rho_l/\Theta_l 
\leq \sum_{l=1}^K \rho_l\Theta_l \sum_{l=1}^K \rho_l/\Theta_l^2$.
First note that by the Cauchy-Schwarz inequality,
\[
	\left(\sum_{l=1}^K \frac{\rho_l}{\Theta_l}\right)^2
	= \left(\sum_{l=1}^K\frac{\sqrt{\rho_l}}{\Theta_l}\times\sqrt{\rho_l}\right)^2
	\leq \sum_{l=1}^K\frac{\rho_l}{\Theta_l^2} \sum_{l=1}^K\rho_l 
	= \sum_{l=1}^K\frac{\rho_l}{\Theta_l^2}.
\]
Then it suffices to show that $1\leq\sum_{l=1}^K\rho_l\Theta_l
\sum_{l=1}^K \rho_l/\Theta_l$. Again by the Cauchy-Schwarz inequality,
\[
	\sum_{l=1}^K \rho_l\Theta_l \sum_{l=1}^K \frac{\rho_l}{\Theta_l} 
	\geq \sum_{l=1}^K \sqrt{\rho_l\Theta_l}\sqrt{\frac{\rho_l}{\Theta_l}} 
	= \sum_{l=1}^K\rho_l = 1.
\]

\section{Proof of Lemma \ref{lma:lower bound of I_h over A_delta}}
\label{pf:lower bound of I_h over A_delta}


It suffices to show the case that $\delta=1/n$. For each $n$, let $\phi_n\in A_{1/n}$, 
such that 
$\inf_{\phi\in A_{1/n}}I_h(\phi)\leq I_h(\phi_n)<\inf_{\phi\in A_{1/n}}I_h(\phi)+1/n$; 
$\{I_h(\phi_n)\}$ are bounded from above by $\inf_{\phi\in A}I_h(\phi)+1<\infty$. 
Because $I_h$ is a good rate function, and by Proposition B.13 of \cite{Gartner1988},
compactness is equivalent to sequentially compactness in 
$C([0,T],M_\infty(\mathbb{R}))$, $\{\phi_n\}$ has a convergent subsequence 
$\{\phi_{n_k}\}$ whose limit $\phi^*$ is in $A$. As $I_h$ is lower semicontinuous, then
\[
	\lim_n\inf_{\phi\in A_{1/n}}I_h(\phi) = \lim_k I_h(\phi_{n_k})
	= \liminf_k I_h(\phi_{n_k}) \geq I_h(\phi^*) \geq \inf_{\phi\in A}I_h(\phi).
\]

\section{Proofs in Section \ref{sec:small h}}

\subsection{Proof of Theorem \ref{thm:LD for zero h and Gaussian path}}
\label{pf:LD for zero h and Gaussian path}

We prove it in three steps. The first step is to show that there exists a uniform lower 
bound of $I_0(\phi)$, for all $\phi \in A$.
\begin{lemma}
	If $h=0$, then $\inf_{\phi \in A} I_0(\phi) \geq 2\xi_0^2/(\sigma^2 T)$.
\end{lemma}
\begin{proof}
	For any $\phi \in A$, $a(t)$ denotes $\int y \phi(t,dy)$. We observe that
	\[
		J_h(\phi)
		= \sup_{f: \langle \phi, f_y^2 \rangle \neq 0}
		\langle \phi_t - \mathcal{L}_\phi^*\phi, f \rangle^2
		/ {\langle \phi, f_y^2 \rangle}
		\overset{f \equiv y}{\geq} 
		\langle \phi_t - \mathcal{L}_\phi^*\phi, y\rangle^2,
	\]
	because $\langle \phi , 1\rangle=1$. Note that 
	$\langle\phi_t, y\rangle=\frac{d}{dt}\langle\phi, y\rangle=\frac{d}{dt}a(t)$, and
	\begin{equation*}
		\langle \mathcal{L}_\phi^*\phi, y \rangle 
		= \langle \frac{1}{2}\sigma^2 \phi_{yy}
		+ \theta\frac{\partial}{\partial y} \left[ (y-a(t)) \phi\right], y\rangle
		= -\theta \langle (y-a(t))\phi, 1 \rangle = 0.
	\end{equation*}
	Then after taking the infimum over $\phi \in A$, we have 
	\[
		\inf_{\phi \in A} I_0(\phi) \geq 
		\inf_{\phi \in A} \frac{1}{2\sigma^2}\int_0^T\left(\frac{d}{dt}a\right)^2 dt
		= \inf_{\substack{a(t):0 \leq t \leq T\\ a(0)=-\xi_0\\ a(T)=\xi_0}}
		\frac{1}{2\sigma^2}\int_0^T\left(\frac{d}{dt}a\right)^2 dt
		= \frac{2\xi_0^2}{\sigma^2 T}.
	\]
	The last equality is obtained by a simple calculus of variation with the optimal 
	path $a(t)= 2\xi_0 t/T - \xi_0$.
\end{proof}

The second step is to show that $I_0(p^e) = 2\xi_0^2/(\sigma^2 T)$. 
Then $\inf_{\phi \in A} I_0(\phi) =  2\xi_0^2/(\sigma^2 T)$ and 
therefore $p^e$ is a minimizer for (\ref{eq:general LD}).
\begin{lemma}
	If $h=0$, and 
	\[
		p^e(t,y) = \frac{1}{\sqrt{2\pi \frac{\sigma^2}{2\theta}}}
		\exp\left\{ -\frac{(y - a^e(t))^2}{2\frac{\sigma^2}{2\theta}} \right\}, \quad
		a^e(t)= \frac{2\xi_0}{T} t - \xi_0,
	\]
	then $p^e \in A$ and $I_0(p^e) =  2\xi_0^2/(\sigma^2 T)$.
\end{lemma}
\begin{proof}
	By reading (\ref{eq:Fokker-Planck equation with the driving force}) with $\phi=p^e$ 
	and $h=0$, we have $p^e_t = \mathcal{L}_{p^e}^{*}p^e + (p^e g)_y$.
	One can easily check that $\mathcal{L}_{p^e}^{*}p^e = 0$ and 
	$p^e_t = -p^e_y\frac{d}{dt}a^e(t)$. Then we have $g=-\frac{d}{dt}a^e(t)$ and 
	by (\ref{eq:I_h in terms of the driving force}),
	\[
		I_0(p^e) = \frac{1}{2\sigma^2} \int_0^T \langle p^e, g^2\rangle dt
		= \frac{1}{2\sigma^2} \int_0^T \left(\frac{d}{dt}a^e\right)^2 dt 
		= \frac{2\xi_0^2}{\sigma^2 T}.
	\]
\end{proof}

Finally we prove that for $h=0$, the minimizer $p^e$ is unique. 
\begin{lemma}
	For $h=0$, $p^e$ is the unique minimizer for (\ref{eq:general LD}).
\end{lemma}
\begin{proof}
	From the previous lemmas, we find that if $\phi$ is a minimizer then 
	$a(t)=\int y \phi(t,dy)$ must be $a^e(t)$, and $f = -\frac{d}{dt}a^e(t) y$ 
	is a global maximizer of $J_0(\phi,\cdot)$. Then for any test function $\tilde{f}$, 
	$\frac{d}{d\epsilon}J_0(\phi,-\frac{d}{dt}a^e(t) y + \epsilon \tilde{f})=0$ at 
	$\epsilon=0$. By a simple calculus of variations, $\phi$ satisfies the linear 
	parabolic PDE:
	\[
		\phi_t = \frac{1}{2}\sigma^2 \phi_{yy} 
		+ \theta\frac{\partial}{\partial y}\left[ (y-a^e(t))\phi\right]
		- \frac{d}{dt}a^e(t) \phi_y,	
	\]
	with the initial condition $\phi(0)=u_{-\xi_0}^e$, and that implies 
	the uniqueness of the minimizer, which is $p^e$.
\end{proof}

\subsection{Proof of Theorem \ref{thm:cotinuity for small h} (Upper Bounds)}
\label{pf:cotinuity for small h, upper bounds}

Define the test function:
\[
	p^u(t,y) = \frac{1}{\sqrt{2\pi \frac{\sigma^2}{2\theta}}}
	\exp\left\{ -\frac{(y - a^u(t))^2}{2\frac{\sigma^2}{2\theta}} \right\},
	\quad a^u(t) = \frac{2\xi_b}{T}t - \xi_b.
\]
We recall that from (\ref{eq:xi equals m(xi)}) and 
(\ref{eq:explicit solution for xi equal to m(xi)}), $\xi_b$ depends on $h$ and 
$\xi_b\rightarrow \xi_0$ as $h\rightarrow 0$.

\begin{proposition}
	\label{prop:upper bound for I}
	For any $\epsilon>0$, then for all sufficiently small $h$, 
	\begin{equation}
	\label{eq:upper bound for I}
	\inf_{\phi\in A}I_h(\phi) \leq \frac{1}{2\sigma^2} \int_0^T 
	\langle p^u, ( \frac{d}{dt}a^u - h(y^3-y) )^2 \rangle dt + \epsilon.
	\end{equation}
\end{proposition}
It is not difficult to see that the first term of the right hand side of (\ref{eq:upper bound for I}) 
 is equal to  $2\xi_b^2/(\sigma^2 T)$ up to a term of order $h$ as $h\rightarrow 0$.

\begin{proof}
	We construct the test function $\phi^u\in A$ as follows:
	\[
		\phi^u(t)=
		\begin{cases}
			(1-\frac{t}{\delta T})u^e_{-\xi_b} + \frac{t}{\delta T} p^u(t), 
			& t\in[0,\delta T],\\
			p^u(t), & t\in(\delta T, T-\delta T),\\
			(1-\frac{t-(T-\delta T)}{\delta T})p^u(t) 
			+ \frac{t-(T-\delta T)}{\delta T} u^e_{\xi_b}, & t\in[T-\delta T, T],
		\end{cases}
	\]
	where $\delta T$ will be determined later. Note that 
	$\inf_{\phi\in A}I_h(\phi) \leq I_h(\phi^u)$ so we just need to compute 
	$I_h(\phi^u)$. Let $g^u$ satisfy 
	(\ref{eq:Fokker-Planck equation with the driving force}) for $\phi=\phi^u$. For 
	$t\in(\delta T, T-\delta T)$, $\phi^u(t)=p^u(t)$, and it is easy to see that 
	$p^u_t=-\frac{d}{dt}a^u p^u_y$ and $\mathcal{L}_{p^u}^{*}p^u=0$. Therefore for 
	$t\in(\delta T, T-\delta T)$, $g^u = -\frac{d}{dt}a^u - h(y^3-y)$ by 
	(\ref{eq:Fokker-Planck equation with the driving force}). From 
	(\ref{eq:I_h in terms of the driving force}), we have 
	\begin{multline*}
		I_h(\phi^u) = \frac{1}{2\sigma^2} 
		\left(\int_0^{\delta T} + \int_{\delta T}^{T-\delta T} 
		+ \int_{T-\delta T}^T\right) \langle \phi^u, (g^u)^2\rangle dt\\
		\leq \frac{1}{2\sigma^2} \int_0^T 
		\langle p^u, ( -\frac{d}{dt}a^u - h(y^3-y) )^2 \rangle dt
		+ \frac{1}{2\sigma^2} \left(\int_0^{\delta T} + \int_{T-\delta T}^T\right) 
		\langle \phi^u, (g^u)^2\rangle dt.
	\end{multline*}
	The rest is to show that for any $\epsilon>0$, there exists a sufficiently small $h$ such that the last 
	term in the last equation is bounded by $\epsilon$. It suffices to show that for any $\delta T>0$, we 
	can choose a sufficiently small $h$ such that $\langle \phi^u, (g^u)^2\rangle$ is bounded by a 
	$\delta T$-independent constant $c^u>0$ for $t\in[0,\delta T]\cup[T-\delta T, T]$. If so, then let 
	$\delta T <\epsilon\sigma^2/c^u$ and
	\[
		\frac{1}{2\sigma^2} \left(\int_0^{\delta T} + \int_{T-\delta T}^T\right) 
		\langle \phi^u, (g^u)^2\rangle dt \leq \frac{1}{2\sigma^2}(2\delta T)c^u < \epsilon,
	\]
	for sufficiently small $h$.
		
	For $t\in[0,\delta T]$, because $\phi^u$ is simply the convex combination of $u^e_{-\xi_b}$ and $p^u$, 
	$\phi^u$ can be bounded by a $\delta T$-independent constant. To compute $g^u$ from 
	(\ref{eq:Fokker-Planck equation with the driving force}), it is also easy to see that 
	$\mathcal{L}_{\phi^u}^{*}\phi^u$ and $\mathcal{M}^{*}\phi^u$ can be bounded by $\delta T$-independent 
	constants. The only term we need to worry is  $(p^u(t)-u^e_{-\xi_b})/\delta T$ from computing 
	$\phi^u_t(t)$. However, $p^u(t)$ is differentiable at $t=0$ and $p^u(0)\rightarrow u^e_{-\xi_b}$ as 
	$h\rightarrow 0$ so we can bound $(p^u(t)-u^e_{-\xi_b})/\delta T$ by a $\delta T$-independent constant 
	with suitable $h$. Thus $g^u$ is bounded independently of $\delta T$ and so we can find a 
	$\delta T$-independent constant $c^u>0$ such that $\langle \phi^u, (g^u)^2\rangle<c^u$.
	
	The same argument works for $t\in[T-\delta T,T]$ and we have the desired result.
\end{proof}

\subsection{Proof of Theorem \ref{thm:cotinuity for small h} (Lower Bounds)}
\label{pf:cotinuity for small h, lower bounds}

From (\ref{eq:upper bound for I}), there exists some constant $C$ such that 
$\inf_{\phi\in A}I_h(\phi)\leq C$ for all $h\leq h_0$. Then we can assume that 
$I_h(\phi)\leq C$ for all $\phi\in A$ and all $h\leq h_0$ without loss of generality. 
The following lemma shows that the first and second moments of all $\phi\in A$ are 
uniformly bounded.
\begin{lemma}
	\label{lma:uniform bound for first and second moments}
	Given $C>0$, there exists $R>0$ such that for any $\phi\in A$ with $I_h(\phi)\leq C$ 
	for some $h\geq 0$, then
	\[
		\sup_{t\in[0,T]}\langle\phi(t),y\rangle^2 
		\leq \sup_{t\in[0,T]}\langle\phi(t),y^2\rangle 
		\leq R.
	\]
\end{lemma}
\begin{proof}
	Recall that $M_R(\mathbb{R})=\{\phi\in M_1(\mathbb{R}),\int\varphi(y)\phi(dy)\leq 
	R\}$ and $M_\infty(\mathbb{R})=\cup_{R>0}M_R(\mathbb{R})$ with the inductive 
	topology. Here we focus on the case that $\varphi=1+y^2$ in order to obtain the 
	uniform result, and let $M_R^2(\mathbb{R})$ and $M_\infty^2(\mathbb{R})$ denote the 
	spaces with the quadratic Lyapunov function $\varphi$.
	
	The proof is an application of Theorem 5.1(c), Theorem 5.3 and Lemma 5.5 of 
	\cite{Dawson1987}. By Theorem 5.1(c), if $\phi\in C([0,T],M_\infty^2(\mathbb{R}))$ 
	with $\phi(0)=u^e_{-\xi_b}$ and $I_h(\phi)\leq C$ for some $h\geq 0$, then $\phi$ is 
	in an $h$-dependent compact set $K$. By Theorem 5.3 the compact set $K$ is contained 
	in $C([0,T],M_R^2(\mathbb{R}))$ for an $h$-dependent $R>0$. Finally, by Lemma 5.5 
	and Theorem 5.1(c), it suffices to let $R\geq e^{\lambda T}(C+r)$, where $r$ and 
	$\lambda$ satisfy
	\[
		r \geq 2\int \varphi(y)u^e_{-\xi_b}(y)dy,\quad
		\lambda \geq \sup_{\mu\in M_1(\mathbb{R})}
		\langle\mu,\mathcal{L}_\mu\varphi+h\mathcal{M}\varphi
		+\frac{1}{2}\varphi_y^2\rangle
		/\langle\mu,\varphi\rangle,
	\]
	with $\varphi(y)=1+y^2$. Obviously we can find the uniform $r$ and $\lambda$ for all
	$h\geq 0$ and also the uniform $R$. Then any $\phi$ of interest are in 
	$C([0,T],M_R^2(\mathbb{R}))$ and thus have the uniform bounded first and second 
	order moments.
\end{proof}

Now we derive that lower bound. The key idea is that because we have the universal upper 
bound for the first and second moments of all $\phi\in A$ and for all $h\leq h_{0}$, 
Chebyshev's inequality implies the uniform convergence. 

\begin{proposition}
	\label{prop:lower bound for I}
	For any $\epsilon>0$, then for all sufficiently small $h$, 
	\begin{equation}
		\label{eq:lower bound for I}
		\inf_{\phi\in A}I_h(\phi) \geq \frac{1}{2\sigma^2} \int_0^T 
		\langle p^u, ( \frac{d}{dt}a^u - h(y^3-y) )^2 \rangle dt - \epsilon.
	\end{equation}
\end{proposition}
\begin{proof}
	Define $f^M=\iota*\hat{f}^M$, where $\hat{f}^M$ is a piecewise linear function and 
	$\iota$ is the standard mollifier: 
	\[
		\hat{f}^M(y)
		=\begin{cases}
			y,		& y\in(-M,M)\\
			-y+2M, 	& y\in[M,2M]\\
			-y-2M, 	& y\in[-2M,-M]\\
			0, 		& \text{otherwise}
		\end{cases}, \quad 
		\iota(y)
		=\begin{cases}
			Z\exp(\frac{1}{y^2-1}), & y^2 < 1\\
			0, 						& \text{otherwise.}
		\end{cases}
	\]
	Then $f^M$ is a smooth function with the compact support $[-2M-1,2M+1]$. In 
	addition, $f^M(y)\equiv y$ on $(-M+1,M-1)$, $|f_x^M|\leq 1$, and $|f_{xx}^M|$ is 
	uniformly bounded for all $M$ and is nonzero only on $\cup_{i=-2}^2 (iM-1,iM+1)$.

	Because for all $\phi\in A$, $\langle \phi(t),(f_y^M)^2\rangle \leq 1$,
	we can estimate the rate function:
	\begin{align*}
		I_h(\phi) &\geq \frac{1}{2\sigma^2}\int_0^T
		\langle \phi_t - \mathcal{L}_\phi^*\phi 
		- h\mathcal{M}^*\phi, f^M \rangle^2 dt\\
		&\geq \frac{1}{2\sigma^2 T}\left(\int_0^T\langle \phi_t 
		- \mathcal{L}_\phi^*\phi - h\mathcal{M}^*\phi, f^M \rangle dt\right)^2.
	\end{align*}
	Then we estimate the integrand term by term. By Lemma 
	\ref{lma:uniform bound for first and second moments}, the following convergences are 
	all uniform in $\phi\in A$ and $h\leq h_0$.

	First we have 
	\[
		\int_0^T \langle \phi_t, f^M \rangle dt
		= \langle u_{\xi_b}^e, f^M \rangle - \langle u_{-\xi_b}^e, f^M \rangle.
	\]
	$u_{\pm\xi_b}^e$ are exponentially decaying functions so 
	$\langle u_{\pm\xi_b}^e, f^M \rangle $ converges to $\pm\xi_{b}$ rapidly as 
	$M\rightarrow\infty$.

	We note that $\langle \mathcal{L}_\phi^* \phi, f^M\rangle 
	=\sigma^2\langle \phi,f_{yy}^M \rangle/2 - \theta\langle \phi,(y-a)f^M_y \rangle$. 
	By reading the properties of $f_{yy}^M$ and Chebyshev's inequality, we have 
	$\langle \phi,f_{yy}^M\rangle \rightarrow 0$ as $M\rightarrow\infty$. 
	We write $\langle \phi,(y-a)f^M_y\rangle$ as 
	\[
		\langle \phi,(y-a)f^M_y\rangle 
		= a(1-\langle \phi,f^M_y\rangle) + (\langle \phi, y f^M_y\rangle -a).
	\]
	Since $a$ is bounded and $\langle \phi,f^M_y \rangle \rightarrow 1$	as 
	$M\rightarrow\infty$, $a(1-\langle \phi,f^M_y \rangle)\rightarrow 0$ as 
	$M\rightarrow\infty$. We see that 
	\begin{align*}
		|\langle \phi, y f^M_y \rangle -a|^2
		& \leq \left( 2 \int_{(-M+1,M+1)^c} |y| \phi(dy) \right)^2\\
		& \leq 4 \int_{(-M+1,M+1)^c} y^2 \phi(dy) \int_{(-M+1,M+1)^c} \phi(dy).
	\end{align*}
	Again by Chebyshev's inequality, the right hand side vanishes as 
	$M\rightarrow\infty$.

	Finally we estimate $\langle \mathcal{M}^*\phi, f^M \rangle $. Since $f^M$ is 
	compactly supported,
	\[
		|\langle \mathcal{M}^* \phi, f^M \rangle |
		= |\langle \phi,(y^3-y)f^M_y \rangle | \leq (2M+1)^3 + (2M+1).
	\]
	For a fixed $M$, we can choose a sufficiently small $h$ such that 
	$h|\langle \mathcal{M}^* \phi, f^M\rangle |$ is small.

	Consequently, for any $\epsilon>0$, we can first choose a sufficiently large $M$ and 
	then there exists a sufficiently small $h$ such that
	\[
		\inf_{\phi\in A} I_h(\phi) \geq \frac{2\xi_b^2}{\sigma^2 T} - \epsilon.
	\]
\end{proof}

\subsection{Proof of Lemma \ref{lma:lower bound of I_h over B_delta}}
\label{pf:lower bound of I_h over B_delta}
It suffices to show the case that $\delta=1/n$. For each $n$, let $\phi_n\in B_{1/n}$ and $t_n\in(0,T]$ 
such that $\rho(\phi_n(t_n),u^e_{\xi_b})<\delta$ and 
$\inf_{\phi\in B_{1/n}}I_h(\phi)\leq I_h(\phi_n)<\inf_{\phi\in B_{1/n}}I_h(\phi)+1/n$; $\{I_h(\phi_n)\}$ are 
bounded from above by $\inf_{\phi\in B}I_h(\phi)+1<\infty$. Let $\{t_{n_k}\}$ be a convergent subsequence of 
$\{t_n\}$. Because $I_h$ is a good rate function, and by Proposition B.13 of \cite{Gartner1988}, compactness 
is equivalent to sequentially compactness in $C([0,T],M_\infty(\mathbb{R}))$, $\{\phi_{n_k}\}$ has a 
convergent subsequence $\{\phi_{n_{k'}}\}$ whose limit $\phi^*$ is in $A(t^*)$ where $t^*=\lim t_{n_k}$. As 
$I_h$ is lower semicontinuous, then
\[
	\lim_n\inf_{\phi\in B_{1/n}}I_h(\phi) = \liminf_{n_{k'}} I_h(\phi_{n_{k'}}) 
	\geq I_h(\phi^*) \geq \inf_{\phi\in A(t^*)}I_h(\phi) \geq \inf_{\phi\in B}I_h(\phi).
\]

\section{Proofs in Section \ref{sec:reduced LDP}}

\subsection{Proof of Lemma \ref{lma:the optimal g^0 and g^1}}
\label{pf:the optimal g^0 and g^1}

We note that $p_t = -p_y \frac{d}{dt}a$ and therefore 
\begin{align*}
	\phi_t &= -p_y \frac{d}{dt}a 
	+ h\sum_{n=2}^{\infty}\frac{d}{dt}b_n\frac{\partial^n}{\partial y^n}p
	- h\frac{d}{dt}a\sum_{n=2}^{\infty}b_n\frac{\partial^{n+1}}{\partial y^{n+1}}p\\
	&\quad + h^2\sum_{n=2}^{\infty}\frac{d}{dt}c_n\frac{\partial^n}{\partial y^n}p
	- h^2\frac{d}{dt}a\sum_{n=2}^{\infty}c_n\frac{\partial^{n+1}}{\partial y^{n+1}}p.
\end{align*}
After collecting $O(1)$ terms in 
(\ref{eq:Fokker-Planck equation with the driving force}) and integrating over $y$, 
we have 
\[
	-p \frac{d}{dt}a = \frac{1}{2}\sigma^2 p_y + \theta(y-a)p + p g^{(0)} = p g^{(0)}.
\]
Then $g^{(0)} = -\frac{d}{dt}a$.

Now we collect $O(h)$ terms in (\ref{eq:Fokker-Planck equation with the driving force}) 
and integrating over $y$. We get 
\begin{multline*}
	\sum_{n=1}^{\infty}\frac{d}{dt}b_{n+1}\frac{\partial^n}{\partial y^n}p
	- \frac{d}{dt}a\sum_{n=2}^{\infty}b_n\frac{\partial^n}{\partial y^n}p
	= \frac{1}{2}\sigma^2\sum_{n=2}^{\infty}b_n\frac{\partial^{n+1}}{\partial y^{n+1}}p\\
	+ \theta(y-a)\sum_{n=2}^{\infty}b_n\frac{\partial^n}{\partial y^n}p
	+ \sum_{n=0}^3 \delta_n\frac{\partial^n}{\partial y^n}p
	+ g^{(0)} \sum_{n=2}^{\infty}b_n\frac{\partial^n}{\partial y^n}p
	+ \sum_{n=0}^{\infty}\beta_n \frac{\partial^n}{\partial y^n}p.
\end{multline*}
Using the fact that 
\[
	\frac{1}{2}\sigma^2\frac{\partial^{n+1}}{\partial y^{n+1}}p
	= -\theta(y-a)\frac{\partial^n}{\partial y^n}p
	-n\theta \frac{\partial^{n-1}}{\partial y^{n-1}}p,
\]
we have 
\[
	\sum_{n=1}^{\infty}\frac{d}{dt}b_{n+1}\frac{\partial^n}{\partial y^n}p
	= -\theta\sum_{n=1}^{\infty}(n+1)b_{n+1}\frac{\partial^n}{\partial y^n}p\\
	+ \sum_{n=0}^3 \delta_n\frac{\partial^n}{\partial y^n}p
	+ \sum_{n=0}^{\infty}\beta_n \frac{\partial^n}{\partial y^n}p,
\]
and the optimal $\beta_n$ are obtained by comparing the coefficients.

\subsection{Proof of Lemma \ref{lma:the optimal g^2}}
\label{pf:the optimal g^2}

Let $\psi^{(2)}$ denote the anti-derivative of $q^{(2)}$ that vanishes at $-\infty$. After collecting 
$O(h^2)$ terms in (\ref{eq:Fokker-Planck equation with the driving force}) and integrating over $y$. We have 
\begin{equation}
	\label{eq:eqn for the optimal f^2}
	\psi^{(2)}_t = \frac{1}{2}\sigma^2 q^{(2)}_y + \theta(y-a)q^{(2)} + U(y)q^{(1)} + q^{(2)} g^{(0)}
	+ q^{(1)} g^{(1)} + p g^{(2)}.
\end{equation}
Note that $p g^{(2)} = \sum_{n=0}^\infty \gamma_n \frac{\partial^n}{\partial y^n}p$, 
so $\gamma_0$ is obtained by integrating (\ref{eq:eqn for the optimal f^2}) from 
$y=-\infty$ to $y=\infty$. Then we have 
$\gamma_0 = -\langle q^{(1)}, U(y)+g^{(1)} \rangle$.

\subsection{Proof of Proposition \ref{prop:inf_A I_h up to O(h)}}
\label{pf:inf_A I_h up to O(h)}

We write $a(t) = a_0(t) + h a_1(t) + O(h^2)$ with $a_0(t) = 2\xi_0 t/T - \xi_0$ 
and $a_1(t) = 2\xi_1 t/T - \xi_1$. Then we put $a(t)$ into 
(\ref{eq:variational problem of a}) and we have
\[
	\inf_{\phi\in A}I_h(\phi) 
	= \frac{1}{2\sigma^2} \int_0^T \left\{ (\frac{d}{dt}a_0)^2 
	+ 2h(\frac{d}{dt}a_0)(a_0^3+(3\frac{\sigma^2}{2\theta}-1)a_0 
	+ \frac{d}{dt}a_1)\right\} dt + O(h^2).
\]
We note that $\frac{d}{dt}a_0$ is a constant, and $a_0(t)$ and $a_0^3(t)$ 
are odd functions with respect to $t=T/2$. Then 
\begin{multline*}
	\inf_{\phi\in A}I_h(\phi) 
	= \frac{1}{2\sigma^2} \int_0^T \left\{ \left(\frac{d}{dt}a_0\right)^2 
	+ 2h\frac{d}{dt}a_0\frac{d}{dt}a_1 \right\} dt + O(h^2)\\
	= \frac{1}{2\sigma^2} \int_0^T \left\{ \left(\frac{2\xi_0}{T}\right)^2 
	+ 2h\frac{2\xi_0}{T}\frac{2\xi_1}{T} \right\} dt + O(h^2)
	= \frac{2\xi_0}{\sigma^2 T}(\xi_0 + 2h\xi_1) + O(h^2).
\end{multline*}
\section{Proofs in Section \ref{sec:diversity2}}

\subsection{Proof of Proposition \ref{prop:xbar_diversity}}
\label{pf:xbar_diversity}

The system of SDEs (\ref{eq:systembarxk}) for the vector $\bar{X}(t) =(\bar{x}_k(t))_{k=1,\ldots,K}$ has the 
form
\[
	d\bar{X} (t) =M\bar{X}(t) + \frac{\sigma}{\sqrt{N}} R^{-1/2} d\bar{W}(t)
\]
where $\bar{W}(t)= (\bar{w}_k(t))_{k=1,\ldots,K}$ is a column vector. This system can be solved:
\[
	\bar{X}(t) = e^{Mt} \bar{X}(0) +\frac{\sigma}{\sqrt{N}} \int_0^t e^{M(t-s)} R^{-1/2} d \bar{W}(s)
\]
If $\bar{x}_k(0)= -\xi_b$, then, using the fact that the uniform vector is in the null space of $M$, we 
have $e^{Mt}\bar{X}(0)=\bar{X}(0)$. As a corollary we get the explicit representation of the empirical mean:
\[
	\bar{x}(t) = -\xi_b+\frac{\sigma}{\sqrt{N}} \int_0^t \varrho^\mathbf{T} e^{M(t-s)} R^{-1/2} d\bar{W}(s)
\] 
This shows the desired result.

\subsection{Proof of Proposition \ref{prop:the diversity on the transition probability}}
\label{pf:the diversity on the transition probability}

The expansion of $\xi_b^2$ follows from the explicit expression 
(\ref{eq:explicit solution for xi equal to m(xi), diversity case}). The expansion of $\sigma_T^2$ follows 
from the expansion of (\ref{eq:sigmaT2}) and uses the properties of the matrix $M$. We have 
$M=-\bar{\theta}\bar{M}-\delta\bar{\theta}N$, with
\begin{align*}
	\bar{M} &= I-  u\varrho^\mathbf{T},\quad 
	\text{where $u=(1,\ldots,1)$ is the $K$-dimensional column vector,}\\
	N_{ij} &=\alpha_i (  \delta_{ij} - \rho_j ) ,\quad i,j=1,\ldots,K.
\end{align*}
The matrix $\bar{M}$ satisfies $\bar{M}^n =\bar{M}$ for all $n\geq 1$ and therefore 
\[
	e^{-\bar{\theta} \bar{M} t} 
	= \sum_{n=0}^\infty \frac{(-\bar{\theta} t)^n}{n!} \bar{M}^n 
	= I + \sum_{n=1}^\infty \frac{(-\bar{\theta} t)^n}{n!} \bar{M} 
	=  I + ( e^{-\bar{\theta} t }-1)\bar{M}.
\]
We have
\[
	e^{Mt} =\sum_{n=0}^\infty \frac{(-\bar{\theta} t)^n}{n!} ( \bar{M}+\delta N)^n
\]
Using the fact that $\bar{M}^\mathbf{T}\varrho=0$ (and again that $\bar{M}^n=\bar{M}$ for $n \geq 1$), we 
can expand
\begin{align*}
	\varrho^\mathbf{T}e^{M t} 
	&=\varrho^\mathbf{T} 
	+\delta\varrho^\mathbf{T}\big\{(-\bar{\theta}t)N+(e^{-\bar{\theta}t}-1+\bar{\theta}t)N\bar{M}\big\}\\
	&\quad + \delta^2 \varrho^\mathbf{T}  \Big\{ \frac{(\bar{\theta}t)^2}{2} N^2 
	+ \big[ e^{- \bar{\theta}t}-1+\bar{\theta}t -\frac{(\bar{\theta}t)^2}{2}  \big]
	\big[ N^2 \bar{M} -3 (N \bar{M} )^2 + N  \bar{M} N \big] \\
	&\quad -\bar{\theta} t \big[ e^{- \bar{\theta}t}-1+\bar{\theta}t \big] ( N \bar{M})^2 \varrho \Big\} 
	+ O(\delta^3).
\end{align*}
Using the fact that $\bar{M}^\mathbf{T} N^\mathbf{T} \varrho = N^\mathbf{T} \varrho$ and  
$\bar{M}^\mathbf{T} (N^\mathbf{T})^2 \varrho = (N^\mathbf{T})^2 \varrho$, this can be simplified into
\[
	\varrho^\mathbf{T} e^{M t} 
	=  \varrho^\mathbf{T} + \delta  \varrho^\mathbf{T}  (e^{-\bar{\theta} t} -1)  N    
	+ \delta^2 \varrho^\mathbf{T} \big[ (\bar{\theta} t)^2 
	- (1+\bar{\theta}t)(e^{-\bar{\theta} t} -1+\bar{\theta}t) \big]   N^2 +O(\delta^3).
\]
Consequently
\begin{align*}
	&\varrho^\mathbf{T} e^{Mt} R^{-1} (e^{Mt})^\mathbf{T} \varrho 
	= \varrho^\mathbf{T}  (I+(e^{-\bar{\theta}t}-1) \bar{M} )  R^{-1}  
	(I+(e^{-\bar{\theta}t}-1) \bar{M}^\mathbf{T} ) \varrho \\
	&\quad +2\delta\varrho^\mathbf{T}(e^{-\bar{\theta}t}-1)NR^{-1}
	(I+(e^{-\bar{\theta}t}-1) \bar{M}^\mathbf{T})\varrho\\
	&\quad  +2 \delta^2  \varrho^\mathbf{T} \big[ (\bar{\theta} t)^2 -(1+\bar{\theta}t)(e^{-\bar{\theta} t} 
	-1+\bar{\theta}t) \big]   N^2 R^{-1}  (I+(e^{-\bar{\theta}t}-1) \bar{M}^\mathbf{T} )  \varrho  \\
	&\quad 
	+\delta^2 \varrho^\mathbf{T}  ( e^{-\bar{\theta} t} -1 )   N R^{-1}  (e^{-\bar{\theta} t} -1 )   N^\mathbf{T} \varrho +O(\delta^3).
\end{align*}
Using the fact that $\bar{M}^\mathbf{T} \varrho=0$ and $NR^{-1} \varrho = Nu=0$, we obtain
\[
	\varrho^\mathbf{T} e^{Mt} R^{-1} (e^{Mt})^\mathbf{T} \varrho = 
	\varrho^\mathbf{T}   R^{-1}  \varrho +\delta^2 (1-e^{-\bar{\theta} t} )^2 \varrho^\mathbf{T}    
	N R^{-1} N^\mathbf{T} \varrho  +O(\delta^3)
\]
We have $\varrho^\mathbf{T}   R^{-1}  \varrho=1$ and 
$\varrho^\mathbf{T} N R^{-1}N^\mathbf{T} \varrho = \sum_k \rho_k \alpha_k^2 $ which gives the expansion of 
the variance $\sigma_T^2$.

Finally the expansion of the transition probability can be obtained by substituting the expansions of 
$\xi_b^2$ and $\sigma_T^2$ into (\ref{eq:pTdiverse}).

\bibliographystyle{siam}
\bibliography{reference}

\end{document}